\documentclass[11pt, a4paper]{article}

\usepackage{setup_long, jmp} 

\hypersetup{pdfauthor={Manu Navjeevan},
            pdftitle={An Identification and Dimensionality Robust Test IV Models}}

\title{\vspace{-2em} An Identification-and Dimensionality-Robust Test for Instrumental Variables Models}%Title
\author{\sc Manu Navjeevan% 
    \thanks{Email: mnavjeevan@ucla.edu. Revised \today. The latest version of this paper can be found \href{https://navjeevan.dev/files/Papers/WeakDimRobust.pdf}{here}. I thank Denis Chetverikov for his guidance, numerous discussions and permanent support. I am also grateful to the other members of my dissertation committee, Andres Santos and Zhipeng Liao, and participants in UCLA's econometrics proseminar, Jinyong Hahn, Rosa Matzkin, Shuyang Sheng, and Daniel Ober-Reynolds for a half-decade of useful comments.}}
\shorttitle{An Identification and Dimensionality Robust Test}
\affil{\vspace{-0.9em}Texas A\&M University}
\date{\vspace{-1em}}

\begin{document}
% Title Page and Abstract
\maketitle 

\thispagestyle{empty}
\vspace{-2em}
\begin{abstract}
    \noindent
    \textsc{Abstract.} 
    Using modifications of Lindeberg's interpolation technique, I propose a new identification-robust test for the structural parameter in a heteroskedastic instrumental variables model. While my analysis allows the number of instruments to be much larger than the sample size, it does not require many instruments, making my test applicable in settings that have not been well studied. Instead, the proposed test statistic has a limiting chi-squared distribution so long as an auxiliary parameter can be consistently estimated. This is possible using machine learning methods even when the number of instruments is much larger than the sample size. To improve power, a simple combination with the sup-score statistic of \citet{BCCH-2012} is proposed. I point out that first-stage F-statistics calculated on LASSO selected variables may be misleading indicators of identification strength and demonstrate favorable performance of my proposed methods in both empirical data and simulation study. 

    \vspace{1em}
    
    \textsc{Keywords:} Instrumental Variables, Weak Identification, High-Dimensional 
    
    \textsc{JEL Codes:} C12, C36, C55
\end{abstract}

% \vspace{1em}

% \begin{center}
%
% \end{center}

 % Main Sections
% \newblankpage
\section{Introduction}
\label{sec:intro}

When instruments are suspected to be weak, researchers may want to test hypotheses about structural parameters using testing procedures that are robust to identification strength. These procedures all rely on some conditions on the rate of growth of the number of instruments, \(d_z\), in relation to the sample size, \(n\). The initial identification robust tests developed in \citet{StaigerStock-1997}, \citet{Moreira_2003}, and \citet{Kleibergen-2005} are shown by \citet{Andrews-Stock-2007} to control size under heteroskedasticity when the number of instruments cubes is small relative to the sample size, \(d_z^3/n\to 0\). Meanwhile, recent and interesting ``many-instrument'' tests (\citet{crudu_mellace_sandor_2021}, \citet{Mikusheva_Sun_2022}, \citet{Matsushita_Otsu_2022}, \citet{lim2022conditional}) allow the number of instruments to be proportional to the sample size, \(d_z/n\to \varrho \in [0,1)\), but require that the number of instruments itself be large, \(d_z \to\infty\).

In practice, these conditions can be difficult to interpret and the variety of tests available under alternate regimes may make it difficult for the researcher to know which test, if any, should be applied in her exact setting. As examples, consider settings such as that of \citet{derenoncourt-2022}, where \(d_z = 9\) and \(n = 130\), \citet{Paravisini-2014}, where \(d_z = 10\) and \(n = 5{,}995\), and \citet{gilchrist-glassberg-2016} where \(d_z = 52\) and \(n = 1{,}671\). In all three cases, the number of instruments cubed, \(d_z^3 = 729\), \(1{,}000\) and \(140{,}608\), respectively, cannot be treated as negligible relative to the sample size. Indeed, all of these papers use post-LASSO estimates of the first-stage, suggesting concern about the large number of instruments by the authors. However as the number of instruments is only moderately large in all of these settings, asymptotic approximations that rely on \(d_z \to \infty\) may not accurately resemble the finite sample distribution of the test statistic. This can make size control of many-instrument tests questionable.\footnote{\citet{Mikusheva_Sun_2022} point out that, when errors are homoskedastic, under fixed \(d_z\) their proposed test statistic weakly converges to a distribution whose 95th quantile is close to that of the standard normal. However, it is unclear what the behavior of the test would be when errors are heteroskedastic.}

By comparison, the test proposed in this paper can be applied in any of the settings listed above as well as when the instruments are high-dimensional, \(d_z \gg n\), a setting for which there has been little progress on identification robust testing to this point. The main problem in these settings has been that the exact limiting behaviors of the regularized first-stage estimators used when \(d_z \gg n\) are difficult to analyze and typically unknown. Existing analyses sidestep this issue by assuming strong identification and exploiting a certain orthogonality, termed ``Neyman Orthogonality'' by \citet{CCDDHNR-2018}, of the structural parameter estimate to the first-stage estimation error \citep{BCCH-2012}. This approach is explicitly not applicable under weak identification where the first-stage estimation error is on the same or lesser order as the signal from the instruments and thus relevant to asymptotic analysis. As such, the few existing proposals for identification robust
testing that allow \(d_z \gg n\) (\citet{BCCH-2012}, \citet{gautier2021highdimensional}, \citet{Mikusheva-2023-ManyTimeSeries}) either fail to incorporate first-stage information or rely on sample splitting, both of which may reduce power.

To construct the test statistic I borrow an idea from \citet{Kleibergen_2002,Kleibergen-2005} and leverage a conditional slope parameter, which can be simply estimated using regularized methods even when \(d_z \gg n\), to partial out the structural error from the endogenous variable. I then use this partialled out version of the endogenous variable to construct first-stage estimates. 
The key idea is that, if variables in the model followed a jointly Gaussian distribution, we could exploit the fact that uncorrelated jointly Gaussian random variables are independent to show that these first-stage estimates are independent of the structural error under the null. Thus, even under weak identification where their behavior becomes relevant to the distribution of the proposed test statistic, one could easily derive the limiting \(\chi^2\) distribution of the proposed test statistic by conditioning on these estimates. Deriving the limiting distribution of the test statistic in the general model then reduces to showing that one can treat all observations as if they followed a jointly Gaussian distribution.
% This conditional slope parameter can be consistently estimated using standard methods when \(d_z \ll n\) or using machine learning tools when \(d_z \gg n\). In practice, I propose an \(\ell_1\)-penalized estimate of the conditional slope parameter which can work in both low- and high-dimensional settings and is easily implementable using out-of-the-box methods. This estimator has the additional benefit of being trivially consistent when both first- and second-stage errors are homoskedastic.  

% \footnote{Interestingly, the asymptotic analysis does not require the first-stage jackknife-ridge estimates to be consistent, allowing some flexibility if the researcher wanted to use an alternative method to construct first-stage estimates.} This conditional s

When the number of instruments is small, this can be justified in large samples through central limit and continuous mapping theorems. However, when the number of instruments is large these standard tools cannot be applied. Instead, my asymptotic analysis proceeds using modifications of Lindeberg's interpolation argument \citep{Lindeberg1922}, which roughly proceeds by showing to be negligible the total change in distribution from a one-by-one replacement of each observation in the expression of the test statistic with a Gaussian version. The modifications of Lindeberg's original interpolation method are non-trivial and required to deal with the fact that derivatives of the test statistic with respect to individual observations may be unbounded generally, an issue arising from the ``divide-by-zero'' problem of weak identification \citep{ASS-2019}.\footnote{Existing interpolation results require these derivatives to be bounded while in this setting they may
not even have finite moments. Given this, these modified approaches may be of independent interest to a growing literature on direct Gaussian approximation techniques (\citet{Chatterjee-2006}, \citet{cck2013}, \citet{Pouzo-Quadratic}, \citet{CMY-2020}).} This interpolation argument requires some conditions on the first-stage estimates, in particular I require that the first-stage estimates take on a ``jackknife-linear'' form and suggest using a jackknife ridge regression in practice to allow \(d_z \gg n\). Interestingly, though, these first-stage estimates are not required to converge to limiting values which allows the researcher some flexibility if she wanted to use a different approach.

% When there is a single endogeneous variable analysis can be considerably simplified by taking advantage of the particular form of the test statistic, while in the general case a more involved argument is needed which relies on slightly stronger moment conditions.

Through the Gaussian approximation result, I examine the power properties of my proposed testing procedure in local neighborhoods of the null. These local neighborhoods are characterized by a bounded local power index. In the case of a single endogenous variable I show that, under an additional regularity condition, the local power index diverging implies that the test is consistent. Unfortunately, the process of partialling out the structural error can introduce bias into the first-stage estimate under the alternative hypothesis. Against certain alternatives, this bias can be particularly pronounced and erase the first-stage signal from the instruments, a problem pointed out by \citet{Andrews-2016} in the context of the \citet{Kleibergen-2005} K-statistic. To address this, I propose a simple combination with the sup-score test of \citet{BCCH-2012}, which can also allow \(d_z \gg n\). As with the Anderson-Rubin test, while the
sup-score statistic does not incorporate first-stage information it does not face a power decline against particular alternatives.

Identification-robust testing procedures may be of particular interest in high-dimensional settings due to a lack of clarity on how to pretest for weak identification. Current empirical practice when using post-LASSO estimates of the first-stage appears to be to use standard t-test inference if the first stage F-statistic on the LASSO selected variables is larger than 10 (\citet{Paravisini-2014}, \citet{gilchrist-glassberg-2016}, \citet{derenoncourt-2022}). Using a simple numerical demonstration, I argue that first stage F-statistics on LASSO selected variables may not be reliable indicators of identification strength. Given uncertainty about the strength of identification I apply the newly proposed testing procedures to the data of \citet{gilchrist-glassberg-2016} and generate weak instrument-robust confidence intervals for the effect of social spillovers on movie consumption. The newly proposed confidence intervals are smaller than those obtained by inverting the sup-score and many
instrument tests across all specifications. To verify these results, I also generate confidence intervals in the data of \citet{AK-1991}, again finding that the newly proposed methods generate tighter confidence bands than the many instruments methods.  In \Cref{sec:power-properties}, I argue these improvements in power may be explained by my proposed test statistic's use of higher quality first-stage estimates and individual scores that are uncorrelated across observations.

Finally, I examine the applicability of the theoretical results in this paper through a simulation study. While existing tests seem to face size distortions in alternate regimes, the test based on my proposed test statistic has nearly exact size in a variety of settings. While this new test may have diminished power against certain alternatives, this deficiency is ameliorated through combination with the sup-score test. After combining the new test with the sup-score statistic, I find the newly proposed testing is also have favorable power properties compared to the many instruments and sup-score test in my simulation setup.

The outline of this paper is as follows. \Cref{sec:setup} formally defines the model considered and introduces the jackknife K-statistic and it's combination with the sup-score test. \Cref{sec:empirical} demonstrates the usefuleness of the proposed testing procedures in two empirical applications while \Cref{sec:simulations} provides evidence from simulation study. \Cref{sec:single-endogeneous} provides an overview of the Gaussian approximation approach and characterizes the limiting behavior of the test statistic. \Cref{sec:power-properties} uses this characterization to examine the power properties of the test and establishes the validity of the combination test. Proofs of the main results as well as the presentation of some auxiliary results are deferred to the Appendix.

\paragraph{Notation.} For any \(n \in \SN\) let \([n]\) denote the set \(\{1,\dots,n\}\). I work with a sequence of probability measures \(P_n\) on the data \(\{(y_i,x_i,z_i): i \in [n]\}\). To accommodate independent but not identically distributed observations, let \(\E_n[f_i] = n^{-1}\sum_{i=1}^n f_i\) denote the empirical expectation and \(\bar\E[f] = \E_n[\E[f_i]]\) denote the average expectation operator.

\section{Model and Setup}
\label{sec:setup}

Though the analysis below allows for exogenous regressors, to simplify the exposition I follow \citet{Mikusheva_Sun_2022} and assume that they have already been partialed out of both the outcome, \(y_i\), and the endogenous regressors, \(x_i\). As the controls are assumed to be of fixed dimension, this is without loss of generality.\footnote{For discussion refer to \Cref{sec:exogenous-controls}.} Along with the first stage, the IV model can then be written as a system of simultaneous equations:
\begin{equation}
    \label{eq:iv-model}
    \begin{split}
        y_i &= x_i'\beta + \varepsilon_i \\ 
        x_i &= \Pi_i + v_i
    \end{split}
\end{equation}
The researcher observes the outcome \(y_i \in \SR\), the endogenous variable \(x_i \in \SR^{d_x}\), and the instruments \(z_i \in \SR^{d_z}\) but neither the structural error \(\varepsilon_i \in \SR\) nor the first-stage errors \(v_i \in \SR^{d_x}\). The structural error is assumed to be conditional-mean independent of the instruments, \(\E[\varepsilon_i|z_i] = 0\). I denote \(\E[x_i|z_i]\) as \(\Pi_i \coloneqq \E[x_i|z_i]\) and make no assumptions about the functional form of the conditional expectation so the instruments are allowed to affect the endogenous variable in a nonlinear fashion.

The random variables \(\{(z_i,\varepsilon_i, v_i)\}_{i=1}^n\) are assumed to be independent across observations. Observations need not be identically distributed but the errors are assumed to have a common covariance structure conditional on the instruments \(z_i\):
\[
    \Var((\varepsilon_i, v_i)'|z_i) \coloneqq \Omega(z_i) = 
    \begin{pmatrix} \sigma^2_{\eps\eps}(z_i) & \Sigma_{v\eps}(z_i) \\ \Sigma_{\eps v}(z_i) & \Sigma_{vv}(z_i) \end{pmatrix} \in \SR^{(1 + d_x) \times (1 + d_x)}
\]
As \(\Omega(z_i)\) is otherwise left unrestricted, the errors are allowed to be heteroskedastic. All results in this paper hold conditionally on a realization of the instruments \(\vz := (z_1',\dots, z_n') \in \SR^{n\,\times\,d_z}\) so from this point forth they are treated as fixed and all expectations can be understood as conditional on the instruments.
% Having described the basic model, I condition the rest of my analysis on the instruments \(z := (z_1',\dots, z_n')'\in \SR^{n\,\times\, d_z}\) so that they are treated as fixed in what is below.

Under this setup, the researcher wishes to test a two-sided restriction on the structural parameter:
\[
    H_0: \beta = \beta_0\vsbox H_1: \beta \neq \beta_0
\]
I am interested in constructing powerful tests for this null-alternate pair that are asymptotically valid under arbitrarily weak identification and with minimal restrictions on the number of instruments \(d_z\). To this end, define the null errors \(\varepsilon_i(\beta_0) \coloneqq y_i - x_i'\beta_0\). Using these, I construct a variable, \(r_i\), that is a ``partialed-out'' version of the endogenous variable satisfying \(\Cov(r_i,\eps_i(\beta_0))= 0\):
\begin{align*}
    r_i \coloneqq x_i - \rho(z_i)\eps_i(\beta_0),\;\;\; \rho(z_i) 
    &\coloneqq \frac{\Cov(\eps_i(\beta_0),x_i)}{\Var(\eps_i(\beta_0))}  \in \SR^{d_x} \\
    &\;= \frac{\Sigma_{v\eps}(z_i) + \Sigma_{vv}(z_i)(\beta - \beta_0)}{(1,\beta-\beta_0)'\Omega(z_i)'(1,\beta-\beta_0)}.
\end{align*}
Each element of the nuisance parameter \(\rho(z_i)\), \(\rho_\ell(z_i)\) for \(\ell = 1,\dots,d_x\), can be interpreted as the (conditional) slope coefficient from a simple linear regression of \(x_{\ell i}\) on \(\eps_i(\beta_0)\). Thus, if \(\rho_\ell(\cdot)\) falls in some function class \(\Phi\) it can be estimated directly under \(H_0\) by solving empirical analogs of:\footnote{Under \(H_1\), \(\rho_\ell(z_i)\) can be estimated directly by solving empirical analogs of \(\rho_\ell(z_i) = \arg\min_{\phi \in \Phi} \E[(x_{\ell i} - \eta_i(\beta_0)\phi(z_i))^2]\) where \(\eta_i(\beta_0) = \eps_i(\beta_0) - \E[\eps_i(\beta_0)|z_i]\). This requires an initial estimate of \(\E[\eps_i(\beta_0)|z_i]\), however.} 
\[
    \rho_\ell(z_i) = \arg\min_{\varphi\in\Phi} \bar\E[(x_{\ell i} - \eps_i(\beta_0)\varphi(z_i))^2].
\]
While other estimators of \(\rho(\cdot)\) are, in principle, possible (see \Cref{subsec:estimation-error-single}), I will focus on \(\ell_1\)-penalized/LASSO estimators. These estimators are consistent under the assumption that \(\rho(z_i)\) has an approximately sparse representation in some basis \(b(z_i) \coloneqq (b_1(z_i),\dots,b_{d_b}(z_i))' \in \SR^{d_b}\), that is \(\rho_\ell(z_i) = b(z_i)'\phi_\ell + \xi_{\ell i}\) where \(\xi_{\ell i}\) represents an approximation error that tends to zero with the sample size and \(\phi_\ell\) is sparse in the sense that many of its coefficients are zero. This allows for nesting of the low-dimensional case, where the number of instruments is fixed, and the high dimensional case, where the number of instruments is potentially much larger than the sample size, under a unified estimation procedure. Under homoskedasticity, \(\rho_\ell(z_i)\) is a constant function and thus has a spare representation in any basis that contains a constant term. In general, the aproximate sparsity assumption can either be interpreted as an assumption that there are only a few instruments that are
important for explaining variation in the covariance matrix \(\Omega(z_i)\) or as an assumption that the function \(\rho(z_i)\) can be accurately approximated using only a smaller set of basis terms in \(b(z_i)\).

The parameter \(\phi_\ell\) can be estimated via LASSO:
\begin{equation}
    \label{eq:gamma-hat-equation}
    \hat\phi_\ell = \arg\min_{\phi\in\SR^{d_b}} \E_n[(x_{\ell i} - \eps_i(\beta_0) b(z_i)'\phi)^2] + \lambda\|\phi\|_1,
\end{equation}
or via post-LASSO, refitting an unpenalized version of \eqref{eq:gamma-hat-equation} using only the basis terms associated with nonzero coeffecients in the inital LASSO regression. The estimating procedure in \eqref{eq:gamma-hat-equation} is a simple \(\ell_1\)-penalized regression of \(x_{\ell i}\) against \(\eps_i(\beta_0)b(z_i)\). It can be easily implemented using out-of-the-box software available on most platforms. Under standard conditions, this leads to a consistent estimate of \(\rho_\ell(z_i)\) as long as the sparsity condition \(s^2\log^M(d_bn)/n \to 0\) where \(s\) is the number of nonzero elements of \(\phi_\ell\) and \(M\) is a positive constant that depends on the moment bounds imposed. The estimation procedure is discussed in more detail in \Cref{subsec:estimation-error-single}. With \(\hat\rho(z_i) \coloneqq b(z_i)'\hat\phi_\ell\), I construct the estimated version of \(r_{\ell i}\), \(\hat r_{\ell i} \coloneqq x_i -
\hat\rho(z_i)\eps_i(\beta_0)\) for each \(\ell \in [d_x]\).

\subsection{Test Statistic}
\label{subsec:test-stat}

% With this setup, I introduce the test-statistic.
The test statistic is based on an arbitrary jackknife-linear estimate of the first stage,
\[
    \widehat\Pi_{\ell i} = \sum_{j\neq i} h_{ij}\hat r_{\ell j},\;\; \ell \in [d_x]
\]
for some ``hat'' matrix \(H = [h_{ij}] \in \SR^{n\times n}\). % This hat matrix can differ for each \(\ell =1,\dots, d_x\), but the notation becomes  messy. This extension is explored in the Appendix. 
The phrase ``hat matrix'' is borrowed from ordinary least squares (OLS) where the projection matrix, \(\vz(\vz'\vz)^{-1}\vz'\), is sometimes referred to as the hat matrix in the sense that \(\hat x = \vz(\vz'\vz)^{-1}\vz' x\). In practice, the hat matrix,  \(H\), can be any matrix that depends only on \(\vz\). It is important to note that while \(\widehat\Pi_{\ell i}\) does not depend on \(\hat r_{\ell i}\), it may depend on \(z_i\) through the hat matrix \(H\). This gives the test power against alternatives where \(\E[\eps_i(\beta_0)z_i] \neq 0\).
For technical reasons, I will assume that \(h_{ii} = 0\) for each \(i \in [n]\) so that \(\widehat\Pi_{\ell i}\) can be written as \(\widehat\Pi_{\ell i} = \sum_{j=1}^n h_{ij} r_{\ell j}\).

Formally, the only structure I require on the hat matrix \(H\) is a balanced-design condition described in \Cref{sec:single-endogeneous}. However, for reasons explained in \Cref{sec:power-properties} it may be optimal to introduce some regularization in estimating the first-stage models \(\widehat\Pi_{\ell i}\) so I suggest using a jackknife ridge regression procedure setting:
\begin{equation}
    \label{eq:ridge-hat-matrix}
    \widehat\Pi_i = z_i'\hat\pi_{-(i)}(\lambda^\star) 
\end{equation}
where \(\hat\pi_{(-i)}(\lambda)\) is the coefficient from a ridge regression of \(\hat r\) on \(\vz\), leaving out observation \(i\) and with penalty parameter set equal to \(\lambda\).
\footnote{A ridge regression coefficient estimate from a regression of \(\tilde\vy \in \SR^n\) on \(\tilde\vx\in \SR^{k\times n}\) with penalty parameter \(\lambda \in \SR_+\) solves \(\hat\pi \in \arg\min_{\pi\in\SR^K} \|\tilde\vy - \pi'\tilde\vx\|_2^2 + \lambda\|\pi\|_2^2\). \citet{nyquist1988applications} shows how the fitted values from a jackknife ridge regression can be calculated without having to recompute \(\hat\pi_{(-i)}\) for each observation.  \citet{Angrist_Krueger_JIVE} provides a similar analysis for jackknife OLS.} Following recommendations in \citet{harrell_2015} and \citet{lecture_notes_ridge}, the penalty parameter \(\lambda^\star\) is set so that the effective degrees of freedom is no more than a fraction of the sample size:
\[
    \lambda^\star = \inf \{\lambda \geq 0: \trace(\vz(\vz'\vz + \lambda I_{d_z})^{-1}\vz') \leq n/5\}
\]
The jackknife ridge first-stage estimate has the benefit of being well defined even when the number of instruments is larger than the sample size. 
% The eigenvalues of the ridge-hat matrix are always less weakly than one, suggesting that the matrix stays well behaved when the number of instruments is large, in the sense that \(\sum_{j\neq i} h_{ij}^2 \leq 1\). 
I stress, though, that the \(\widehat\Pi_{\ell i}\) estimators are not required to be consistent and the researcher may use any other hat matrices that she believes will lead to plausible first-stage estimates. Other possible choices of hat matrix include the jackknife OLS hat matrix of \citet{Angrist_Krueger_JIVE}, the deleted diagonal projection matrix introduced in \citet{Chao-2012-JIVE} and successfully used in \citet{KSS-2018,crudu_mellace_sandor_2021,Mikusheva_Sun_2022}, and \citet{Matsushita_Otsu_2022}, or hat matrices based on selecting instruments via some preliminary unsupervised technique such as principal component analysis (PCA). \Cref{rem:verifying-balanced-design} below discusses how the balanced-design condition may be verified for arbitrary choices of hat matrices.

For each \(i = 1,\dots,n\), define \(\widehat\Pi_i = (\widehat\Pi_{1i},\dots,\widehat\Pi_{d_xi})\in \SR^{d_x}\) and \(\widehat\Pi_{\eps i} = \eps_i(\beta_0)\widehat\Pi_i\). Collect these in the matrices 
\begin{equation}
    \label{eq:hatPi-hatPiEps}
    \begin{split}
        \varepsilon(\beta_0) 
        &= \big(\varepsilon_1(\beta_0),\dots,\varepsilon_n(\beta_0)\big)' \in \SR^n\\ 
        \widehat\Pi 
        &= \big(\widehat\Pi_1',\dots, \widehat\Pi_n'\big)' \in \SR^{n\,\times\,d_x} \\ 
        \widehat\Pi_\varepsilon 
        &= \big(\widehat\Pi_{\eps 1}',\dots,\widehat\Pi_{\eps n}'\big)' \in \SR^{n\,\times\,d_x}
    \end{split}
\end{equation}
The jackknife K-statistic can then be defined
\begin{align}
    \label{eq:test-stat}
    \JK(\beta_0) =  \eps(\beta_0)'\widehat\Pi\big(\widehat\Pi_\eps'\widehat\Pi_\eps\big)^{-1}\widehat\Pi'\eps(\beta_0)\times \bm{1}{\{\lambda_\text{min}(\widehat\Pi_\eps'\widehat\Pi_\eps) > 0\}}
\end{align}

I will show that, under appropriate moment bounds and conditions on the hat matrix, \(H\), the limiting distribution of \(\JK(\beta_0)\) under \(H_0\) is \(\chi^2_{d_x}\). For exposition, I will largely focus on the case where \(d_x = 1\), in which case the form of the test statistic simplifies to \(\JK(\beta_0) = \big(\sum_{i=1}^n \eps_i(\beta_0)\widehat\Pi_i\big)^2/\sum_{i=1}^n \eps_i^2(\beta_0)\widehat\Pi_i^2\). The extension to \(d_x > 1\) is not immediate but is possible under strengthened moment conditions.

\begin{remark}[]
    \label{rem:}
    % While the jackknife K-statistic is similar in spirit to the K-statistics of \citet{Kleibergen_2002,Kleibergen-2005} in that 
    While use of first-stage estimates that are uncorrelated with the structural error is inspired by \citet{Kleibergen_2002,Kleibergen-2005}, the form of the jackknife K-statistic is distinct from that of the original K-statistics. One major difference is in how both test statistics account for heteroskedasticity. The K-statistic of \citet{Kleibergen-2005} accounts for heteroskedastic errors using a \(d_z \times d_z\) matrix, which cannot be consistently estimated when \(d_z\) is large. 
    In contrast, the jackknife K-statistic uses the heteroskedasticity robust variance estimate \((\widehat\Pi_\eps'\widehat\Pi_\eps)^{-1} \in \SR^{d_x\,\times\,d_x}\). Showing that these variance estimates can be used to account for heteroskedasticity is a feature of the direct Gaussian approximation approach. Under weak identification the distribution of the variance estimate is relevant to the distribution of the test-statistic. However, even when \(d_z \ll n\), the distribution of this variance estimate would be difficult to analyze using traditional central limit theorems as it is not a continuous function of a sample mean or even of a quadratic form.
\end{remark}

\subsection{Combination with Sup-Score Test}

As will be discussed further in \Cref{sec:power-properties}, the test based on the \(\JK(\beta_0)\) statistic can have deficient power against certain alternatives. This loss of power is similar to that faced by tests based on the K-statistics of \citet{Kleibergen_2002,Kleibergen-2005} and derives from the fact that the process of partialling out the null errors, \(\eps(\beta_0)\), from the endgenous variables introduces bias into the first stage estimates, \(\widehat\Pi\), under the alternative hypothesis. Against certain alternatives, this bias can ``erase'' the first stage signal from the instruments.

The power deficiency in tests based on the K-statistic is typically adressed by combining the K-statistic with the Anderson-Rubin statistic based on a constructed conditioning variable. Prominent examples of such combinations include the celebrated conditional likelihood test of \citet{Moreira_2003}, GMM-M test of \citet{Kleibergen-2005}, and minimax regret tests of \citet{Andrews-2016}. I take a similar approach here in combining the newly proposed tests with tests based on the sup-score statistic of \citet{BCCH-2012},
\begin{equation}
    \label{eq:sup-score-statistic0}
    S(\beta_0) \coloneqq \sup_{\ell \in  [d_z]} \bigg|\frac{\sum_{i=1}^n \eps_i(\beta_0)z_{\ell i}}{(\sum_{i=1}^n z_{\ell i}^2)^{1/2}}\bigg|.
\end{equation}
which have correct asymptotic size even when the instruments is much larger than the sample size, \(d_z \gg n\). A level \(\alpha \in (0,1)\) test based on the sup-score statistic rejects whenever \(S(\beta_0) > c^S_{1-\alpha}\) where, for \(e_1,\dots,e_n\) iid standard normal and generated independently of the data, \(c^S_{1-\alpha}\) is the simulated multiplier bootstrap critical value:\footnote{This conditional quantile can also be approximated using an empirical bootstrap procedure as demonstrated by \citet{deng2020beyond}.}
\[
    c^S_{1-\alpha} \coloneqq (1-\alpha)\text{ quantile of } \sup_{1 \leq \ell \leq d_z} \bigg|\frac{\sum_{i=1}^n e_i \eps_i(\beta_0)z_{\ell i}}{(\sum_{i=1}^n z_{\ell i}^2)^{1/2}}\bigg|\text{ conditional on }\{(y_i,x_i,z_i)\}_{i=1}^n.
\]
As with the Anderson-Rubin test, tests based on the sup-score statistic may have suboptimal power properties in overidentified models as it does not incorporate first-stage information. However, the sup-score statistic does retain the benefit of directing power evenly in all directions, avoiding pitfalls of tests based on \(\JK(\beta_0)\) against certain alternatives. 

The combination test decides which test to run based on an attempt to detect whether the alternative \(\beta\) is such that \(\E[\widehat\Pi_{\ell,i}^I] = 0\) for all \(i \in [n]\) and \emph{some} \(\ell \in [d_x]\).  When this happens, tests based on the \(\JK(\beta_0)\) statistic have trivial power against deviations in the \(\ell\textsuperscript{th}\) coordinate of \(\beta\), so in local neighborhoods of these values of \(\beta\), tests based on the sup-score statistic may be preferable. Detection of whether \(\E[\widehat\Pi_{\ell,i}^I] = 0\) is based on the conditioning statistic:
\begin{equation}
    \label{eq:conditioning-statistic}
    \begin{split}
        C = \inf_{\ell \in [d_x]}\sup_{i \in [n]} \bigg|\frac{\sum_{j\neq i} h_{ij}\hat r_{\ell j}}{(\sum_{j\neq i} h_{ij}^2)^{1/2}}\bigg|.
    \end{split}
\end{equation}
Under the assumption that \(\E[\widehat\Pi_i^I] = 0\) for all \(i \in [n]\), quantiles of the conditioning statistic can be simulated analogously to the sup-score critical value. For a new set of \(\{(e_{\ell1},\dots,e_{\ell n}): \ell \in [d_x]\}\) iid standard normal and generated independently of the data, and for any \(\theta\in (0,1)\), define the conditional quantile 
\begin{equation}
    \label{eq:pre-test-quantile}
    c^C_{1-\theta} \coloneqq (1-\theta)\text{ quantile of }\inf_{\ell \in [d_x]}\sup_{i \in [n]} \bigg|\frac{\sum_{j\neq i}e_i h_{ij}\hat r_{\ell j}}{(\sum_{j\neq i} h_{ij}^2)^{1/2}}\bigg| \text{ conditional on }\{(y_i,x_i,z_i)\}_{i=1}^n
\end{equation}
The thresholding test decides which test to run by comparing the conditioning statistic \(C\) to a treshhold value \(\tau\), 
\begin{equation}
    \label{eq:conditioning-test0}
    T(\beta_0;\tau) = 
    \begin{cases}
        \bm{1}\{\JK(\beta_0) > \chi^2_{d_x, 1-\alpha} &\text{if }C \geq \tau \\
            \bm{1}\{S(\beta_0) > c_{1-\alpha}^S &\text{if } C < \tau
    \end{cases}.
\end{equation}
In principle, the thresholding statistic has correct size for any (preset) choice of parameter \(\tau\). In practice, however, I find that setting \(\tau = c^C_{0.75}\) leads to a reasonable balance of power between local and distant alternatives. 

% As will be discussed further in \Cref{sec:power-properties}, the process of partialling out the null errors, \(\eps(\beta_0)\), from the endogenous variable introduces bias into the first stage estimates \(\widehat\Pi\). Against certain alternatives, this bias can be particularly pronounced and ``erase'' the first stage signal from the instruments. Formally, this occurs when \(\E[\widehat\Pi_i] = 0\) for all \(i = 1,\dots,n\). This loss of power against certain alternatives is inherited from the

\section{Empirical Application}
\label{sec:empirical}

I apply the testing procedures proposed in this paper to the data of \citet{gilchrist-glassberg-2016}, who examine the effect of social spillovers in movie ticket sales, and to the data of \citet{AK-1991}, who examine the returns to education. In both studies the number of instruments, \(d_z = 52\) and \(d_z = 180\), respectively, cannot be treated as negligible relative to the sample size (\(n = 1{,}671\) and \(n = 329{,}509\), respectively). To deal with the large number of instruments, \citet{gilchrist-glassberg-2016} employ a post-LASSO estimate of the first stage. This strategy is also shown to work well in the data of \citet{AK-1991} by \citet{angrist2022machine}. Using a simple simulation study, I demonstrate that the first-stage F-statistics on LASSO selected variables typically reported by authors can be misleading indicators of identification strength. When revisting the initial analyses using identification robust testing procedures, the confidence intervals constructed
by inverting the tests proposed in \Cref{sec:setup} are consistently narrower than those constructed from inverting the sup-score and many-instrument testing procedures.

\subsection{Application to Social Spillovers in Movie Consumption}

The \citet{gilchrist-glassberg-2016} sample consists of 1,671 opening weekend days between January 1, 2002 and January 1, 2012. For each opening weekend, the authors observe gross ticket sales for movies wide released in theaters in the United States with a run in theaters of at least six weeks.\footnote{An opening weekend day is a Friday, Saturday, or Sunday of opening weekend and a wide released movie is any movie that ever shows on 600 or more screens.} The data are obtained through Box Office Mojo, a subsidiary of the Internet Movie Database (IMDb).

The outcome variables of interest are gross ticket sales of movies that opened in a given weekend in the second through sixth weeks of their run, while the endogenous variable is the gross ticket sales of a movie in its opening weekend. To control for seasonal periodicity in both the supply of and demand for movies, a vector of date controls are included. Formally, the authors are interested in the parameters \(\beta_w\), \(w = 2,\dots,7\) from the linear IV model(s):
\begin{equation}
    \label{eq:gs-main}
    \text{Sales}_{wi}^\perp = \beta_w\text{Sales}^\perp_{1i} + \eps_{wi}
\end{equation}
where, for \(w = 1,\dots,6\), \(\text{Sales}_{wi}^\perp\) represents gross national ticket sales, after the partialing out of date controls and a constant, \(7(w-1)\) days after day \(i\), of movies that opened on the opening weekend of \(i\). The variable \(\text{Sales}_{7i}^\perp = \sum_{w = 1}^6 \text{Sales}_{wi}^\perp\) denotes the cumulative national ticket sales in the second through sixth running weekends of movies who opened in weekend \(i\), after the partialing out of date controls and a constant.
The parameter \(\beta_w\) represents the social spillover effect of strong opening weekend sales on sales in later weeks. 

To instrument for sales on opening weekend the authors employ a vector of nationally aggregated weather measures. These weather measures include the proportion of movie theaters experiencing maximum temperatures in \(5^\circ\) Fahrenheit bins on the interval \([10^\circ, 100^\circ]\), the proportion of movie theaters experiencing precipitation levels in 0.25 inch per hour increments on the interval \([0,1.5]\), and the proportions of theaters experiencing any type of snow and of theaters experiencing any type of rain. Since unusually poor weather may cause people to substitute away from outdoor activities and into watching a movie, these measures provide a source of exogenous variation in opening weekend sales that can be used to identify the effect of social spillovers. 

Putting together the nationally aggregated weather measures leaves \citet{gilchrist-glassberg-2016} with 48 linearly independent instrumental variables.
\footnote{There are 52 instruments in total, but four linearly dependent ones are ignored in the following.} 
To handle the large number of instruments, the authors  employ a post-LASSO estimate of the first stage \citep{BCCH-2012}; they set the first-stage penalty parameter so that the number of instrument selected is one, two, or three. The resulting first-stage F-statistics using the selected instrument(s), 38.80, 25.86, and 20.95, respectively, seem to indicate strong identification. However, the first-stage F-statistic on the full set of instrumental variables is only 3.80. Moreover, since the LASSO objective is an \(\ell_1\) penalized version of the OLS loss, using the variables selected by LASSO may mechanically lead to higher F-statistics even if the underlying relationship between the instruments and the endogenous variables is weak.

% After the partialing out a constant and the date controls, four of these are linearly dependent. I discard these and work with the remaining 48 partialed-out instruments in my analysis.

\Cref{tab:selected_fstat_sim} provides evidence from a simple simulation experiment to demonstrate this. For the simulation experiment I generate an iid sample of size \(n = 1000\). For each \(i \in [n]\), I generate 10 mutually independent instruments \(Z_{ki} \sim N(0,1)\) for \(1 \leq k \leq 10\). The endogenous variable is generated to only have a weak relationship with the instruments, \(X_i = \frac{2}{\sqrt{n}}\sum_{\ell = 1}^{10} Z_{\ell i}+ v_i\), and the outcome is generated \(Y_i = X_i + \eps_i\) where \((\eps_i,v_i)\) are independent standard normals. From the initial set of 10 instrumental variables I generate an additional 55 technical instruments by squaring and taking all interactions between variables in the initial set. These generated instruments are correlated with the initial instruments but do not directly enter the first stage.

I then set the LASSO penalty so that only a certain number of instruments are chosen and report the resulting average first stage F-statistics and 95\% confidence interval coverage over one thousand simulations. As a comparasion I also report the average first-stage F-statistics and 95\% confidence interval coverage from the oracle estimator, which only uses the relevant 10 initial instruments. Despite the fact that the first-stage F-statstic on selected instruments is more than double the first-stage F-statistic using the oracle first stage estimator, the coverage rate of 95\% confidence intervals based on LASSO selected instruments is significantly degraded compared to both the nominal coverage probability and the coverage probability using the oracle first-stage estimator. 

\begin{table}
    \centering\small
    \begin{tabular}{|c||cc||cc|}
        \hline
        & \multicolumn{2}{c||}{\it Selected Instruments} & \multicolumn{2}{c|}{\it Oracle Estimator} \\
        Number of Instruments & F-stat. & Coverage Prob. & F-stat. & Coverage Prob.\\
        \hline
        One Instrument & 12.539 & 0.302 & 4.911 & 0.904 \\  
        Two Instruments & 11.185 & 0.150 & 5.040 & 0.830 \\ 
        Three Instruments & 10.060 & 0.070 & 4.820 & 0.810 \\ 
        \hline
    \end{tabular}
    \captionsetup{width = 0.6\linewidth, font = small}
    \caption{Comparasions of first-stage F-statistics and 95\% confidence interval coverage Probability using selected and oracle instruments}
    \label{tab:selected_fstat_sim}
\end{table}

% As seen in \Cref{tab:selected_fstat_data}, these first-stage F-statistics increase significantly as the number of selected instruments decreases. While the ``true'' F-statistic, computed with only the 10 initial instruments directly relevant for the first stage, is only 5.234, the average F-statistic on the selected variables can be larger than 40. The persistence of this pattern between sample sizes \(n = 500\) and \(n = 1000\) suggests that this is not a small-sample issue and that pretesting for weak identifications based on post-LASSO F-statistics may be problematic generally. \Cref{fig:selected_fstat_data} shows how the first stage F-statistic changes with the number of LASSO-selected variables in the \citet{gilchrist-glassberg-2016} data. The pattern is similar to that seen in the \Cref{fig:selected_fstat_simulation} simulation experiment. 
%
% \begin{figure}[p!]
%     \centering
%     \includegraphics[width=0.8\linewidth]{figures/SelectedF_statistics_data.png}
%     \captionsetup{width = 0.6\linewidth, font = small}
%     \caption{First-Stage F-statistic as Function of Number of LASSO-Selected Variables in the Data of \citet{gilchrist-glassberg-2016}.}
%     \label{fig:selected_fstat_data}
% \end{figure}

Given a lack of clarity on the strength of identification, I seek to validate the results of \citet{gilchrist-glassberg-2016} using the weak identification testing procedures proposed in this paper. The setting is particularly suitable for weak IV testing using the jackknife K-statistic. With 48 instruments and a sample size of 1671, \(d_z^3 = 110{,}592 \gg n \), making the tests of \citet{Moreira_2003,Moreira_2009,Kleibergen-2005}, and \citet{Andrews-2016} inapplicable. On the other hand, it is unclear whether asymptotic approximations based on \(d_z \to \infty\) will accurately describe the finite-sample distribution of test statistics with 48 instruments. Moreover, since fluctuations in movie theater attendance seem to be largely driven by either particularly cold or particularly hot weather  (see Figure 4 in \citet{gilchrist-glassberg-2016}), the nuisance parameter \(\rho(z_i)\) is plausibly approximately sparse.

I compare the 95\% confidence intervals based on the jackknife K-statistic to those based on the jackknife Lagrange-Multiplier (JLM) statistic of \citet{Matsushita_Otsu_2022} and the  sup-score statistic of \citet{BCCH-2012}. Confidence intervals based on the jackknife AR statistic of \citet{Mikusheva_Sun_2022} are not reported as they were empty for all specifications, a result that could indicate misspecification of the linear model. Similarly, confidence intervals based on the thresholding statistic, implemented as recommended in \Cref{sec:power-properties}, are also not reported as they always align with the those of jackknife K-statistic. The narrower confidence bands of the jackknife K-statistic in specifications where the sup-score confidence interval is non-empty indicates higher power from the jackknife K-test in this setting, so the combination test suggesting its use may be expected. The jackknife
K-statistic is implemented using the jackknife OLS hat matrix of \citet{Angrist_Krueger_JIVE}, that is by setting \(\lambda = 0\) in \eqref{eq:ridge-hat-matrix}.

Tables~\ref{table:CI_empirical} reports the \(95\%\) confidence intervals for \(\beta_1,\dots,\beta_7\) generated by weak-instrument robust confidence intervals for three sets of instruments: the first is the initial set of 48 instruments in \citet{gilchrist-glassberg-2016}, the second set includes only the temperature instruments for \(d_z = 36\), and the final includes the initial instruments as well as all interactions between the temperature instruments and the remaining instruments for \(d_z = 524\). 
For reference, I also provide point estimates and standard errors for \(\beta_1,\dots,\beta_7\) from \citet{gilchrist-glassberg-2016}, Table 2. To facilitate comparison, these point estimates and standard errors come from a specification that uses all the instruments in the first stage of a 2SLS procedure. 

% To form the \(\JK(\beta_0)\) statistic I use the choice of hat matrix in \eqref{eq:ridge-hat-matrix} and estimate the auxiliary parameter \(\rho(z_i)\) as in \eqref{eq:gamma-hat-equation}. The penalty parameter \(\lambda\) is chosen with leave-one-out cross-validation using the \texttt{cv.glmnet} command from the \texttt{glmnet} package in R \citep{Rsoftware,glmnet}. The critical value for the sup-score statistic \(S(\beta_0)\) is simulated using 2,500 bootstrap draws. Confidence intervals based on the combination test, \(T(\beta_0;\tau)\), are not directly reported as the pretesting procedure based on simulating the 75\textsuperscript{th} quantile of \(C\) as in \eqref{eq:pre-test-quantile} always suggests using the \(\JK(\beta_0)\) statistic. 
% The variance of the jackknife LM test statistic is estimated using the cross-fit variance estimators proposed by \citet{Mikusheva_Sun_2022}, which are shown to improve power by \citet{Mikusheva_Sun_2022} and \citet{lim2022conditional}.

Qualitatively, the results from the weak-instrument robust confidence intervals are similar to that of the author's original analysis; indeed the \citet{gilchrist-glassberg-2016} point estimates are always in the identification robust confidence intervals when using either the initial instrument set (\(d_z = 48\)) or the reduced instrument set (\(d_z = 36\)). However, when using the larger instrument set of \(d_z = 524\) we obtain confidence bands using the jackknife K-test that rule out the author's initial point estimates for the parameters \(\beta_5, \beta_6\) and \(\beta_7\) and suggest somewhat smaller social spillover effects in movie consumption. Across all specifications and instrument sets, confidence intervals obtained from inverting the jackknife K-test are consistently smaller than those obtained from inverting both the sup-score and jackknife Lagrange Multiplier test. These reductions in confidence interval length are most noticeable when using the complete set of
interactions, across all parameters the jackknife K-confidence intervals are nearly half the length of their jackknife Lagrange Multiplier counterparts.  As with the jackknife Anderson-Rubin test, the sup-score confidence intervals are also often empty which again could suggest misspecification of the linear IV model.

\renewcommand{\arraystretch}{1.75}
\begin{table}[h!]
    \centering \small 
    \resizebox{0.9\columnwidth}{!}{
    \begin{tabular}{|c|cccccc|}
        \hline 
        Parameter & \(\beta_2\) & \(\beta_3\) & \(\beta_4\) & \(\beta_5\) & \(\beta_6\) & \(\beta_7\) \\ 
        \hhline{|=|======|}
        \(\underset{\text{(s.e.)}}{\text{Estimate}}\) 
                  & \(\underset{\vphantom{\big(}(0.024)}{0.475}\) 
                  & \(\underset{\vphantom{\big(}(0.023)}{0.269}\) 
                  & \(\underset{\vphantom{\big(}(0.017)}{0.164}\) 
                  & \(\underset{\vphantom{\big(}(0.013)}{0.121}\) 
                  & \(\underset{\vphantom{\big(}(0.010)}{0.093}\) 
                  & \(\underset{\vphantom{\big(}(0.074)}{1.222}\)  \\
        \hhline{|=|======|} 
        & \multicolumn{6}{c|}{Initial instrument set, \(d_z = 48\)}\\
        \hhline{|=|======|}
        \(\JK(\beta_0)\) & \(\overset{\bflarrow 0.114 \flarrow}{[0.441, 0.555]}\) 
                         & \(\overset{\bflarrow 0.114 \flarrow}{[0.234, 0.348]}\)
                         & \(\overset{\bflarrow 0.074 \flarrow}{[0.127, 0.201]}\)
                         & \(\overset{\bflarrow 0.074 \flarrow}{[0.936, 0.167]}\) 
                         & \(\overset{\bflarrow 0.046  \flarrow}{[0.0736, 0.120]}\)
                         & \(\overset{\bflarrow 0.375 \flarrow}{[0.989, 1.365]}\)   \\ 
        \(S(\beta_0)\) & \(\emptyset\) 
                       & \(\overset{\bflarrow 0.033 \flarrow}{[0.294, 0.328]}\) 
                       & \(\emptyset\) 
                       & \(\emptyset\) 
                       & \(\emptyset\)
                       &  \(\overset{\bflarrow 0.561 \flarrow}{[0.989, 1.551]}\) \\  
        JLM & \(\overset{\bflarrow 0.140 \flarrow}{[0.428, 0.569]}\) 
            & \(\overset{\bflarrow 0.127 \flarrow}{[0.221, 0.348]}\)  
            & \(\overset{\bflarrow 0.087 \flarrow}{[0.134, 0.221]}\) 
            & \(\overset{\bflarrow 0.074 \flarrow}{[0.100, 0.174]}\)  
            & \(\overset{\bflarrow 0.060 \flarrow}{[0.080, 0.140]}\)
            & \(\overset{\bflarrow 0.441 \flarrow}{[0.989, 1.384]}\)  \\
        \hhline{|=|======|} 
        & \multicolumn{6}{c|}{Temperature instruments only, \(d_z = 36\)}\\    
        \hhline{|=|======|}
        \(\JK(\beta_0)\) & \(\overset{\bflarrow 0.147 \flarrow}{[0.462, 0.609]}\) 
                         & \(\overset{\bflarrow 0.134 \flarrow}{[0.268, 0.401]}\) 
                         & \(\overset{\bflarrow 0.107 \flarrow}{[0.158, 0.254]} \)
                         & \(\overset{\bflarrow 0.080 \flarrow}{[0.114, 0.194]} \) 
                         & \(\overset{\bflarrow 0.067 \flarrow}{[0.094, 0.161]}\)
                         & \(\overset{\bflarrow 0.455 \flarrow}{[1.117, 1.572]}\)   \\ 
        \(S(\beta_0)\) & \(\overset{\vphantom{\big(}}{\emptyset}\) 
                       & \(\overset{\vphantom{\big(}}{\emptyset}\) 
                       & \(\overset{\vphantom{\big(}}{\emptyset}\)     
                       & \(\overset{\vphantom{\big(}}{\emptyset}\) 
                       & \(\overset{\vphantom{\big(}}{\emptyset}\)
                       & \(\overset{\vphantom{\big(}}{\emptyset}\) \\  
        JLM & \(\overset{\bflarrow 0.161 \flarrow}{[0.448, 0.609]}\) 
            & \(\overset{\bflarrow 0.161 \flarrow}{[0.248, 0.408]}\)  
            & \(\overset{\bflarrow 0.107 \flarrow}{[0.147, 0.254]}\) 
            & \(\overset{\bflarrow 0.087 \flarrow}{[0.114, 0.200]}\) 
            & \(\overset{\bflarrow 0.067 \flarrow}{[0.094, 0.161]}\)
            & \(\overset{\bflarrow 0.542 \flarrow}{[1.070, 1.612]}\)  \\ 
        \hhline{|=|======|}
        &\multicolumn{6}{c|}{Initial instruments plus all interactions with temp. instruments, \(d_z = 524\)}\\
        \hhline{|=|======|}
        \(\JK(\beta_0)\) & \(\overset{\bflarrow 0.040 \flarrow}{[0.462, 0.502]}\) 
                         & \(\overset{\bflarrow 0.033 \flarrow}{[0.268, 0.301]}\) 
                         & \(\overset{\bflarrow 0.020 \flarrow}{[0.154, 0.174]}\) 
                         & \(\overset{\bflarrow 0.013 \flarrow}{[0.094, 0.107]}\) 
                         & \(\overset{\bflarrow 0.007 \flarrow}{[0.067, 0.074]}\) 
                         & \(\overset{\bflarrow 0.120 \flarrow}{[1.043, 1.164]}\)   \\ 
        \(S(\beta_0)\) & \(\overset{\bflarrow 0.047 \flarrow}{[0.415, 0.462]}\) 
                       & \(\overset{\vphantom{\big(}}{\emptyset}\) 
                       & \(\overset{\vphantom{\big(}}{\emptyset}\)     
                       & \(\overset{\bflarrow 0.207 \flarrow}{[0.040, 0.247]}\) 
                       & \(\overset{\bflarrow 0.060 \flarrow}{[0.161, 0.221]}\) 
                       &  \(\overset{\vphantom{\big(}}{\emptyset}\) \\  
        JLM & \(\overset{\bflarrow 0.080 \flarrow}{[0.441, 0.522]}\) 
            & \(\overset{\bflarrow 0.060 \flarrow}{[0.247, 0.308]}\)  
            & \(\overset{\bflarrow 0.040 \flarrow}{[0.147, 0.187]}\) 
            & \(\overset{\bflarrow 0.027 \flarrow}{[0.094, 0.120]}\)  
            & \(\overset{\bflarrow 0.013 \flarrow}{[0.067, 0.080]}\) 
            & \(\overset{\bflarrow 0.227 \flarrow}{[0.990, 1.217]}\) 
            \\ \hline
    \end{tabular}}
    \captionsetup{width = 0.95\linewidth, font = small}
    \caption{95\% Confidence Intervals and Interval Lengths in the data of \citet{gilchrist-glassberg-2016}.}
    \label{table:CI_empirical}
\end{table}
\renewcommand{\arraystretch}{1}

\subsection{Application to Returns to Education}

For additional comparison, I revisit the setting of \citet{AK-1991}, who study the effect of education on log-wages, instrumenting for education using various combinations of quarter-of-birth (QOB), year-of-birth (YOB), and place-of-birth indicators (POB). This dataset has the benefit of being used in the empirical studies of both \citet{Mikusheva_Sun_2022} and \citet{Matsushita_Otsu_2022}, facilitating an easy comparison of the performance of the newly proposed testing procedure to these existing ``many-instrument'' methods. Moreover, \citet{angrist2022machine} report improved power in this dataset when using post-LASSO estimates in the first stage, suggesting that high-dimensional or machine-learning techniques can be useful in this setting. 

\Cref{table:CI-AK} displays the resulting identification robust confidence intervals using two set of instruments. The first set, which contains all QOB \(\times\) YOB and all QOB \(\times\) POB interactions for \(d_z = 180\),  corresponds to the specification in Table VII of \citet{AK-1991}.  The second set of instruments, \(d_z = 1{,}530\), contains all QOB \(\times\) YOB \(\times\) POB interactions. The confidence intervals for the jackknife Anderson-Rubin (JAR) and jackknife Lagrange-Multiplier (JLM) tests are taken directly from \citet{Mikusheva_Sun_2022} and \citet{Matsushita_Otsu_2022}, respectively. To implement the jackknife K-test I follow a similar procedure to that of the prior empirical exercise, setting \(\lambda = 0\) when constructing the \(\widehat\Pi_i\) values. However, for computational reasons I opt to split the data into 11 pieces when estimating \(\widehat\Pi_i\) rather than using a true ``leave-one-out'' jackknife approach. Confidence intervals from the combination test are not reported as the sup-score confidence interval is empty in both specifications. 

The results in \Cref{table:CI-AK} are similar to those in \Cref{table:CI_empirical}. In both specifications the confidence intervals obtained from inverting the jackknife K-test are substantially smaller than those obtained by inverting the ``many-instrument'' JAR and JLM tests. Indeed, as before, the confidence intervals obtained from inverting the jackknife K-statistic are nearly half the length of those obtained from inverting the JLM test and less than a quarter the length of those obtained from inverting the JAR test. It is notable that the range of plausible returns to education values implied by the jackknife K-confidence interval with \(d_z = 1{,}530\) are strictly below that implied by the jackknife-K confidence interval with \(d_z = 180\). This may be related to misspecification of the linear model as indicated by the empty sup-score confidence interval. ``Downward bias'' of the confidence
intervals when using the larger instrument set is also seen for the JAR and JLM confidence intervals, though to a lesser extent. \Cref{rem:jlm-comparasion}, below, provides reasoning for why confidence intervals based on the \(\JK(\beta_0)\) statistic may be tighter than those based on the JAR and JLM statistics.  
\renewcommand{\arraystretch}{1.75}
\begin{table}[h!]
    \centering \small
    \begin{tabular}{|c|cccc|}
        \hline
        Number of Instruments & JAR & JLM & \(\JK(\beta_0)\) & \(S(\beta_0)\)\\ 
        \hhline{|=|====|} 
        \(\underset{(d_z = 180)}{\text{Initial Instrument Set}}\) & 
        \(\overset{\bflarrow 0.193 \flarrow}{[0.008, 0.201]}\) & \(\overset{\bflarrow 0.066 \flarrow}{[0.067, 0.133]}\) & \(\overset{\bflarrow 0.034 \flarrow}{[0.067, 0.101]}\) & 
        \(\emptyset\) \\ 
        \(\underset{(d_z = 1{,}530)}{\text{All Interactions}}\) & 
        \(\overset{\bflarrow 0.249 \flarrow}{[-0.047, 0.202]}\) & 
        \(\overset{\bflarrow 0.098 \flarrow}{[0.025, 0.123]}\) &  
        \(\overset{\bflarrow 0.034 \flarrow}{[0.008, 0.042]}\) & 
        \(\emptyset\)\\ 
                                                               \hline
    \end{tabular}
    \captionsetup{width = 0.95\linewidth, font = small}
    \caption{95\% Confidence Intervals and Interval Lengths in the data of \citet{AK-1991}.}
    \label{table:CI-AK}
\end{table}
\renewcommand{\arraystretch}{1}

As always, these are just particular empirical examples, and my results should not be interpreted as a critique of \citet{BCCH-2012}, \citet{Mikusheva_Sun_2022}, or \citet{Matsushita_Otsu_2022}, whose prior work I relied upon and was inspired by.

\section{Simulation Study}
\label{sec:simulations}

In a simple simulation study, I examine the performance of tests based on the \(\JK(\beta_0)\) statistic and compare it with that of other tests that may be used in settings where the number of instruments is nonnegligible as a fraction of sample size.  I consider a reduced-form data-generating process (DGP) similar to that of \citet{Matsushita_Otsu_2022}. The outcome variable, \(y_i\), and endogenous variable, \(x_i\), are generated according to
\begin{equation}
   \begin{split}
       y_i &= x_i + \eps_i \\ 
       x_i &=  \Pi_i + v_i
   \end{split} 
\end{equation}
where \(\Pi_i = \frac{1}{r_n}\sum_{k=1}^5 \frac{3}{4}\bar z_{ki} + \frac{1}{4}\bar z_{ki}^2 + \frac{1}{4}\bar z_{ki}^3\) is a (dense) transformation of an initial set of instruments \(\bar z_i \in \SR^{15}\) generated as described below. The value of \(r_n\) varies depending on the strength of identification considered; under weak identification, \(r_n = n^{-1/2}\) while for power curves I consider an intermediate identification strength, \(r_n = n^{-1/3}\).\footnote{Intermediate identification strength is considered to let the power curves come up to one at the boundaries of the considered range of \(\beta\) values. Power curves with \(r_n = n^{-1/2}\) look similar, but shrunk towards zero.} To model heteroskedasticity, the errors \((e_i,v_i)\) are generated \(\eps_i = (1 + \varrho_1(\bar z_{1i}^2 + \bar z_{2i}^2 + \bar z_{2i}\bar z_{3i}))e_{1i}\), and \(v_i = \varrho_2(1 + \bar z_{1i})\eps_i + (1 - \varrho_2)^2e_{2i}\) where \(e_{1i}\) and \(e_{2i}\) are generated independently of each other and other variables in the model according to a Laplace distribution with location parameter \(\mu = 0\) and scale parameter \(b = 1\). Since jackknife K-statistic has an exact \(\chi^2\) distribution when the errors are jointly Gaussian and \(\rho(z_i)\) is known, I purposefully avoid normally distributed errors to investigate the quality of asymptotic approximations.  The parameters \(\varrho_1\) and \(\varrho_2\) control the degree of heteroskedasticity and endogeneity, respectively.

% Simulation tables
\begin{table}[htb!]
    \centering\small
    \resizebox{0.9\columnwidth}{!}{
    \begin{tabular}{|cccc|ccccccc|}
    \hline
    \multicolumn{4}{|c|}{\textbf{DGP}} & \multicolumn{7}{c|}{\textbf{Testing Procedure}} \\
    \rule{0pt}{0.1em} & & & & & & & & & &  \\
    \(n\) & \(d_z\) & \(\varrho_1\) & \(\varrho_2\) & \(\JK(\beta_0)\) & \(S(\beta_0)\) & \(T(\beta_0; \tau_{0.3})\) & \(T(\beta_0;\tau_{0.75})\) & A.Rbn. & JAR & JLM \\
    \hhline{|====|=======|}
    200 & 15 & 0.2 & 0.3 &  0.0514 & 0.0356 & 0.0482 & 0.0454 & 0.0234 & 0.0454 & 0.0434 \\
    & & 0.2 & 0.6 &   0.0500 & 0.0376 & 0.0460 & 0.0412 & 0.0258 & 0.0728 & 0.0436 \\
    & & 0.5 & 0.3 &  0.0466 & 0.0384 & 0.0430 & 0.0402 & 0.0238 & 0.0784 & 0.0450 \\
    & & 0.5 & 0.6 & 0.0454 & 0.032 & 0.0432 & 0.0394 & 0.0220 & 0.0734 & 0.0458 \\ 
    \rule{0pt}{0.1em} & & & & & & & & & &  \\ 
    & 45 & 0.2 & 0.3 &  0.0430 & 0.0102 & 0.0372 & 0.0268 & 0.0062 & 0.0930 & 0.0386 \\
    & & 0.2 & 0.6 &   0.0422 & 0.0102 & 0.0406 & 0.0302 & 0.0078 & 0.0890 & 0.0414 \\
    & & 0.5 & 0.3 &  0.0446 & 0.0104 & 0.0372 & 0.0242 & 0.0074 & 0.1058 & 0.0306 \\
    & & 0.5 & 0.6 & 0.0452 & 0.0110 & 0.0414 & 0.0308 & 0.0040 & 0.1052 & 0.0342 \\ 
    \rule{0pt}{0.1em} & & & & & & & & & &  \\ 
    & 150 & 0.2 & 0.3 &  0.0490 & 0.0044 & 0.0416 & 0.0242 & 0.0000 & 0.1066 & 0.0460 \\
    & & 0.2 & 0.6 &   0.0480 & 0.0068 & 0.0422 & 0.0288 & 0.0000 & 0.1074 & 0.0408 \\
    & & 0.5 & 0.3 &  0.0482 & 0.0060 & 0.0424 & 0.0244 & 0.0000 & 0.1070 & 0.0458 \\
    & & 0.5 & 0.6 & 0.0434 & 0.0070 & 0.0404 & 0.0268 & 0.0000 & 0.1120 & 0.0414 \\ 
    \rule{0pt}{0.1em} & & & & & & & & & &  \\ 
    500 & 15 & 0.2 & 0.3 &  0.0540 & 0.0448 & 0.0510 & 0.0490 & 0.0374 & 0.0702 & 0.0512 \\
    & & 0.2 & 0.6 &   0.0516 & 0.0424 & 0.0478 & 0.0488 & 0.0368 & 0.0674 & 0.0444 \\
    & & 0.5 & 0.3 &  0.0474 & 0.0398 & 0.0452 & 0.0466 & 0.0294 & 0.0690 & 0.0488 \\
    & & 0.5 & 0.6 & 0.0490 & 0.0392 & 0.0466 & 0.0464 & 0.0320 & 0.0718 & 0.0446 \\ 
    \rule{0pt}{0.1em} & & & & & & & & & &  \\ 
    & 45 & 0.2 & 0.3 &  0.0554 & 0.0196 & 0.0480 & 0.0364 & 0.0198 & 0.0840 & 0.0340 \\
    & & 0.2 & 0.6 &   0.0496 & 0.0206 & 0.0456 & 0.0392 & 0.0202 & 0.0812 & 0.0378 \\
    & & 0.5 & 0.4 &  0.0552 & 0.0192 & 0.0514 & 0.0368 & 0.0166 & 0.0904 & 0.0330 \\
    & & 0.5 & 0.6 & 0.0518 & 0.0224 & 0.0472 & 0.0346 & 0.0188 & 0.0950 & 0.0328 \\ 
    \rule{0pt}{0.1em} & & & & & & & & & &  \\ 
    & 150 & 0.2 & 0.3 &  0.0476 & 0.0168 & 0.0456 & 0.0380 & 0.0044 & 0.0754 & 0.0432 \\
    & & 0.2 & 0.6 &   0.0456 & 0.0146 & 0.0428 & 0.0386 & 0.0036 & 0.0730 & 0.0426 \\
    & & 0.5 & 0.3 &  0.0540 & 0.0116 & 0.0486 & 0.0332 & 0.0052 & 0.0856 & 0.0380 \\
    & & 0.5 & 0.6 & 0.0456 & 0.0180 & 0.0436 & 0.0364 & 0.0036 & 0.0784 & 0.0418 \\ \hline
    \end{tabular}}
   \captionsetup{width = 0.95\linewidth, font = small}
   \caption{Simulated Size of Identification and Heteroskedasticity Robust Tests under Weak Identification. Each DGP is simulated 5000 times.}
   \label{table:weak_size}
\end{table}

I examine the size of the test under three different instrument regimes. In all three regimes, I begin with an initial set of instruments \(\bar z_i = (\bar z_{1i},\dots,\bar z_{15i})'\) generated independently across indices according to a multivariate Gaussian distribution with Toeplitz covariance structure, \(\Cov(\bar z_{\ell i},\bar z_{ki}) = 2^{-|\ell-k|}\). In the first regime, the instruments only include the initial set, \(z_i = \bar z_i\) so that \(d_z = 15\). In the second regime, the full set of instruments \(z_i\) additionally includes all quadratic and cubic terms, \((z_{\ell i}^2,z_{\ell i}^3)\), \(\ell = 1,\dots,15\) so that in total \(d_z = 45\). In the third regime, the full set of instrument includes the initial set of instruments, \(\bar z_i\),  all quadratic and cubic terms (30 additional terms) and interactions of the initial set of instruments (\(\binom{15}{2} = 105\) additional terms), so that in total \(d_z = 150\). Under each regime, the full set of instruments is passed to the test statistics with no indication about which instruments correspond to the initial set, and thus no indication about which instruments are relevant to the DGP.

\begin{figure}[ht!]
    \centering
    \includegraphics[width=0.9\linewidth]{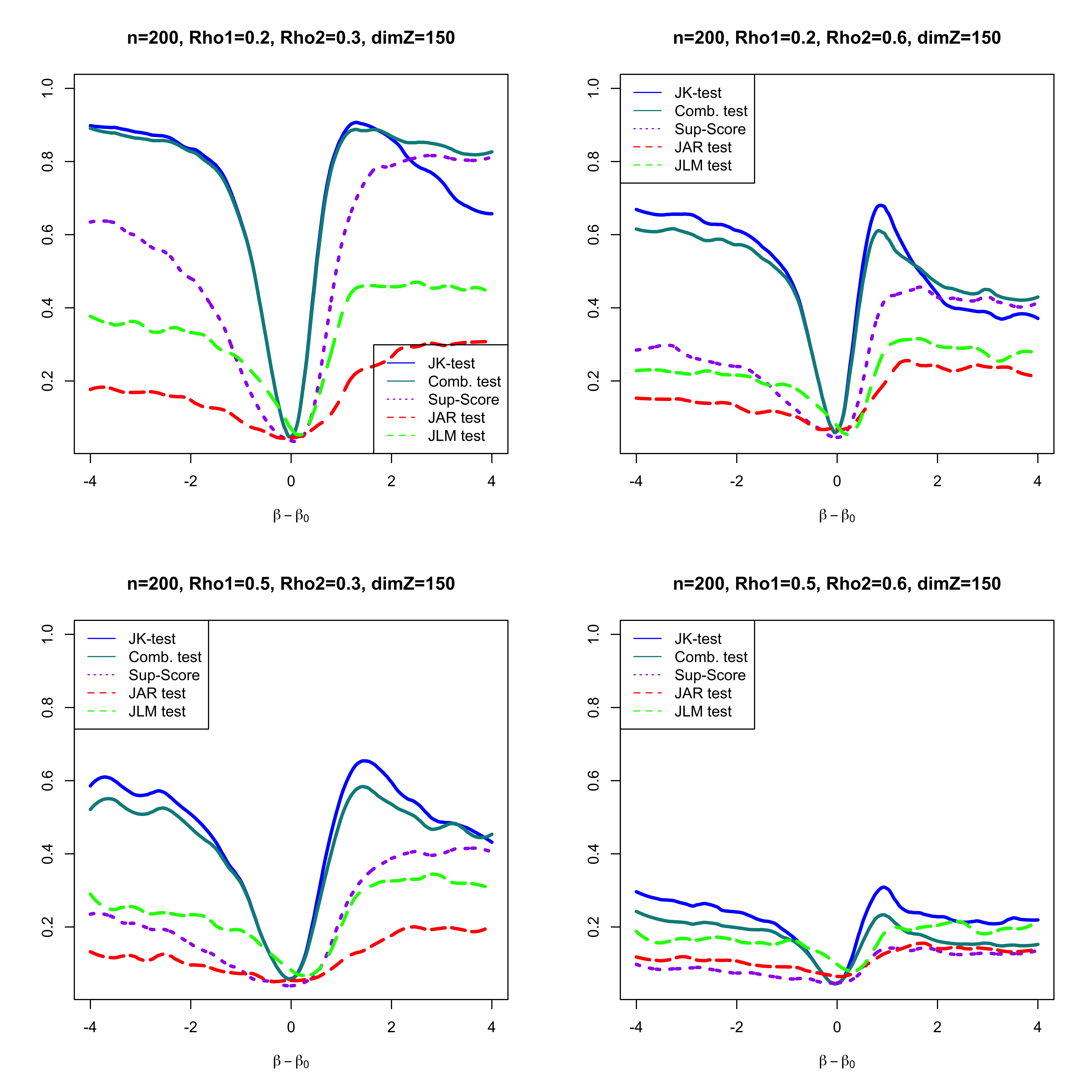}
    \captionsetup{font = small, width = 0.95\linewidth}
    \caption{Calibrated Local Power Curves under Intermediate Identification Strength and 150 Instruments. Sample size is 200 and rejection probability is calculated on a grid of 100 \((\beta_0 - \beta)\) points between -4 and 4. At each point the DGP is simulated 1000 times.}
    \label{fig:power65}
\end{figure}

In constructing the jackknife K-statistic, I opt to use the deleted diagonal ridge matrix, \(H = [h_{ij}]\) where \(h_{ij} = [\vz(\vz'\vz + \lambda I_{d_z})\vz']_{ij}\bm{1}\{i \neq j\}\), instead of a true jackknife ridge procedure in order to make the simulations compuationally tractable. Following reccomendations in \citet{lecture_notes_ridge}, the penalty paramter \(\lambda\) is set so that effective degrees of freedom of the resulting hat matrix is no more than \(n/5\).\footnote{Precisely, the penalty parameter is set \(\lambda = \max{0, (n/5)^{-1}d_1^2(\vz)}(d_z - n/5)\), where \(d_1^2(\vz)\) is the square of the maximum singular value of the design matrix \(\vz\).} To estimate the parameter \(\rho(z_i)\), I implement the default cross-validated \(\ell_1\)-penalized procedure of \eqref{eq:gamma-hat-equation} via the \texttt{glmnet} package in \texttt{R} \citep{glmnet}. I use the full vector of instruments as the basis to approximate \(\rho(z_i)\). 

I compare the simulated size of the jackknife K test and to the performance of the sup-score test, \(S(\beta_0)\), of \citet{BCCH-2012}, the thresholding test introduced in \Cref{subsec:combination-test}, the standard Anderson-Rubin (A.Rbn.) test of \citet{Anderson-Rubin-1949} and \citet{StaigerStock-1997}, the jackknife AR  test (JAR) of \citet{crudu_mellace_sandor_2021} and \citet{Mikusheva_Sun_2022}, and the jackknife LM test (JLM) of \citet{Matsushita_Otsu_2022}. Critical values of the sup-score and conditioning statistic are simulated with the procedures described in \Cref{sec:power-properties} with 1000 bootstrap replications. For the combination test cutoff, I report results using two different quantiles of the conditioning statistic under the assumption that \(\E[\widehat\Pi_i^I] = 0\) for all \(i \in [n]\); \(\tau_{0.3}\) corresponding to the 30\textsuperscript{th} quantile and \(\tau_{0.75}\) corresponding to the 75\textsuperscript{th} quantile. 

\Cref{table:weak_size} reports the simulated size for all tests under weak and strong identification, respectively. One can see that the \(\JK(\beta_0)\) statistic has nearly exact size in almost all the setups considered. In contrast, both the jackknife AR and jackknife LM test all exhibit moderate size distortions in various regimes. The jackknife AR test in particular appears to overreject in nearly all setups considered. This property also observed in the simulation study of \citet{Matsushita_Otsu_2022} and so may be driven by the similarity of this simulation setup to theirs. The jackknife LM statistic appears to have good size properties when \(d_z = 10\) and \(d_z = 150\) but is consistently (though only moderately) undersized in the intermediate regime where \(d_z = 45\). This size distortion does not improve when moving from \(n = 200\) to \(n = 500\), suggesting that the requirement that \(d_z \to \infty\) is important for the quality of finite-sample approximation by its limiting distribution. Though the good performance of the jackknife LM statistic when \(d_z = 10\) is notable, it should also be remarked that this is the setup with the least amount of correlation between the instruments.

The sup-score and Anderson Rubin test seem to be undersized in all regimes, with the Anderson-Rubin test nearly never rejecting when \(d_z = 150\). However, and in line with the theory, both of their size properties appear to improve when increasing the sample size from \(n = 200\) to \(n = 500\). It is possible that the size properties of the sup-score test could be improved by using an empirical bootstrap based approach to simulate the critical value, as proposed by \citet{deng2020beyond}, however I do not consider that approach here. The threshholding test appears to inherit the conservative nature of the sup-score test, though to a lesser degree due to the combination with the \(\JK(\beta_0)\) test.
\begin{figure}[ht!]
    \centering
    \includegraphics[width=0.9\linewidth]{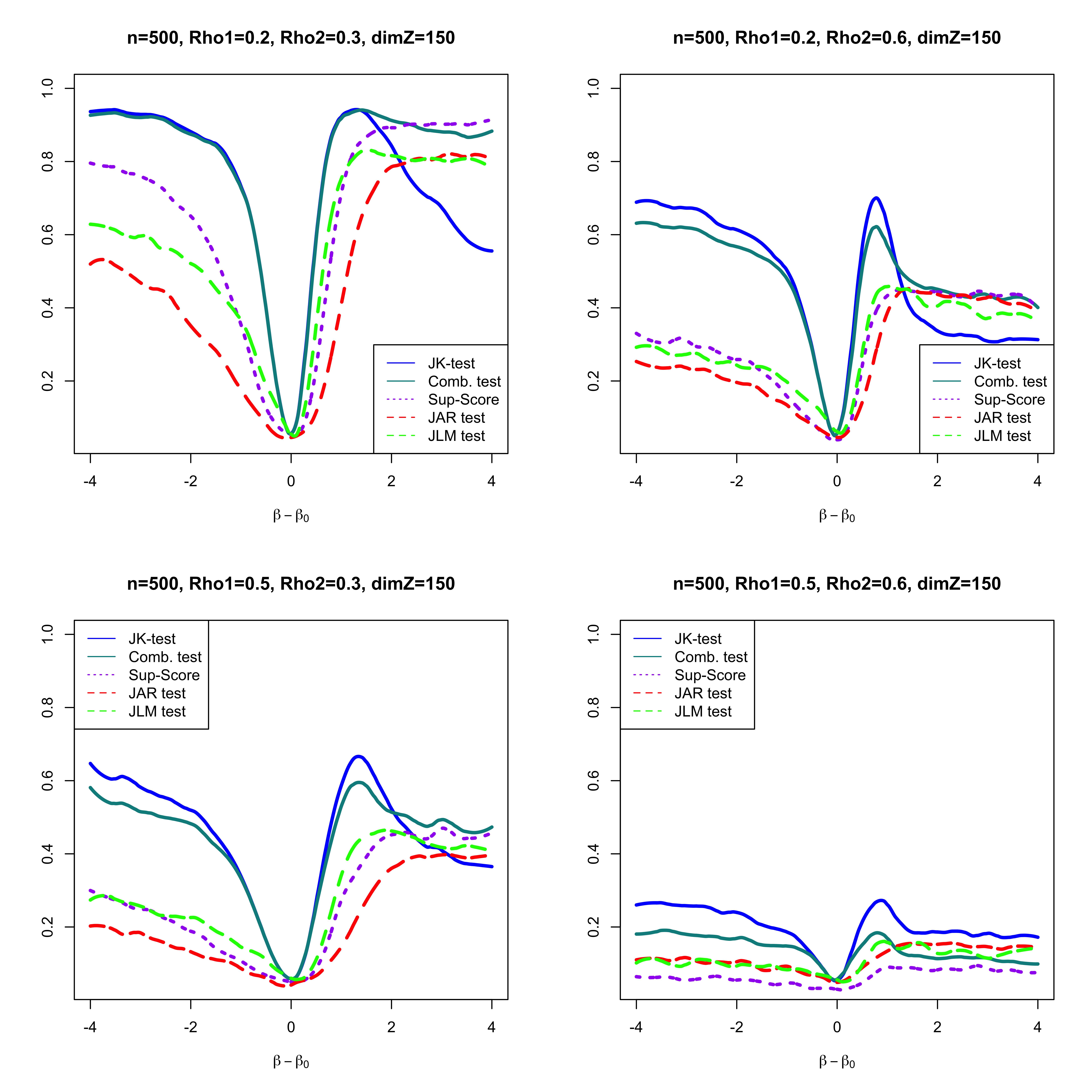}
    \captionsetup{font = small, width = 0.95\linewidth}
    \caption{Calibrated Local Power Curves under Intermediate identification Strength and 150 Instruments. Sample size is 500 and rejection power is calculated on a grid of 100 \((\beta_0 - \beta)\) points between -4 and 4. At each point the DGP is simulated 2000 times.}
    \label{fig:power75}
\end{figure}
In addition to examining the size properties of the given tests I also investigate the power properties of the tests in this setting.
\Cref{fig:power65,fig:power75} plot calibrated local power curves under an intermediate identification strength where the first stage is in a \(n^{-1/3}\) neighborhood of zero for \(n \in \{200, 500\}\), the number of instrument is plausible large, \(d_z = 150\), \(\varrho_1 \in \{0.2,0.5\}\) and \(\varrho_2 \in \{0.3,0.6\}\). The critical value of each test is set to simulated 95\textsuperscript{th} quantile of the distribution of the corresponding test-statistic under \(H_0\). I compare the calibrated local power curves of the \(\JK(\beta_0)\) test, the combination test with cutoff \(\tau_{0.75}\), the jackknife AR test, the Jackknife LM test, and the sup-score test.\footnote{For both the jackknife AR and jackknife LM tests, I use cross-fit estimates of test statistic variances proposed and shown to improve power by \citet{Mikusheva_Sun_2022}.}

In all setups considered, the jackknife K and combination tests have substantially stronger power than the jackknife AR, jackknife LM, and sup-score tests in local neighborhoods of the null as well as for negative values of \((\beta_0 - \beta)\). For values of \((\beta_0 - \beta)\) larger than 1.5, tests based on the jackknife K-statistic suffer from a loss of power as described in \Cref{sec:power-properties}. This power decline is largerly ameliorated by combining the jackknife K-statistic with the sup-score statistic and the thresholding test has good power properties over all alternatives considered. However, tests based on the jackknife AR or jackknife LM statistic can still provide better power than the threshholding test for very positive values of \((\beta_0 - \beta)\) in some setups.

\section[Single Endogeneous Variable]{Limiting Behavior of the Test Statistic}
\label{sec:single-endogeneous}

The limiting behavior of the test statistic is analyzed via a direct Gaussian approximation technique. In this section, I describe the approach and characterize the limiting behavior of the test statistic under local alternatives to \(H_0\). This direct approach has the advantage of not relying on any particular central limit theorem, which allows a great deal of flexibility in the choice of hat matrix \(H\) used to construct the first stage estimates. When there is only a single endogenous variable, \(d_x = 1\), the approach can be considerably simplified. The extension to \(d_x > 1\) requires a more involved argument which relies on stronger moment conditions.

Formally, I show that quantiles of the jackknife K-statistic can be approximated by analogous quantiles of the Gaussian statistic:
\begin{equation}
    \label{eq:gaussian-stat}
    \JK_G(\beta_0) \coloneqq \tilde\eps(\beta_0)\tilde\Pi(\tilde\Pi_\eps' \tilde\Pi_\eps)^{-1}\tilde\Pi'\tilde\eps(\beta_0);
\end{equation}
where, for each \(i \in [n]\), \((\tilde\eps_i(\beta_0),\tilde r_i')'\) are generated independently across indices following a Gaussian distribution with the same mean and covariance matrix as \((\eps_i(\beta_0),r_i')'\). Further, for \(\tilde\Pi_{\ell i} = \sum_{j\neq i} h_{ij} \tilde r_{\ell j}\) define \(\tilde\Pi_i \coloneqq (\tilde\Pi_{1i},\dots,\tilde\Pi_{d_xi})'\in\SR^{d_x}\), \(\tilde\Pi_{\eps i} \coloneqq (\E[\eps_i^2(\beta_0)])^{1/2}\tilde\Pi_i\),
\begin{align*}
    \tilde\eps(\beta_0) 
    &\coloneqq (\tilde\eps_1(\beta_0),\dots,\tilde\eps_n(\beta_0))' \in \SR^n,\\ 
    \tilde\Pi &\coloneqq (\tilde\Pi_1,\dots,\tilde\Pi_n)' \in \SR^{n\times d_x},\\ 
    \text{and }\;\tilde\Pi_\eps &\coloneqq (\tilde\Pi_{\eps 1},\dots,\tilde\Pi_{\eps n})' \in \SR^{n\times d_x}.
\end{align*}
As uncorrelated jointly Gaussian random variables are independent, under \(H_0\) the vector \(\tilde\eps(\beta_0)\) is mean zero and independent of \((\tilde\Pi, \tilde\Pi_\eps)\). Conditional on any realization of \((\tilde\Pi,\tilde\Pi_\eps)\) the \(\JK_G(\beta_0)\) statistic then follows a \(\chi^2_{d_x}\) distribution, and thus, its unconditional distribution is also \(\chi^2_{d_x}\).

To outline the approach, consider functions \(\varphi_\gamma(\cdot) \in C_b^3(\SR)\) that approximate the indicators \(\bm{1}\{\cdot \leq a\}\), where \(a \in \SR\) is arbitrary and as \(\gamma \to 0\) the quality of the approximation improves but the derivatives of \(\varphi_\gamma\) become larger in magnitutde. A primary goal is to show, for a sequence \(\gamma_n\) tending to zero, that
\begin{equation}
    \label{eq:multiple-goal}
    \E[\varphi_{\gamma_n}(\JK_I(\beta_0)) - \varphi_{\gamma_n}(\JK_G(\beta_0))] \to 0
\end{equation}
for a version of the test statistic, \(\JK_I(\beta_0)\), that could be constructed if \(\rho(\cdot)\) was known to the researcher.
In order to establish \eqref{eq:multiple-goal}, existing interpolation methods cannot be applied as they require the derivative of the test statistic with respect to individual obserations are bounded. In this setting, the derivative of the test statistics with respect to terms in the denominator matrix, \(\widehat\Pi_\eps'\widehat\Pi_\eps\), may be as large as the inverse of the minimum eigenvalue of the denominator matrix. When identification is sufficiently weak, the eigenvalues of the denominator matrix can be arbitrarily close to zero and thus inverse of its minimum eigenvalue may not even have finite moments. 

To get around this, I modify the argument by considering a ``data-dependent'' choice of approximation parameter \(\gamma_n\). This choice of approximation parameter inversely scales with the determinant of the denominator matrix and thus, since the determinant is the product of the eigenvalues, inversely scales with the minimum eigenvalue.\footnote{The determinant has the benefit of being a smooth function of elements of the matrix. This makes it nicer to work with than the minimum eigenvalue itself, which loses differentiability when the dimension of its eigenspace is larger than one.} Geometrically, this approach can be thought of as ``stretching out'' the function \(\varphi_{\gamma_n}(\cdot)\) in directions where the minimum eigenvalue of the denominator matrix is close to zero. Through the chain rule, this allows for control of the overall derivative of \(\varphi_{\gamma_n}(\JK_I(\beta_0))\) with respective to an
individual observation. 

\subsection{High Level Assumptions}

I now detail the assumptions needed for the argument, starting with the assumptions that are commmon to the cases with \(d_x = 1\) and \(d_x > 1\).
In what comes below \(c  > 1\) can be considered an arbitrary constant that may be updated upon each use but that does not depend on sample size \(n\). 

\begin{assumption}[Balanced Design]
    \label{assm:balanced-design}
    (i) Let \(s_{\ell, n}^{-2} = \max_{1 \leq i \leq n}\E[(\widehat\Pi_{\ell i}^I)^2]\) for each \(\ell \in [d_x]\); then, the minimum eigenvalue of the following matrix is bounded away from zero:
    \[
        c^{-1} \leq \lambda_{\min}\E\begin{pmatrix} \frac{s_{\ell,n}s_{k,n}}{n}\sum_{i=1}^n (\widehat\Pi_{\ell i}^I)(\widehat\Pi_{k i}^I) \end{pmatrix}_{\substack{1 \leq \ell \leq d_x \\ 1 \leq k \leq d_x}}
    \] 
    (ii) \(\max_i s_n \sum_{j\neq i} h_{ji}^2 \leq c\); and 
    (iii)  the following ratio is bounded away from zero: \(\frac{\sum_{k=2}^n\lambda_k^2(HH')}{\sum_{k=1}^n\lambda_k^2(HH')} \geq c^{-1}\) where \(\lambda_k(HH')\) represents the \(k^{\text{th}}\) largest eigenvalue of the matrix \(HH'\).
 \end{assumption}

\Cref{assm:balanced-design}(i) requires that the average second moment of the infeasible first-stage estimators be on the same order as the maximum first-stage estimator second moment. This is imposed mainly to rule out hat matrices that are all zeroes or nearly all zeros so that the effective number of observations used to test the null is growing with the sample size. \Cref{rem:verifying-balanced-design} provides further intuition for this assumption and below discusses how this assumption and \Cref{assm:balanced-design}(ii) may be verified in practice. \Cref{rem:balanced-design}  compares this balanced design assumption to that in the many-instruments literature \citep{crudu_mellace_sandor_2021,Mikusheva_Sun_2022,Matsushita_Otsu_2022,lim2022conditional}, noting that their balanced design neither implies nor is implied by the one in this paper.

\Cref{assm:balanced-design}(ii) requires that the maximum leverage of any observation be bounded. When \(H\) is symmetric, it is automatically satisfied.\footnote{To see this for \(d_x = 1\) notice that \(s_n^{-2} = \max_i\E[(\widehat\Pi_i^I)^2] \geq \max_i \Var(\widehat\Pi_i^I) = \max_i \sum_{j\neq i} h_{ij}^2 \Var(r_j)\), while \(\Var(r_j)\) will be assumed bounded from below by \(c^{-1}\). Inverting this chain of inequalities yields that \(s_n^2 \sum_{j\neq i} h_{ij}^2\) is bounded from above uniformly over all \(i \in [n]\).} \Cref{assm:balanced-design}(iii) can be viewed as a mild technical requirement that there be more than one ``effective'' instrument in the hat matrix.\footnote{In the case of a standard projection matrix (no deleted diagonal), \Cref{assm:balanced-design}(iii) would be satisfied whenever \(\rank(z(z'z)^{-1}z) > 1\), which occurs whenever there are at least two linearly independent instruments.} This condition can be easily verified in practice by examining the eigenvalues of \(HH'\).

Next, I make a high level assumption that the estimation error in \(\hat \rho(\cdot)\) can be treated as negligible in both the numerator and denominator. Later, I will verify this assumption for the particular choice of \(\hat \rho(\cdot)\) described in \Cref{sec:setup}. For each \(\ell \in [d_x]\) define \(\widehat\Pi^I_{\ell,i} \coloneqq \sum_{j\neq i} h_{ij}r_{\ell j}\), the version of the first stage estimates that could be constructed if \(\rho(\cdot)\) was known to the researcher. Using these, define the magnitude of estimation error in the numerator and denominator as
\begin{align*}
    \Delta_{N} &\coloneqq \max_{\ell \in [d_x]} \big|\frac{s_{\ell,n}}{\sqrt{n}}\sum_{i=1}^n \eps_i(\beta_0)\big(\widehat\Pi_{\ell,i} - \widehat\Pi_{\ell,i}^I\big) \big|\\ 
    \Delta_{D,\ell} &\coloneqq \max_{\ell \in [d_x]} \frac{s_{\ell,n}^2}{n}\sum_{i=1}^n \eps_i^2(\beta_0)\big(\widehat\Pi_{\ell,i}- \widehat\Pi_{\ell,i}\big)^2
\end{align*}

\begin{assumption}[Estimation Error]
    \label{assm:estimation-error}
    Estimation error in both the numerator and denominator of the test statistic can be treated as negligible, \((\Delta_N, \Delta_D)\to_p 0\). 
\end{assumption}
Showing that \((\Delta_N, \Delta_D)\to_p 0\) implies that estimation error can be treated as negligible for the test statistic \(\JK(\beta_0)\) requires some care. In a standard approach, this would straightforwardly follow from application of the continous mapping theorem. However, this approach requires that the scaled numerator and denominator each have well defined distributional limits, something that is not required by the direct Gaussian approximation. Instead, I establish and make use of anticoncentration bounds to show that the basic results of the continuous mapping theorem still apply even when the numerator and denominator do not have weak limits.

Finally, in addition to characterizing the limiting distribution of \(\JK(\beta_0)\) under \(H_0\), I also examine the behavior of \(\JK(\beta_0)\) in local neighborhoods of the null. These local neighborhoods are characterized by the local power index \(P\), defined below, as well as an additional regularity condition that restricts the size of \(\E[\eps_i(\beta_0)]\) relative to \(\E[r_{\ell i}]\).

\begin{assumption}[Local Identification]
    \label{assm:local-identification}
    (i) The local power index is bounded \(P \leq c\) for
    \[
        P = \sum_{\ell = 1}^{d_x} \E\bigg[\bigg(\frac{s_{\ell,n}}{\sqrt{n}}\sum_{i=1}^n \widehat\Pi_{\ell i}^I\Pi_i'(\beta - \beta_0)\bigg)^2\bigg]
    \]
    (ii) \(\E[(s_{n,\ell}\sum_{j\neq i} h_{ji} \eps_j(\beta_0))^2] \leq c\) for all \(\ell = 1,\dots,d_x\).
\end{assumption}
Under \(H_0\), \Cref{assm:local-identification} is trivially satisfied since \((\beta - \beta_0) = 0\) and \(\sum_{j\neq i} s_{\ell,n}^2h_{ji}^2 \leq c\) for each \(\ell \in [d_x]\). The local power index is the second moment of the scaled numerator of the test statistic is a measure of the association between the true first stage \(\Pi_i\) and the first-stage estimates \(\widehat\Pi_i\). In \Cref{sec:power-properties}, I discuss how the strength of this association is related to the power of the test under local alternatives. \Cref{assm:local-identification}(ii) can be roughly interpreted as requiring the local neighborhoods of \(H_0\) considered to be those in which the means of \((\eps_1(\beta_0),\dots,\eps_n(\beta_0))\) are of the same or lesser order than the means of \((r_1,\dots,r_n)\).

\subsection{Limiting Behavior of the Test Statistic}

Having detailed the assumptions needed for the argument, I now present results showing that, in local neighborhoods of \(H_0\), the distribution of the test statistic can be uniformly approximated by the distribution of the Gaussian statistic described in \eqref{eq:gaussian-stat}. As mentioned, the argument can be simplified to require lighter moment bounds when \(d_x = 1\). These moment bounds will be made on the model primitives, \(\eta_i \coloneqq (\beta - \beta_0)v_i + \eps_i = \eps_i(\beta_0) - \E[\eps_i(\beta_0)]\) and \(\zeta_{\ell i} \coloneqq v_{\ell i} - \rho_\ell(z_i)\eta_i = r_{\ell i} - \E[r_{\ell i}]\) for each \(\ell \in [d_x]\). 

\begin{theorem}[Single Endogenous Variable]
    \label{thm:feasible-local-power}
    Suppose that \Cref{assm:balanced-design,assm:estimation-error,assm:local-identification} hold. In addition suppose that (i) \(\{|\Pi_i| + |\rho(z_i)| + |(\beta - \beta_0)|\} \leq  c\) and (ii) for any \(r,s \in \SZ_+\) satisfying \(r + s \leq 6\) \(c^{-1} \leq \E[|\eta_i|^r|\zeta_i|^s] \leq c\). Then, for \(d_x = 1\),

    \[
        \sup_{a\in\SR} \big|\Pr(\JK(\beta_0) \leq a) - \Pr(\JK_G(\beta_0) \leq a)\big|\to 0.
    \]
    In particular, under \(H_0\), \(\JK(\beta_0) \rightsquigarrow \chi^2_1\).
\end{theorem}

In the case of \(d_x = 1\), I additionally show that the test based on the \(\JK_I(\beta_0)\) statistic is consistent whenever the power index diverges, \(P\to\infty\), and \Cref{assm:local-identification}(ii) holds. 
\begin{theorem}[Consistency]
    \label{thm:consistency}
    Suppose that Assumptions \ref{assm:balanced-design}, \ref{assm:estimation-error}, and~\ref{assm:local-identification}(ii) hold along with the additional conditions of \Cref{thm:feasible-local-power}. Then, if \(P \to \infty\) the test based on \(\JK(\beta_0)\) is consistent; i.e for any fixed \(a \in \SR\), \(\Pr(\JK(\beta_0) \leq a) \to 0\).
\end{theorem}
The dependence of the consistency result on \Cref{assm:local-identification}(ii) is a nontrivial restriction because of the bias taken on in constructing \(r_i\). In particular, against certain alternatives it is possible that \(\E[\widehat\Pi_i^I] = 0\) for all \(i \in [n]\) even under strong identification. This is an extreme case, however. In general, bias in \(\E[r_i]\) does not imply a violation of \Cref{assm:local-identification}(ii), which requires only that the \emph{size} of \(\E[r_i]\) be of a weakly greater order than that of \(\E[\eps_i(\beta_0)]\). Moreover, as discussed in \Cref{rem:consistency}, \Cref{thm:consistency} does not necessarily rule out consistency when \(P \to \infty\) but \Cref{assm:local-identification}(ii) fails.

Regardless, bias taken on in constructing \(r_i\) has consequences for the power of the test in finite samples.  The combination with the sup-score test described in \Cref{sec:setup} is an attempt to rectify this. While this attempt is not perfect, it appears to work well both in the empirical application to the data of \citet{gilchrist-glassberg-2016} and in the simulation study of \Cref{sec:simulations}.

The argument when \(d_x > 1\) is considerably more involved than the case where \(d_x = 1\) and requires strengthened moment condition on the variables \(\eta_i\) and \(\zeta_i\). Given a random variable \(X\) and \(\upsilon > 0\) the Orlicz (quasi-)norm is defined
\[
    \|X\|_{\psi_\upsilon} \coloneqq \inf\{t > 0: \E\exp(|X|^\upsilon/t^\upsilon) \leq 2\}
\]
Random variables with a finite Orlicz norm for some \(\upsilon \in (0,1] \cup\{2\}\) are termed \(\alpha\)-sub-exponential random variables \citep{Gotze-subexponential-polynomial,sambale2022notes}. This class of  encompasses a wide range of potential distributions including all bounded and sub-Gaussian random random variables (with \(\upsilon = 2\)), all sub-exponential random variables such as Poisson or noncentral \(\chi^2\) random variables (with \(\upsilon = 1\)), as well as random variables with ``fatter'' tails such as Weibull distributed random variables with shape parameter \(\upsilon \in (0,1]\). Thus, while imposing that the variables \(\eta_i\) and \(\zeta_i\) are \(\alpha\)-sub-exponential is notably stronger than the finite sixth moments required by \Cref{thm:feasible-local-power}, it may still be plausible in a wide range of empirical settings.

\begin{theorem}[Uniform Approximation]
    \label{thm:multiple-feasible-local-power}
    Suppose that \Cref{assm:balanced-design,assm:estimation-error,assm:local-identification} hold. In addition suppose that (i) \(c^{-1} \leq \lambda_{\min}(\E[\eta_i\eta_i']) \leq \lambda_{\max}(\E[\eta_i\eta_i']) \leq c\) and (ii) for some \(\nu \in (0,1]\cup\{2\}\) both \(\|\eta_i\|_{\psi_\nu} \leq c\) and \(\|\zeta_{\ell i}\|_{\psi_\nu} \leq c\). Then,
    \[
        \sup_{a\in\SR}|\Pr(\JK(\beta_0) \leq a) - \Pr(\JK_G(\beta_0) \leq a)|\to 0
    \]
    In particular, under \(H_0\), \(\JK(\beta_0)\rightsquigarrow \chi^2_{d_x}\).
\end{theorem}

While \(\JK_G(\beta_0)\) does not have a fixed distribution, examining its behavior is still tractable and allows for insight into the power properties of the jackknife K-test. In \Cref{sec:power-properties}, I use this result to analyze the local power of the proposed test. To improve power against certain alternatives, I suggest a combination with the sup-score statistic of \citet{BCCH-2012}.

\subsection{Controlling Estimation Error}
\label{subsec:estimation-error-single}

The final step is to show that estimation error in \(\hat \rho\) can be treated as negligible, that is verify \Cref{assm:estimation-error}. Establishing that this holds even when identification is weak makes use of the fact that estimation error in \(\hat \rho(\cdot)\) enters the test statistic only through its interaction with the implied error \(\eps_i(\beta_0)\). When identification is weak, the implied error is nearly conditional mean independent of the instruments and thus the product, \(\eps_i(\beta_0)\hat\rho(\cdot)\), is insensitive to small pertubations in \(\hat\rho(\cdot)\).\footnote{In the langauge of \citet{CCDDHNR-2018}, this is termed ``Neyman Orthogonality.'' Under strong identification, Neyman orthogonality allows for the use of many machine learning techniques to estimate the first stage. This (approximate) orthogonality also holds under structural parameter sequences that are local to the null.}

In theory, this orthogonality could be combined with cross-fitting to allow for the use of other machine learning methods to estimate \(\hat\rho(\cdot)\), as in \citet{CCDDHNR-2018}. This possibility is explored briefly in \Cref{sec:cross-fit}. The use of other machine learning methods may be useful if sparsity of \(\rho(\cdot)\) is not a plausible assumption in the researcher's empirical setting as estimators such as random forests and neural networks can be consistent in high dimensional settings under alternate assumptions.\footnote{\citet{chi2022asymptotic} provide results for random forests under a ``sufficient impurity decrease'' assumption while \citet{farrellNeuralNetwork2021} and \citet{Schmidt-Hieber-2020-neural-networks} provide results for neural networks under a generalized heirarchical model assumption.} However, I focus on the \(\ell_1\)-penalized procedure proposed in \Cref{sec:setup} both for expositional simplicity and because the sparsity assumption required for the consistency of this procedure mirrors that needed for the popular post-Lasso first-stage estimator. 

\begin{assumption}[Estimation Error]
    \label{assm:estimation-error2} 
    For each \(\ell \in [d_x]\)
    (i) there is a fixed constant \(\upsilon \in (0,1]\cup\{2\}\) such that \(\|\eta_i\|_{\psi_\upsilon} \leq c\);
    (ii) the basis terms \(b(z_i)\) are bounded, \(\|b(z_i)\|_\infty \leq  c\);
    (iii) the approximation error satisfies \((\E_n[\xi_{\ell i}^2])^{1/2} = o(n^{-1/2})\);
    (iv) the researcher has access to an estimator \(\widehat\phi\) of \(\phi\) satisfying \(\log(d_b n)^{2/(\upsilon\wedge 1)}\|\widehat\phi_\ell- \phi_\ell\|_1 \to_p 0\); 
    (v) the following moment bounds hold 
    \begin{enumerate}[(va)]
        \item \(\max_{1 \leq l \leq d_b} \big|\E\big[\frac{s_n}{\sqrt{n}}\sum_{i=1}^n \sum_{j\neq i}  h_{ij}\eps_i(\beta_0)b_l(z_j)\eps_j(\beta_0)\big]\big| \leq c \)
        \item  \(\max_{\substack{1 \leq i \leq n \\ 1 \leq l \leq d_b}} |\E[s_n\sum_{j\neq i} h_{ij}b_l(z_j)\eps_j(\beta_0)]| \leq c.\)
    \end{enumerate}
\end{assumption}

\Cref{assm:estimation-error}(i) ensures that \(\eta_i\) has finite exponential moments  (i.e has a well defined moment generating function), which is required to allow the number of basis terms used to approximate \(\rho(\cdot)\) to grow at a near exponential rate compared to the sample size (\(d_b \gg n\)). When using fewer basis terms this assumption may be relaxed. \Cref{assm:estimation-error}(ii) is a standard condition in \(\ell_1\)-penalized estimation. At the cost of extra notation, it can be relaxed and the sup-norm of the basis terms can be allowed to grow slowly with the sample size to accommodate bases such as normalized b-splines or wavelets. \Cref{assm:estimation-error}(iii) is a bound on the rate of decay of
the approximation error, similar to the approximate sparsity condition of \citet{BCCH-2012}.

\Cref{assm:estimation-error}(iv) is a high-level condition on the rate of consistency of the parameter estimate \(\hat\phi\) in the \(\ell_1\) norm. This can be verified under approximate sparsity for both the LASSO estimator in \eqref{eq:gamma-hat-equation} or post-LASSO procedures based on refitting an unpenalized version of \eqref{eq:gamma-hat-equation} only using the basis terms selected in a LASSO first stage. See \citet{BCCH-2012,vanDerGreer2016,Tan-2017}, and \citet{CS-2021} for references under various choices of penalty parameter. This condition allows for the dimensionality of the basis terms, \(d_b\), to grow near exponentially as a function of the sample size. Following the analysis of \citet{Tan-2017}  this may be satisfied as long as \(s^2\log^{2(\upsilon + 1)/\upsilon}(d_b n)/n\to 0\), where the sparsity index \(s\) denotes the number of nonzero elements of \(\phi\). 

\Cref{assm:estimation-error}(v) is a strengthening of the definition of local neighborhoods and can be interpreted similarly to \Cref{assm:local-identification}(ii). Since the moment conditions in \Cref{assm:estimation-error}(va,vb) hold with \(b_\ell(z_j)\eps_j(\beta_0)\) replaced with \(r_j\),  \Cref{assm:estimation-error}(v) can be interpreted as requiring that \(|\E[\sum_{j\neq i} h_{ij}b_\ell(z_j)\eps_j(\beta_0)]|\) is on the same order as \(|\E[\sum_{j\neq i} h_{ij}r_j]|\) for all \(i =1,\dots,n\) and \(\ell = 1,\dots,d_b\). As with \Cref{assm:local-identification}(ii), it is trivially satisfied under \(H_0\) or, using the fact that \(\max_i \sum_{j\neq i} s_n^2 h_{ij}^2 \leq c\), whenever \(\E[\eps_i(\beta_0)] = \Pi_i(\beta - \beta_0)\) is in a \(\sqrt{n}\)-neighborhood of zero. 

\begin{theorem}[Estimation Error]
    \label{thm:estimation-error}
    Suppose that \Cref{assm:balanced-design,assm:estimation-error2} hold. Then \((\Delta_N,\Delta_D) \to_p 0\).
\end{theorem}

\begin{remark}
    \label{rem:verifying-balanced-design}
    When \(d_x = 1\), a sufficient condition for \Cref{assm:balanced-design}(i) is that there is some fixed quantile \(q \in (0,100)\) such that \((cq)^{-1} \leq \frac{q^\text{\tiny th}\text{-quantile of }\E[(\widehat\Pi_i^I)^2]}{\max_i \E[(\widehat\Pi_i^I)^2]}\). In practice this can be verified by checking that there is some quantile \(q\) such that both
    \begin{equation}
        \label{eq:verifying-balanced-design}
        \frac{q^{\text{\tiny th}}\text{-quantile of }\sum_{j\neq i} h_{ij}^2}{\max_i \sum_{j\neq i} h_{ij}^2} \andbox \frac{q^{\text{\tiny th}}\text{-quantile of }(\sum_{j\neq i} h_{ij} \hat r_j)^2}{\max_i (\sum_{j\neq i} h_{ij} \hat r_j)^2 }     
    \end{equation}
    are bounded away from zero. Similarly, \Cref{assm:balanced-design}(ii) can be verified by checking that \(\max_i \sum_{j\neq i} h_{ji}^2 / \max_i \sum_{j\neq i} h_{ij}^2\) is bounded from above. The scaling factor \(s_n\) captures both the ``size'' of the elements in the hat matrix \(H\) and the strength of identification. If elements of the hat matrix are on the same order as a constant, one would expect \(s_n = O(n^{-1})\) under strong identification (\(\Pi_i\propto 1\)) while \(s_n = O(n^{-1/2})\) under weak identification (\(\Pi_i\lesssim n^{-1/2}\)).
\end{remark}

\begin{remark}[]
    \label{rem:balanced-design}
    The balanced-design condition in \Cref{assm:balanced-design}(i) is neither weaker nor stronger than that in the many instruments literature \citep{crudu_mellace_sandor_2021,Mikusheva_Sun_2022,Matsushita_Otsu_2022,lim2022conditional}. These papers require that the projection matrix \(P = \vz(\vz'\vz)^{-1}\vz'\) satisfies \([P]_{ii} \leq \delta \leq 1\) for some value \(\delta\) and all \(i \in [n]\). Since \(P\) is idempotent, \([P]_{ii} = 1\) for some \(i \in [n]\) implies that \([P]_{ij} = 0\) for \(j \neq i\).\footnote{Since \(P\) is idempotent, \([P]_{ii} = \sum_{i=1}^n [P]_{ij}^2 = [P]_{ii}^2 + \sum_{j\neq i} [P]_{ij}^2\).} This would not violate \Cref{assm:balanced-design} if one were to take \(H\) such that \(h_{ij} = [P]_{ij}\bm{1}\{i \neq j\}\); \(\E[(\widehat\Pi_i^I)^2] = 0\) is allowed for a constant share of \(i \in [n]\). Conversely, if the instruments are fixed or grow
    slowly, it is possible to construct a projection matrix \(P\)
    of rank \(d_z\) where \([P]_{ii}\) is bounded away from one for all \(i\in [n]\), but ``most'' of the rows are zero. I view this as a theoretical edge case, however, that seems unlikely to result from real data.
\end{remark}

\begin{remark}
    \label{rem:infeasible-error}
    The modified Lindeberg interpolation method allows me to give a nearly uniform explicit bound on the Gaussian approximation error in the case where \(d_x = 1\). In particular, I show that for any fixed value \(\Delta > 0\);
    \begin{align*}
        \sup_{a \leq \Delta} \big| \Pr(\JK_I(\beta_0) \leq a) - \Pr(\JK_G(\beta_0) \leq a)\big| \leq C n^{-2/13}
    \end{align*}
    where \(C\) is a constant that depends only on \((c,\Delta)\) and \(\JK_I(\beta_0)\) is the version of the test statistic that could be constructed if \(\rho(\cdot)\) was known to the researcher. While it does not account for estimation error in \(\hat\rho(\cdot)\), obtaining an explicit bound reflects an improvement over the original analyses of K-statistics in \citet{Kleibergen_2002,Kleibergen-2005}. These original studies rely on continuous mapping theorems to obtain the limiting chi-squared distributions, making the rate of decay of the approximation error difficult to analyze. 
    % The error bound of \(O(n^{-4/26})\) is comparable to the error bounds for maxima of high dimensional random vectors in \citet{cck2013,CCK-2018-HDCLT}.
\end{remark}

\begin{remark}[]
    \label{rem:jackknife}
    The interpolation argument relies on the fact that the first and second moments of \((\tilde\eps_i(\beta_0),\tilde r_i)\) are the same as the first and second moments of \((\eps_i(\beta_0),r_i)\) to match the first and moments of one-step deviations with Gaussian analogs. Without the jackknife form of \(\widehat\Pi_i^I\), these one step deviations would additionally contain cross-terms such as \(h_{ii}r_i\eps_i(\beta_0)\), for \(i \in [n]\). While the first moment of this cross-term is matched by the first moment of the Gaussian analog, \(h_{ii}\tilde\eps_i(\beta_0)\tilde r_i\), the second moment is not matched. This is manageable, however, so long as the terms \(h_{ii}\) are ``small.'' An example of when the \(h_{ii}\) terms are small is when \(H\) is taken to be the OLS projection matrix, \(H = \vz(\vz'\vz)^{-1}\vz\), and the number of instruments satisfies \(d_z^3/n \to 0\). See \Cref{sec:main-proofs,sec:exogenous-controls} for details.
\end{remark}

\begin{remark}[]
    \label{rem:consistency}
    \Cref{thm:consistency} does not necessarily rule out that a test based on \(\JK_I(\beta_0)\) is consistent when \(P\to\infty\) but \Cref{assm:local-identification}(ii) fails to hold. There is reason to believe that this issue can be overcome, \citet{AMS-2004} show that the K-statistic of \citet{Kleibergen_2002} is consistent against fixed alternatives under strong identification. However, a full consistency result is not pursued here and left to future work. 
    % I conjecture that this result would require that the procedure used to construct the first stage estimates, \(\widehat\Pi_i\), yields consistent estimators. In that case, these estimators \(\widehat\Pi_i\) could ``partial-out'' bias from the term \(\rho(z_i)\eps_i(\beta_0)\) under strong identification. 
\end{remark}

\begin{remark}[]
\label{rem:approximate-sparsity}
    Approximate sparsity of \(\rho(z_i)\) may be a particularly palatable
    assumption in cases where the instrument set is generated by functions of a smaller initial set of instruments, as in \citet{AK-1991,Paravisini-2014,gilchrist-glassberg-2016}, and \citet{derenoncourt-2022}. In these cases, the dimensionality of the basis, \(d_b\), may not need to be much larger than the dimensionality of the instruments, \(d_z\), to provide a good approximation of \(\rho(z_i)\). 
    Interestingly, if taking \(b(z_i) = z_i\) provides a good approximation of \(\rho(z_i)\), then consistency of \(\hat\phi\) in \(\ell_1\)-norm is achievable under \(d_z^2/n \to 0\) even if \(\phi\) is fully dense. This requirement is substantially weaker than the \(d_z^3/n\to 0\) requirement of the standard K-statistic.
\end{remark}

\section[Power Propreties and Improvements]{Local Power and Improvements}
\label{sec:power-properties}

Using the characterization of the limiting behavior of the test statistic derived in \Cref{sec:single-endogeneous}, I analyze the local power properties of the test. Unfortunately, against certain alternatives the test statistic may have trivial power, a deficiency shared with the K-statistics of \citet{Kleibergen_2002,Kleibergen-2005}. Finally, I describe how the the combination with the sup-score statistic, described in \Cref{sec:setup}, attempts to remedy this and formally establish its validity. 

\subsection{Local Power Properties}
\label{subsec:power-analysis}

For exposition, I focus on the case where \(d_x = 1\).
In local neighborhoods of \(H_0\), as defined in \Cref{assm:local-identification,assm:estimation-error}, \Cref{thm:feasible-local-power} implies that the limiting behavior of \(\JK(\beta_0)\) can be analyzed by examining the behavior of the Gaussian analog statistic, \(\JK_G(\beta_0)\). Conditional on the vector \(\tilde r = (\tilde r_1,\dots,\tilde r_n)\), the distribution of \(\JK_G(\beta_0)\) is nearly non-central \(\chi^2_1\) with noncentrality parameter \(\mu(\tilde r)\), \(\JK_G(\beta_0) |\tilde r \sim A^2(\tilde r) \cdot \chi^2_1(\mu(\tilde r))\):
\begin{align*}
    A(\tilde r) &= \frac{\sum_{i=1}^n \Var(\eta_i)\tilde\Pi_i^2}{\sum_{i=1}^n \{\Pi_i^2(\beta - \beta_0)^2 + \Var(\eta_i)\}\tilde\Pi_i^2 }\\ 
    \mu^2(\tilde r) &=  (\beta - \beta_0)^2 \frac{\big(\sum_{i=1}^n \Pi_i\tilde\Pi_i\big)^2}{\sum_{i=1}^n \{\Pi_i^2(\beta - \beta_0)^2 + \Var(\eta_i)\}\tilde\Pi_i^2}.
\end{align*}
Under local alternatives, the terms \(\Pi_i^2(\beta - \beta_0)^2 \to 0\) so that \(A(\tilde r) \to 1\) and \(|\mu^2(\tilde r) - \mu_\infty^2(\tilde r)|\to 0\), where
\begin{equation}
    \label{eq:limiting-local-power}
    \mu^2_\infty(\tilde r) = (\beta - \beta_0)^2\frac{\big(\sum_{i=1}^n \Pi_i \tilde \Pi_i)^2}{\sum_{i=1}^n \Var(\eta_i)\tilde\Pi_i^2}.
\end{equation}
The numerator of \(\mu^2_\infty(\tilde r)\) suggests that power is maximized when the first-stage estimate \(\tilde \Pi_i\) is close to the true first stage value \(\Pi_i\). Indeed, when errors are homoskedastic \(\mu_\infty^2(\tilde r)\) is maximized by setting \(\tilde \Pi_i = \Pi_i\) reflecting the classical result of \citet{Chamberlain-1987}. The denominator of \(\mu^2_\infty(\tilde r)\) suggests that having first-stage estimates \(\tilde \Pi_i\) with low second moments may increase power. This guides the recommendation for the use of \(\ell_2\)-regularization in constructing the hat matrix, \(H\).

Unfortunately, estimators of \(\Pi_i\) based on \(r_i = x_i - \rho(z_i)\eps_i(\beta_0)\) may not be close to \(\Pi_i\) under \(H_1\). This is because the mean of \(r_i\) will in general differ from \(\Pi_i\)
\[
    \E[r_i] = \Pi_i - \rho(z_i)\Pi_i(\beta - \beta_0)
\]
This deficiency is inherited from the similarity of the \(\JK(\beta_0)\) statistic to the K-statistic. 
As pointed out by \citet{Moreira-2001}, this need not be an issue as long as there is a fixed constant \(C\neq 0\) such that \(\E[r_i] = C\Pi_i\) for all \(i \in [n]\). However, in general, this will introduce bias into the first-stage estimates \(\widehat\Pi_i\) under \(H_1\). The power implications of this bias are particularly pronounced when \(\rho(z_i)\) is a constant \((\beta - \beta_0) = 1/\rho(z_i)\). In this case, \(\E[r_i]\), and thus \(\E[\tilde\Pi_i]\), will equal zero for each \(i\in[n]\), and the \(\JK(\beta_0)\) statistic will select a direction completely at random to direct power into.\footnote{\citet{AMS-2006} and \citet{Andrews-2016} point out this deficiency in the context of the K-statistics of \citet{Kleibergen_2002,Kleibergen-2005}.}

\subsection{A Simple Combination Test}
\label{subsec:combination-test}

To combat this loss of power for tests based on the K-statistic,  a common strategy is to combine the K-statistic with the Anderson-Rubin statistic based on a conditioning statistic. While the Anderson-Rubin statistic does not have optimal power on its own, it has the benefit of directing power equally in all directions avoiding the pitfalls of the K-statistic which lacks power in certain directions. Prominent examples of such tests are the conditional likelihood ratio test of \citet{Moreira_2003}, the GMM-M test of \citet{Kleibergen-2005}, and the minimax regret tests of \citet{Andrews-2016}. These combinations make use of the fact that the Anderson-Rubin statistic is asymptotically independent of both the K-statistic and the conditioning statistic.

Unfortunately, the asymptotic validity of these tests under heteroskedasticity is based on the assumption that \(d_z^3/n\to 0\), which may not reasonably describe many settings discussed above. Instead, to improve the power of tests based on the jackknife K-statistic, I consider a simple combination with the sup-score statistic of \citet{BCCH-2012}. The test based on the sup-score statistic \eqref{eq:sup-score-statistic} is similar in spirit to the Anderson-Rubin test but controls size even when \(d_z\) grows near exponentially as a function of the sample size.   
\begin{equation}
    \label{eq:sup-score-statistic}
    S(\beta_0) \coloneqq \sup_{1 \leq \ell \leq d_z} \bigg|\frac{\sum_{i=1}^n \eps_i(\beta_0)z_{\ell i}}{(\sum_{i=1}^n z_{\ell i}^2)^{1/2}}\bigg|
\end{equation}
A size \(\theta \in (0,1)\) test based on the sup-score statistic rejects whenever \(S(\beta_0) > c^S_{1-\theta}\) where, for \(e_1,\dots,e_n\) iid standard normal and generated independently of the data, \(c^S_{1-\theta}\) is the simulated multiplier bootstrap critical value:
\[
    c^S_{1-\theta} \coloneqq (1-\theta)\text{ quantile of } \sup_{1 \leq \ell \leq d_z} \bigg|\frac{\sum_{i=1}^n e_i \eps_i(\beta_0)z_{\ell i}}{(\sum_{i=1}^n z_{\ell i}^2)^{1/2}}\bigg|\text{ conditional on }\{(y_i,x_i,z_i)\}_{i=1}^n.
\]
As with the Anderson-Rubin test, tests based on the sup-score statistic may have suboptimal power properties in overidentified models as it does not incorporate first-stage information. However, the sup-score statistic does retain the benefit of directing power evenly in all directions, avoiding pitfalls of tests based on \(\JK(\beta_0)\) against certain alternatives.

The combination test will be based on an attempt to detect whether the alternative \(\beta\) is such that \(\E[\widehat\Pi_{\ell,i}^I] = 0\) for all \(i = 1,\dots, n\) and some \(\ell \in [d_x]\). When this is the case, the researcher would prefer to test the null hypothesis using the sup-score statistic. As mentioned in \Cref{sec:setup}, detection of whether \(\E[\widehat\Pi_{\ell,i}^I] = 0\) for some \(i \in [n]\) is based on the conditioning statistic:
\begin{equation}
    \label{eq:conditioning-statistic2}
    \begin{split}
        C = \inf_{\ell \in [d_x]}\sup_{i \in [n]} \bigg|\frac{\sum_{j\neq i} h_{ij}\hat r_j}{(\sum_{j\neq i} h_{ij}^2)^{1/2}}\bigg|.
    \end{split}
\end{equation}
Under the assumption that \(\E[\widehat\Pi_i^I] = 0\) for all \(i \in [n]\), quantiles of the conditioning statistic can be simulated analogously to the sup-score critical value.
For a new set of \(e_1,\dots,e_n\) iid standard normal and generated independently of the data, and for any \(\theta\in (0,1)\), define the conditional quantile 
\begin{equation}
    \label{eq:pre-test-quantile}
    c^C_{1-\theta} \coloneqq (1-\theta)\text{ quantile of }\inf_{\ell \in [d_x]}\sup_{i \in [n]} \bigg|\frac{\sum_{j\neq i}e_i h_{ij}\hat r_j}{(\sum_{j\neq i} h_{ij}^2)^{1/2}}\bigg| \text{ conditional on }\{(y_i,x_i,z_i)\}_{i=1}^n
\end{equation}

Depending on the value of the conditioning statistic, the thresholding test decides whether the test based on \(\JK(\beta_0)\) or one based on \(S(\beta_0)\) should be run.
\begin{equation}
    \label{eq:thresholding-statistic}
    \begin{split}
        T(\beta_0; \tau) = 
        \begin{cases}
            \bm{1}\{\JK(\beta_0) > \chi^2_{1;1-\alpha}\} &\text{if } C \geq \tau  \\
            \bm{1}\{S(\beta_0) > \;c_{1-\alpha}^S\} &\text{if }C < \tau\\ 
        \end{cases}
    \end{split}
\end{equation}
for some cutoff \(\tau\), which I take in the simulation study and empirical exercise to be the 75\textsuperscript{th} quantile of the distribution of \(C\) under the assumption that \(\E[\widehat\Pi_i^I] = 0, \forall i \in [n]\).

To show that the thresholding test controls size, I compare the rejection probability to that of a Gaussian analog.  In addition to \(\JK_G(\beta_0)\), defined in \Cref{sec:single-endogeneous}, define the Gaussian analogs of \(\S(\beta_0)\) and the conditioning statistic \(C\):
\begin{align*}
    S_G(\beta_0) &\coloneqq \sup_{\ell \in [d_z]} \bigg|\frac{\sum_{i=1}^n \tilde\eps_i(\beta_0)z_{\ell i}}{(\sum_{i=1}^n z_{\ell i}^2)^{1/2}} \bigg|
               &
    C_G &\coloneqq \inf_{\ell\in[d_x]}\sup_{ i \in [n]} \bigg|\frac{\sum_{j\neq i} h_{ij} \tilde r_j}{(\sum_{j\neq i} h_{ij}^2)^{1/2}}\bigg|
\end{align*}
where, as in \Cref{sec:single-endogeneous}, \((\tilde\eps_i(\beta_0),\tilde r_i)'\) are generated independently of each other and the data following a Gaussian distribution with the same mean and covariance matrix as \((\eps_i(\beta_0),r_i)\). Since \(\Cov(\tilde\eps_i(\beta_0),\tilde r_i) = 0\) under \(H_0\), the statistics \(C_G\) and \(S_G(\beta_0)\) are independent under the null. Similarly, the null distribution of \(\JK_G(\beta_0)\) is the same conditional on any realization of \((\tilde r_1,\dots,r_n)\); it is also independent of \(C_G\) under the null. The Gaussian analog thresholding test decides whether the researcher should run a test based on \(S_G(\beta_0)\) or \(\JK_G(\beta_0)\) depending on the value of \(C_G\) as in \eqref{eq:thresholding-statistic}.

The test statistics \(\JK_G(\beta_0)\) and \(S_G(\beta_0)\) are only marginally independent of the conditioning statistic \(C_G\) under the null. This limits the ways in which the test statistics can be combined using the conditioning statistic while still controlling size. This marginal independence in the Gaussian limit is enough, however, for the asymptotic validity of the thresholding test, \(T(\beta_0;\tau)\). To establish that the behavior of the pairs \((C,\JK(\beta_0))\) and \((C,S(\beta_0))\) can be approximated by the behavior of \((C_G,\JK_G(\beta_0))\) and \((C_G,S_G(\beta_0))\), respectively, I rely on the following assumption: 
\begin{assumption}[Combination Conditions]
    \label{assm:comb-conditions}
    Assume that for each \(\ell \in [d_x]\) (i) there is a \(\upsilon \in (0,1]\cup\{2\}\) such that \(\|\zeta_{\ell i}\|_{\psi_\upsilon} \leq c\); (ii) \(\max_{i,j} |\frac{h_{ij}}{(\E_n[h_{ij}^2])^{1/2}}| + \max_{l,i}| \frac{z_{li}}{(\E_n[z_{li}^2])^{1/2}} | \leq  c\); and (iii) \(\log^{7 + 4/\upsilon}(d_zn)/n\to 0\).
\end{assumption}

\Cref{assm:comb-conditions}(i) is a strengthening of the moment bound on \(r_i\) similar to that of \Cref{assm:estimation-error}(i). As discussed, while more restrictive than the condition in \Cref{thm:feasible-local-power}, this still allows for a wide range of potential distributions for \(r_i\). \Cref{assm:comb-conditions}(ii) requires that the number of observations used to test \(\E[\widehat\Pi_i] = 0\) via the conditioning statistic and the number of observations used to test the null hypothesis via the sup-score test are both growing with the sample size. It can be verified by looking at the hat matrix \(H\) and the instruments. Finally, \Cref{assm:comb-conditions}(iii) is a light requirement on the number of instruments \(d_z\) needed for the validity of the sup-score test. It allows the number of instruments to grow near exponentially as a function of sample size.

% \begin{theorem}[]
%     \label{thm:joint-combination}
%     Suppose \Cref{assm:moment-assumptions,assm:balanced-design,assm:local-identification,assm:estimation-error,assm:comb-conditions} hold. Then:
%     \begin{align*}
%          &\sup_{(a_1,a_2)\in \SR^2} \big|\Pr(\JK(\beta_0) \leq a_1, C \leq a_2) - \Pr(\JK_G(\beta_0) \leq a_1, C_G \leq a_2)\big| \to 0\\
%         \andbox
%          &\sup_{(a_1,a_2)\in \SR^2} \big|\Pr(\S(\beta_0) \leq a_1, C \leq a_2) - \Pr(\S_G(\beta_0) \leq a_1, C_G \leq a_2)\big| \to 0
%     \end{align*}
%     In particular, since \((\JK_G(\beta_0) \perp C_G)\)  and \((\S_G(\beta_0)\perp C_G)\) under \(H_0\) the test based on \(T(\beta_0;\tau)\) has asymptotic size \(\alpha\) for any choice of cutoff \(\tau\).
% \end{theorem}
\begin{theorem}[]
    \label{thm:joint-combination}
    Suppose the conditions of \Cref{thm:multiple-feasible-local-power} and \Cref{assm:comb-conditions} hold. Then,
    \begin{enumerate}
        \item the test based on \(T(\beta_0;\tau)\) has asymptotic size \(\alpha\) for any choice of cutoff \(\tau\), and 
        \item if \(\E[\widehat\Pi_i^I] = 0\) for all \(i \in [n]\), there exist sequences \(\delta_n \searrow 0\) and \(\beta_n \searrow 0\) such that with probability at least \(1 - \delta_n\),
        \[
            \sup_{\theta\in (0,1)}\big|\Pr_e(C \leq c^C_{1-\theta}) - (1-\theta)\big| \leq \beta_n,
        \] 
        where \(\Pr_e(\cdot)\) denotes the probability with respect to only the variables \(e_1,\dots,e_n\).
    \end{enumerate}
 \end{theorem}

 The first part of \Cref{thm:joint-combination} establishes the asymptotic validity of the thresholding test \(T(\beta_0;\tau)\) for any choice of cutoff \(\tau\). While not explicitly stated in the statment of the theorem, this result is uniform in the choice of \(\tau\); for any sequence \(\{\tau_n\}\subset \SR_+\) the sequnce of testing procedures \(T(\beta_0;\tau_n)\) will also have asymptotic size \(\alpha\). The proof of this statement follows the logic outlined above. The second part of \Cref{thm:joint-combination} establishes the validity of the multiplier bootstrap procedure to approximate quantiles of the conditioning statistic. It follows directly from results in \cite{belloni2018highdimensional} after verifying that the conditions needed for error taken on from estimation of \(\rho(z_i)\) can treated as negligible under \Cref{assm:estimation-error}.

 In the case of a single endogenous variable, \(d_x = 1\), \Cref{thm:joint-combination} could be established under the lighter conditions of \Cref{thm:feasible-local-power} along with \Cref{assm:comb-conditions}. However, for brevity, I do not seperate the two cases here.

\begin{remark}[]
    \label{rem:jlm-comparasion} 
    It is useful to compare the \(\JK(\beta_0)\) statistic to the JLM statistic of \citet{Matsushita_Otsu_2022}, which also converges to a limiting \(\chi^2\) distribution when \(d_z \to \infty\). In the case where \(d_x = 1\) the JLM statistic can be expressed
    \begin{equation}
        \label{eq:jlm-stat}
        \text{JLM}(\beta_0) \coloneqq \bigg(\frac{\sum_{i=1}^n \eps_i(\beta_0) \sum_{j\neq i}P_{ij}x_j}{\sum_{i=1}^n\eps_i^2(\beta_0)\big(\sum_{j\neq i} P_{ij} x_j\big)^2 + \sum_{i=1}^n\sum_{j\neq i}P_{ij}^2\eps_i(\beta_0)\eps_j(\beta_0)x_ix_j}\bigg)^2
    \end{equation}
    where \(P = \vz(\vz'\vz)^{-1}\vz'\) is the standard OLS projection matrix. This expression looks simillar to that of the \(\JK(\beta_0)\) statistic with first stage estimates \(\widehat\Pi_i = \sum_{j\neq i} P_{ij}x_j\). 
    From this expression, one can posit two potential reasons for the increased power of tests based on the \(\JK(\beta_0)\) statistic seen in the empirical applications in \Cref{sec:empirical} and the simulation study of \Cref{sec:simulations}. 

    The first is that, when the bias of \(r_i\) is not too adverse, first stage estimates based on the ``true'' jackknife ridge or jackknife OLS described in \eqref{eq:ridge-hat-matrix} may be closer to the true first stage, \(\Pi_i\) than those based on the deleted-diagonal projection matrix. As seen in \Cref{subsec:power-analysis}, higher quality first stage estimates can improve the power of the test by increasing the correlation between these estimates and \(\eps_i(\beta_0)\) under \(H_1\). This loss in quality of first-stage estimates based on the deleted-diagonal projection matrix may be negligible when the diagonal elements, \(P_{ii}\), are small in which case the deleted diagonal estimates closely resemble standard OLS estimates. However, when either the number of instruments is large relative to the sample size or the instruments are highly correlated the diagonal elements \(P_{ii}\) will be large in which case estimates of
    \(\Pi_i\) based on the deleted-diagonal projection matrix may not be accurate. This pattern can be seen in both empirical applications in \Cref{sec:empirical}; in both the data of \citet{gilchrist-glassberg-2016} and \citet{AK-1991} the improvments in power from using the \(\JK(\beta_0)\) statistic become more pronounced as the number of instruments increases.
    % This in turn numerically increases the numerator of the test statistic.

    A second potential reason for improved power is that the \(\JK(\beta_0)\) statistic uses individual scores, \(\eps_i(\beta_0)\widehat\Pi_i\), that are uncorrelated with each other under \(H_0\). That is, for \(j \neq i\), \(\E[\eps_i(\beta_0)\widehat\Pi_i\eps_j(\beta_0)\widehat\Pi_j] = 0\). Thus, the second term in the denominator of \eqref{eq:jlm-stat}, which accounts for the covariance between individual scores in the numerator of the JLM statistic, does not appear in the expression of \(\JK(\beta_0)\). Note that this second term has a positive expectation under both positive selection, \(\E[\eps_i(\beta_0)x_i] > 0\) for all \(i \in [n]\), and negative selection, \(\E[\eps_i(\beta_0)x_i] < 0\) for all \(i \in [n]\). If this term is large, it can substantially increase the denominator of the JLM statistic relative to that of the \(\JK(\beta_0)\) statistic, reducing its power. This again may be likely when the diagonal elements, \(P_{ii}\), are large due to idempotency of the
    projection matrix: \(P_{ii} = \sum_{j=1}^n P_{ij}^2\). Moreover, by construction  \(\Var(r_i) \leq \Var(x_i)\), so a large second term in the denominator of the JLM statistic may not be offset by a smaller first term, at least in local regions of \(H_0\).

    In sum, when bias taken on in constructing \(r_i\) is not too adverse, the \(\JK(\beta_0)\) statistic may have a larger numerator than the JLM statistic due to the use of higher quality first-stage estimates and a smaller denominator due to the use of uncorrelated individual scores.  Since both the \(\JK(\beta_0)\) and JLM statistics are compared to the same \(\chi^2\) quantile, both of these properties may lead to more likely rejection of tests based on the \(\JK(\beta_0)\) statistic under \(H_1\). \citet{Matsushita_Otsu_2022} note that the power properties of the JLM statistic are similar to those of the JAR statistic of \citet{Mikusheva_Sun_2022}, suggesting that the improvments in power compared to the JAR test seen in \Cref{sec:empirical} and \Cref{sec:simulations} may be explained similarly.
\end{remark}

\begin{remark}[]
    \label{rem:heteroskedasticity}
    As mentioned by \citet{Andrews-2016} in the context of the standard K-statistic, the attempt to rectify the power deficiency via this particular conditioning statistic is not perfect. In particular, under heteroskedasticity, the means of the partialed-out endogenous variables, \(\E[r_i]\), may not be scaled versions of the true first stages. However, as long as \(\E[r_i] \neq 0\), one can still expect \(\E[\widehat\Pi_i^I] = \sum_{j\neq i} h_{ij}\Pi_i + (\beta - \beta_0)\sum_{j\neq i} h_{ij}\rho(z_i)\Pi_i\) to be related to the true fist stage \(\Pi_i\) and for the test to have nontrivial power. Moreover, in light of the dependence of the consistency result in \Cref{thm:consistency} on \Cref{assm:local-identification}(ii), in the case where \(\E[\widehat\Pi_i] = 0\) for all \(i\in[n]\) it may be particularly important to avoid using the jackknife K-statistic to test \(H_0\).
\end{remark}

\section{Conclusion}
\label{sec:conclusion}

I propose a new test for the structural parameter in a linear instrumental variables model. This test is based on a jackknife version of the K-statistic and the limiting behavior of the test is analyzed via a novel direct Gaussian approximation argument. I show that, as long as an auxiliary parameter can be consistently estimated,  the test is robust to both the strength of identification and the number of instruments; the limiting distribution of the test statistic does not depend on either of these factors. Consistency of the auxiliary parameter can be achieved under approximate sparsity using simple-to-implement \(\ell_1\)-penalized methods.

I characterize the behavior of the jackknife K-statistic in local neighborhoods of the null. To address a power deficiency that tests based on jackknife K-statistic inherit from their non-jackknife namesakes, I propose a testing procedure that decides whether the researcher should run a test via the jackknife K-statistic or one via the sup-score statistic based on the value of a conditioning statistic. While this combination may not fully address the power decline, I show that it works well in a simulation study and leave further refinements to future work.

% References
% \newblankpage
\bibliography{bibtex/mlci}

% Appendix
\startappendix
\section{Proof of \Cref{thm:feasible-local-power}}
\label{sec:main-proofs}

\Cref{thm:feasible-local-power} follows from the following two main technical lemmas, the proofs of which will comprise the majority of this appendix section. Let \(\JK_I(\beta_0)\) be the version of the test statistic that could be constructed if \(\rho(\cdot)\) was known to the researcher, defined in more detail shortly.

\begin{lemma}[Infeasible Uniform Approximation]
    \label{thm:local-power}
    Suppose that \Cref{assm:balanced-design,assm:local-identification} hold as well as the moment bounds of \Cref{thm:feasible-local-power}. Then,
    \[
        \sup_{a \in \SR}\big|\Pr(\JK_I(\beta_0) \leq a) - \Pr(\JK_G(\beta_0) \leq a)\big| \to 0
    \]
\end{lemma}

\begin{lemma}[Estimation Error]
    \label{lemma:estimation-error-high-level}
    Suppose that \Cref{assm:balanced-design} and \Cref{assm:local-identification} hold as well as the moment bounds of \Cref{thm:feasible-local-power}. Then, if \((\Delta_N,\Delta_D)\to_p0\),
    \[
        \sup_{a \in \SR} \big|\Pr(\JK(\beta_0) \leq a) - \Pr(\JK_I(\beta_0) \leq a)\big|\to_p 0
    \]
\end{lemma}

\subsection{Proof of \Cref{thm:local-power}}
\label{subsec:local-power-proof}

Before proceeding, we will introduce some notation. Let \(\tilde H = s_n H\) and \(\tilde h_{ij} = s_n h_{ij}\), where \(s_n\) is as in \Cref{assm:balanced-design}. Recall that \(\tilde h_{ii} = 0\) and define
\begin{align*}
    N &:= \frac{1}{\sqrt{n}} \sum_{i=1}^n \eps_i(\beta_0) \sum_{j=1}^n \tilde h_{ij} r_j &  \tilde N &:=  \frac{1}{\sqrt{n}}\sum_{i=1}^n \tilde\eps_i(\beta_0)\sum_{j=1}^n \tilde h_{ij} \tilde r_j\\ 
    D &:= \frac{1}{n}\sum_{i=1}^n \eps_i^2(\beta_0)\big(\sum_{j=1}^n \tilde h_{ij}r_j\big)^2 & \tilde D &:= \frac{1}{n}\sum_{i=1}^n \kappa_i^2(\beta_0)\big(\sum_{j=1}^n \tilde h_{ij}\tilde r_j\big)^2
\end{align*}
where \((\tilde\eps_i(\beta_0),\tilde r_i)\) are jointly Gaussian with the same mean and covariance matrix as \((\eps_i(\beta_0),r_i)\) and \(\kappa_i^2(\beta_0) = \E[\eps_i^2(\beta_0)]\).
Under this notation we can write \(\JK_I(\beta_0) = \frac{N^2}{D}\bm{1}_{\{D  > 0\}}\) and \(\JK_G(\beta_0) = \frac{\tilde N^2}{\tilde D}\). Dealing with these forms of the statistics is difficult for the interpolation argument, since the denominator is random. Instead, we will notice that since \(D = 0 \implies N = 0\) and \(\Pr(\tilde D > 0) = 1\), for any \(a \geq 0\) we can rewrite the events
\begin{equation}
    \label{eq:equivalent-events}
    \begin{split}
        \{\JK_I(\beta_0) \leq a\} = \{N^2 - a D \leq 0\} \andbox \{\JK_G(\beta_0) \leq a\} \overset{\text{a.s}}{=} \{\tilde N^2 - a \tilde D \leq 0\}
    \end{split}
\end{equation}
With this in mind define
\begin{align*}
    \JK^a := N^2 - a D \andbox \tilde \JK^a := \tilde N^2 - a \tilde D
\end{align*}
Showing \Cref{thm:local-power} is then equivalent to showing that \(\sup_a|\Pr(\JK^a \leq 0) - \Pr(\tilde \JK^a \leq 0)|\to 0\). The statement \(\sup_{a < 0}\big|\Pr(\JK_I(\beta_0) \leq a) - \Pr(\JK_G(\beta_0) \leq a)\big| = 0\) is immediate since both \(\JK_I(\beta_0)\) and \(\JK_G(\beta_0)\) are always weakly positive.
It thus suffices to show
\[
    \sup_{a \geq 0} \big|\Pr(\JK_I(\beta_0) \leq a) - \Pr(\JK_G(\beta_0) \leq a)\big| \to 0
\]
We do so in a few lemmas, the final result being shown in \Cref{lemma:approximate-distribution} at the bottom of this subsection.

\begin{lemma}[Lindeberg Interpolation]
    \label{lemma:lindeberg-interpolation}
    Suppose that \Cref{assm:balanced-design,assm:local-identification} hold along with the conditions of \Cref{thm:feasible-local-power}. Let \(\varphi(\cdot):\SR\to\SR\) be such that \(\varphi(\cdot) \in C_b^3(\SR)\) with \(L_2(\varphi) = \sup_x|\varphi''(x)|\) and \(L_3(\varphi) = \sup_x|\varphi'''(x)|\). Then, there is a constant \(M\) that depends only on the constant \(c\) such that:
    \[
        |\E[\varphi(\JK^a) - \varphi(\tilde \JK^a)]| \leq \frac{M(a^3 \vee 1)}{\sqrt{n}}(L_2(\varphi) + L_3(\varphi))
    \]
\end{lemma}
\begin{proof}[Proof of \Cref{lemma:lindeberg-interpolation}]
   Begin by defining the leave-one-out numerator, denominator, and decomposed statistics
   \begin{align*}
       N_{-i} &:= \frac{1}{\sqrt{n}}\sum_{j \neq i} \dot\eps_j(\beta_0)\sum_{\ell \neq i} \tilde h_{j\ell }\dot r_\ell  & 
       D_{-i} &:= \frac{1}{n}\sum_{j \neq i}\ddot\eps_{j}^2(\beta_0)\big(\sum_{\ell\neq i} \tilde h_{j\ell} \dot r_\ell\big)^2 \\ 
               \JK_{-i} &:= N_{-i}^2 - aD_{-i}  &
   \end{align*}
   where for each \(\ell \in [n]\), \(\dot\eps_\ell(\beta_0)\) is equal to \(\eps_\ell(\beta_0)\) if \(\ell > i\) and \(\tilde \eps_\ell(\beta_0)\) if \(\ell < i\), \(\dot r_\ell\) is equal to \(r_\ell\) if \(\ell > i\) and \(\tilde r_\ell\) if \(\ell < i\), and \(\ddot\eps_\ell^2(\beta_0)\) is equal to \(\kappa_\ell^2(\beta_0)\) if \(\ell < i\) and \(\eps_\ell^2(\beta_0)\) if \(\ell > i\). While the definitions of \(\dot\eps_\ell, \dot r_\ell,\) and \(\ddot \eps_\ell\) depend on \(i\) because we will be considering only one deviation at a time, we will supress the dependence of these variables on \(i\) to simplify notation.

   Next, define the one-step deviations
   \begin{equation}
   \label{eq:one-step-deviations}
   \begin{split}
       \Delta_{1i}
       &:=  \eps_i(\beta_0)\sum_{j=1}^n \tilde h_{ij} \dot r_j + r_i \sum_{j=1}^n \tilde h_{ji} \dot\eps_j(\beta_0)\\
       \tilde \Delta_{1i}
       &:=  \tilde\eps_i(\beta_0)\sum_{j=1}^n \tilde h_{ij} \dot r_j + \tilde r_i \sum_{j=1}^n \tilde h_{ji} \dot\eps_j(\beta_0) \\ 
       \Delta_{2i} &:= \underbrace{a\eps_i^2(\beta_0)(\sum_{j=1}^n \tilde h_{ij} \dot r_j)^2 + a r_i^2 \sum_{j=1}^n \tilde h_{ji}^2 \ddot \eps_j^2(\beta_0)}_{\Delta_{2i}^a} + \underbrace{2a r_i\sum_{j=1}^n \ddot \eps_j^2(\beta_0)\sum_{\ell \neq i} \tilde h_{j \ell}\tilde h_{ji}\dot r_\ell}_{\Delta_{2i}^b}\\
       \tilde\Delta_{2i} &:= \underbrace{a\kappa_i^2(\beta_0)(\sum_{j=1}^n \tilde h_{ij} \dot r_j)^2 + a \tilde r_i^2 \sum_{j=1}^n \tilde h_{ji}^2 \ddot \eps_j^2(\beta_0)}_{\Delta_{2i}^a} + \underbrace{2a \tilde r_i\sum_{j=1}^n \ddot \eps_j^2(\beta_0)\sum_{\ell \neq i} \tilde h_{j\ell} \tilde h_{ji}\dot r_\ell}_{\tilde\Delta_{2i}^b}
   \end{split}
   \end{equation}
   These one-step deviations contain all the terms associated with observation \(i\) in the expression of the numerator and denominator of the test statistics. To demonstrate, note that these one-step deviations satisfy \(N_{-1} + n^{-1/2}\Delta_{11} = N\) and \(aD_{-1} + n^{-1}\Delta_{21} = aD\) as
   \begin{align*}
       N &= \frac{1}{\sqrt{n}}\sum_{i=1}^n \eps_i(\beta_0) \sum_{j=1}^n \tilde h_{ij} r_j \\ 
         &= \frac{1}{\sqrt{n}}\sum_{j > 1} \eps_j(\beta_0) \sum_{\ell = 1}^n \tilde h_{j\ell} r_j + \eps_1(\beta_0) \frac{1}{\sqrt{n}}\sum_{j > 1} \tilde h_{1j}r_j \\ 
         &= \frac{1}{\sqrt{n}}\sum_{j > 1}\eps_j(\beta_0)\Big\{\tilde h_{j1} r_1 + \sum_{\ell > 1} h_{j\ell} r_\ell \Big\} + \eps_1(\beta_0) \frac{1}{\sqrt{n}}\sum_{j > 1}\tilde h_{1j}r_j \\ 
         &= \underbrace{\frac{1}{\sqrt{n}}\sum_{j > 1}\eps_j(\beta_0)\sum_{\ell > 1}h_{j\ell} r_\ell}_{N_{-1}} + \underbrace{\eps_1(\beta_0)\frac{1}{\sqrt{n}}\sum_{j > 1}\tilde h_{1j} r_j + r_1 \frac{1}{\sqrt{n}}\sum_{j > 1}\tilde h_{j1}\eps_j(\beta_0)}_{n^{-1/2}\Delta_{11}}
    \intertext{and}
       D &= \frac{1}{n} \sum_{i=1}^n \eps_i^2(\beta_0)\big(\sum_{j=1}^n \tilde h_{ij} r_j)^2 \\ 
         &= \frac{1}{n}\sum_{j > 1}\eps_j^2(\beta_0)\big(\sum_{\ell = 1}^n \tilde h_{j\ell}r_\ell\big)^2 + \eps_1^2(\beta_0)\frac{1}{n}\big(\sum_{j > 1} \tilde h_{1j}r_j\big)^2\\ 
         &= \frac{1}{n}\sum_{j > 1}\eps_j^2(\beta_0) \big(\tilde h_{j1}r_1 + \sum_{\ell \neq 1} \tilde h_{\ell j}r_\ell\big)^2 + \eps_1^2(\beta_0)\frac{1}{n}\big(\sum_{j > 1}\tilde h_{1j}r_j\big)^2\\ 
         &= \underbrace{\frac{1}{n}\sum_{j > 1}\eps_j^2(\beta_0)\big(\sum_{\ell > 1}\tilde h_{\ell, j}r_\ell\big)^2}_{D_{-1}} \\ 
         &\;\;\;\;\;\;+ \underbrace{\eps_1^2(\beta_0)\frac{1}{n}\big(\sum_{j > 1} \tilde h_{1j}r_j\big)^2 + r_1^2\frac{1}{n}\sum_{j > 1} \tilde h_{j1}^2 \eps_j^2(\beta_0) + 2r_1\frac{1}{n}\sum_{j > 1}\eps_j^2(\beta_0)\sum_{\ell > 1}\tilde h_{\ell j} r_\ell}_{(an)^{-1}\Delta_{21}}
   \end{align*}
   Using the one-step deviations, write the difference \(\E[\varphi(K^a) - \varphi(\tilde K^a)]\) as a telescoping sum, one by one replacing \((\Delta_{1i},\Delta_{2i})\) with \((\tilde\Delta_{1i},\tilde\Delta_{2i})\) in the expressions of \(\JK^a = N^2 - aD\) until we arrive at \(\tilde \JK^a = \tilde N^2 - a \tilde D\).
   \begin{equation}
       \label{eq:telescoping-sum}
       \begin{split}
           \E[\varphi(\JK^a) - \varphi(\tilde \JK^a)] 
           &= \sum_{i=1}^n \E[\varphi(\JK_{-i} + n^{-1/2}N_{-i}\Delta_{1i} + n^{-1}\Delta_{1i}^2 - n^{-1}\Delta_{2i})] \\ 
           &\hphantom{=\sum_{i=1}^n\;\;} - \E[\varphi(\JK_{-i} + n^{-1/2}N_{-i}\tilde\Delta_{1i} + n^{-1}\tilde\Delta_{1i}^2 - n^{-1}\tilde\Delta_{2i})]
       \end{split}
   \end{equation}
   Via a second-order Taylor expansion, we can write each term inside the summand
   \begin{align*}
       \E[\text{Term}_i] 
       &= \E[\varphi'(\JK_{-i})\{2n^{-1/2}N_{-i}(\Delta_{1i}-\tilde\Delta_{1i}) + n^{-1}(\Delta_{1i}^2 - \Delta_{1i}^2) - n^{-1}(\Delta_{2i} - \tilde\Delta_{2i})\}] \\ 
       &+ \E[\varphi''(\JK_{-i})\{4n^{-1}N_{-i}^2(\Delta_{1i}^2 - \tilde\Delta_{1i}^2) + n^{-2}(\Delta_{1i}^4 - \tilde\Delta_{1i}^4) - n^{-2}(\Delta_{2i}^2 - \Delta_{2i}^2)\}] \\
       &+ \E[\varphi''(\JK_{-i})\{4n^{-3/2}N_{-i}(\Delta_{1i}^3 - \tilde\Delta_{1i}^3) + 4n^{-3/2}N_{-i}(\Delta_{1i}\Delta_{2i} - \tilde\Delta_{1i}\tilde\Delta_{2i})\}] \\ 
       &+ \E[\varphi''(\JK_{-i})\{2n^{-2}(\Delta_{1i}^2 \Delta_{2i} - \tilde\Delta_{1i}^2 \tilde\Delta_{2i})\}]  + R_i  + \tilde R_i
   \end{align*}
   where \(R_i\) and \(\tilde R_i\) are remainder terms to be examined later. Let \(\calF_{-i}\) denote the sigma algebra generated by all random variables whose index is not equal to \(i\). Since (a) for each \(i \in [n]\) the mean and covariance matrix of \((\eps_i(\beta_0), r_i)\) is the same as the mean and covariance matrix of \((\tilde\eps_i(\beta_0), \tilde r_i)\), (b) \(\E[\eps_i^2(\beta_0)] = \kappa_i^2(\beta_0)\), and (c) random variables are independent across indices, we have that 
   \begin{equation}
        \label{eq:interpolation-matched-moments}
        \begin{split}
            \E[\Delta_{1i} - \tilde\Delta_{1i}|\calF_{-i}] 
            &= \E[\Delta_{1i}^2 - \tilde\Delta_{1i}^2|\calF_{-i}] = \E[\Delta_{2i} - \tilde\Delta_{2i}|\calF_{-i}] \\ &= \E[\Delta_{2i}^b - \tilde\Delta_{2i}^b|\calF_{-i}]  = \E[\Delta_{1i}\Delta_{2i}^b - \tilde\Delta_{1i}\tilde\Delta_{2i}^b|\calF_{-i}] = 0
        \end{split}
    \end{equation}
   Using this we can simplify the prior display 
   \begin{align*}
       \E[\text{Term}_i] 
       &= \underbrace{n^{-2}\E[\varphi''(\JK_{-i})(\Delta_{1i}^4 - \Delta_{1i}^4)]}_{\vA_i} - \underbrace{n^{-2}\E[\varphi''(\JK_{-i})((\Delta_{2i}^a)^2 - (\tilde\Delta_{2i}^a)^2)]}_{\vB_i} \\ 
       &- \underbrace{2n^{-2}\E[\varphi''(\JK_{-i})(\Delta_{2i}^a \Delta_{2i}^b  - \tilde\Delta_{2i}^a \tilde\Delta_{2i}^b)]}_{\vC_i} + \underbrace{4n^{-3/2}\E[\varphi''(\JK_{-i})N_{-i}(\Delta_{1i}^3 - \tilde\Delta_{1i}^3)]}_{\vD_i} \\ 
       &+ \underbrace{4n^{-3/2}\E[\varphi''(\JK_{-i})N_{-i}(\Delta_{1i}\Delta_{2i}^a - \tilde\Delta_{1i}\tilde\Delta_{2i}^a)]}_{\vE_i} + \underbrace{2 n^{-2}\E[\varphi''(\JK_{-i})(\Delta_{1i}^2 \Delta_{2i} - \tilde\Delta_{1i}^2 \tilde\Delta_{2i})}_{\vF_i} \\ 
       &+ R_i + \tilde R_i
   \end{align*}
   where for some \(\bar \JK_{1i}\) and \(\bar \JK_{2i}\) we can write
   \begin{align*}
       R_i &= \E[\varphi'''(\bar \JK_{1i})\{n^{-1/2}N_{-i}\Delta_{1i} + n^{-1}\Delta_{1i}^2 + n^{-1}\Delta_{2i}\}^3]\\
       \tilde R_i &= \E[\varphi'''(\bar \JK_{2i})\{n^{-1/2}N_{-i}\tilde\Delta_{1i} + n^{-1}\tilde\Delta_{1i}^2 + n^{-1}\tilde\Delta_{2i}\}^3]
   \end{align*}
   Applications of \Cref{lemma:delta-moment-bounds,lemma:numerator-bound}, Cauchy-Schwarz, and the generalized  Hölder inequality,\footnote{\(\E[|fgk|]^3 \leq \E[|f|^3]\E[|g|^3]\E[|k|^3]\)} will allow us to bound for a fixed constant \(M\) that depends only on \(c\),
   \begin{align*}
       |\vA_i| &\leq  \frac{M}{n^2}L_2(\varphi) 
      & |\vB_i| &\leq \frac{Ma^2}{n^2}L_2(\varphi) 
      & |\vC_i| &\leq \frac{Ma^2}{n^{3/2}}L_2(\varphi)  \\ 
       |\vD_i| &\leq \frac{M}{n^{3/2}}L_2(\varphi) 
      & |\vE_i| &\leq \frac{M (a\vee 1)}{n^{3/2}}L_2(\varphi)  
      & |\vF_i| &\leq \frac{Ma^3}{n^{3/2}}L_2(\varphi)
   \end{align*}
   and
   \begin{align*}
       |R_i| + |\tilde R_i| &\leq \frac{M}{n^{3/2}}L_3(\varphi) + \frac{Ma^3}{n^3}L_3(\varphi)
   \end{align*}
   Combining these bounds and summing over \(n\) gives the result.
\end{proof}

\begin{lemma}[Gaussian Denominator Anti-Concentration]
    \label{lemma:denom-anti-concentration}
    Suppose that the conditions of \Cref{thm:feasible-local-power} and \Cref{assm:balanced-design} hold. Then, for any sequence \(\delta_n \searrow 0\), 
    \[
       \Pr(\tilde D \leq \delta_n) \to 0
    \]
\end{lemma}
\begin{proof}[Proof of \Cref{lemma:denom-anti-concentration}]
   Since \(\kappa_i^2(\beta_0) \in [c^{-1},c]\) for all \(i = 1,\dots,n\) we have that \(\tilde D \geq \frac{c^{-1}}{n}\sum_{i=1}^n (\sum_{j=1}^n \tilde h_{ij}r_j)^2\). Then
    \begin{align*}
        \Pr(\tilde D \leq \delta_n) 
            &\leq  \Pr\big(\frac{1}{cn}\sum_{i=1}^n\big(\sum_{j=1}^n \tilde h_{ij}\tilde r_j\big)^2 \leq \tilde \delta_n \big) \\ 
            &= \Pr\big(\|\tilde r' \bar H^{1/2}\|^2 \leq \delta_n\big) \numberthis\label{eq:anticoncentration-quadratic-form}
    \end{align*}
    where \(\tilde r := (\tilde r_1,\dots,\tilde r_n)'\in \SR^n\) and \(\bar H := \frac{1}{cn}\tilde H \tilde H' \in \SR^{n\times n}\). \(\bar H\) is symmetric and positive semidefinite so we can take \(\bar H^{1/2}\) to be its symmetric square root, which will also be symmetric and positive semidefinite (and thus not necessarily equal to \(\sqrt{\frac{c}{n}}\tilde H\)). I provide two bounds on \eqref{eq:anticoncentration-quadratic-form}, the first of which corresponds to the strong identification setting while the second corresponds to weak identification.

    \emph{First Bound.}   Since \(\delta_n \searrow 0\) we will eventually have that \(\delta_n < c^{-1}/2\). When this happens we can bound using Chebyshev's inequality and \(c^{-1} < \E[r'\bar H r] < c\):
    \begin{align*}
        \Pr(\tilde r'\bar H \tilde r \leq \delta_n) 
        &= \Pr(\tilde r'\bar H \tilde r - \E[\tilde r' \bar H \tilde r] \leq \delta_n - \E[\tilde r' \bar H \tilde r]) \\ 
        &\leq \Pr(\tilde r'\bar H \tilde r - \E[r'\bar H r] \geq \E[\tilde r'\bar H \tilde r] - \delta_n) \\
        &\leq \Pr(|\tilde r'\bar H \tilde r - \E[r'\bar H r]| \geq \frac{1}{2c}) \\ 
        &\leq 2c\Var(r'\bar H r)\numberthis\label{eq:anticoncentration-first-bound}
    \end{align*}
    Under strong identification we will expect \(\Var(r'\bar H r) \to 0\).

    \emph{Second Bound.} For the second bound, we will directly use bounds on the density of Gaussian quadratic forms from \citet{anticoncentration-bernoulli-gotze-et-al}. The vector \(r'\bar H^{1/2}\) is Gaussian with covariance matrix \(\Sigma_r = \bar H^{1/2}\mR \bar H^{1/2}\) where \(\mR = \diag(\Var(r_1),\dots,\Var(r_n))\). Let \(\Lambda_1 = \sum_{k=1}^n \lambda_k^2(\Sigma_r)\) and  \(\Lambda_2 = \sum_{k=2}^n\lambda_k^2(\Sigma_r)\). By \Cref{assm:balanced-design} and \Cref{lemma:quadratic-form-eigenvalues}, \(\Lambda_2/\Lambda_1\) is bounded away from zero. Using \Cref{thm:quadratic-density-bound} we can then bound for some constant \(C > 0\)
    \begin{equation}
        \label{eq:anticoncentration-second-bound}
        \begin{split}
            \Pr(\|r'H\|^{1/2} \leq \delta_n) \leq C\delta_n\Lambda_1^{-1} 
        \end{split}
    \end{equation}
    
    \emph{Combining Bounds.} To combine the bounds in \eqref{eq:anticoncentration-first-bound} and \eqref{eq:anticoncentration-second-bound}, first write
    \[
        \Var(\tilde r'\bar H\tilde r) = 2\trace(\mR\bar H\mR \bar H) + 4\mu_r\bar H \mR \bar H \mu_r
    \]
    for \(\mu_r = \E[r]\). Using the fact that \(\bar H^{1/2}\mR \bar H^{1/2}\) is symmetric positive definite we can bound:
    \begin{align*}
        \mu_r'\bar H \mR\bar H \mu_r 
        &= (\mu_r'\bar H^{1/2})'(\bar H^{1/2}\mR\bar H^{1/2})(\bar H^{1/2}\mu_r) \\ 
        &\leq \lambda_1(\bar H^{1/2}\mR\bar H^{1/2})\|\mu_r'\bar H^{1/2}\|^2 \\ 
        &= \sqrt{\lambda_1^2(\bar H^{1/2}\mR\bar H^{1/2})}\|\mu_r'\bar H^{1/2}\|^2\\ 
        &= \sqrt{\lambda_1(\bar H^{1/2}\mR\bar H \mR \bar H^{1/2})}\|\mu_r'\bar H^{1/2}\|^2\\ 
        &\leq \sqrt{\trace(\bar H^{1/2}\mR\bar H \mR \bar H^{1/2})}\|\mu_r'\bar H^{1/2}\|^2\\
        &= \sqrt{\trace(\mR\bar H \mR \bar H)}\|\mu_r'\bar H\|^2  
        \leq c^2\Lambda_1^{1/2}\numberthis\label{eq:anticoncentration-combination-bound}
    \end{align*}
    where the first equality uses the symmetric square root of \(\bar H\), the first inequality comes from Courant-Fischer minmax principle and the third equality uses the fact that the eigenvalues of \(A^2\) are the squares of the eigenvalues of \(A\), for any generic symmetric matrix \(A\). The second inequality comes from the fact that a matrix times its transpose is always positive semidefinite and that for \(M\) psd, \(\lambda_1(M) \leq \sqrt{\trace(M^2)}\) since the trace is the sum of the (weakly positive) eigenvalues. The final inequality uses \(\mu_r'\bar H \mu_r = \frac{c}{n}\sum_{i=1}^n(\E[\tilde\Pi_i])^2 \leq \frac{c}{n}\sum_{i=1}^n \E[(\tilde\Pi_i)^2] \leq c^2 \).

    Combining \eqref{eq:anticoncentration-first-bound},~\eqref{eq:anticoncentration-second-bound}, and \eqref{eq:anticoncentration-combination-bound} gives us
    \begin{equation}
        \label{eq:anticoncentration-final-bound}
        \Pr(\tilde D \leq \delta_n) \leq C\min\left\{\Lambda_1 + \Lambda_1^{1/2}, \delta_n \Lambda_1^{-1}  \right\}
    \end{equation}
    Regardless of the behavior of \(\Lambda_1\), this tends to zero as \(\delta_n\to 0\).
\end{proof}
\begin{remark}[Final Anticoncentration Bound]
    \label{rem:explicit-anticoncentration}
    To give an explicit bound on \eqref{eq:anticoncentration-final-bound} in terms of \(\delta_n\) we note that, if \(x^\star\) solves
    \[
        x^\star + \sqrt{x^\star} = \frac{c}{x^\star} 
    \]
    then for any \(x \geq 0\), \(\min\{x + \sqrt{x}, c/x\} \leq x^\star + \sqrt{x^\star}\). Using this, notice that \((x^\star)^2 + (x^\star)^{3/2} = c\) so that \(x^\star \leq \sqrt{c}\). This allows us to bound \eqref{eq:anticoncentration-final-bound}
    \begin{align*}
        \Pr(\tilde D \leq \delta_n) \leq C\min\{\Lambda_1 + \Lambda_1^{1/2},\delta_n\Lambda_1^{-1}\} \leq C(\delta_n^{1/2} + \delta_n^{1/4})
    \end{align*}
    
\end{remark}
\begin{lemma}
    \label{lemma:decomp-to-ratio}
    Let \(X_n\) and \(Y_n\) be two sequences of random variables and let \(W_n = X_n/Y_n\). Then for any \(c \in \SR\) and any \(\delta > 0\):
    \begin{align*}
        \Pr(0 \leq X_n - cY_n \leq  \delta) 
        &\leq \Pr(c \leq W_n \leq \delta^{1/2} + c) + \Pr(Y_n \leq \delta^{1/2}) 
        \intertext{and}
        \Pr( - \delta \leq X_n - cY_n \leq 0) &\leq \Pr( c - \delta^{1/2} \leq W_n \leq c) + \Pr(Y_n \leq \delta^{1/2})
    \end{align*}
\end{lemma}
\begin{proof}
    Define the event \(\Omega = \{Y_n \geq \delta^{1/2}\}\). We can bound
    \begin{align*}
        \Pr(0 \leq X_n - cY_n \leq \delta) 
        &= \Pr(cY_n \leq X_n \leq \delta + cY_n) \\ 
        &\leq \Pr(\{cY_n \leq X_n \leq \delta + cY_n\} \cap \Omega) + \Pr(\Omega^c) \\ 
        &= \Pr(\{c \leq W_n \leq \delta/Y_n + c\}\cap\Omega) + \Pr(\Omega^c) \\ 
        &\leq \Pr(c \leq W_n \leq \delta^{1/2} + c) + \Pr(\Omega^c)
    \end{align*}
    The second statement of the lemma follows symmetrically.
\end{proof}

\begin{lemma}[]
    \label{lemma:Op1-comp}
    Suppose that \(X_n\) and \(Y_n\) are sequences of (real-valued) random variables such that \(Y_n = O_p(1)\) and for any \(x \in \SR\)
    \[
        |\Pr(X_n \leq x) - \Pr(Y_n \leq x)| \to 0
    \]
    Then \(X_n = O_p(1)\).
\end{lemma}
\begin{proof}
    Pick any \(\eps > 0\), and let \(M_{\eps/2}\) be such that \(\Pr(Y_n > M_{\eps /2}) \leq \eps/2 \) for all \(n \geq N_\eps\). In addition, let \(\tilde N_\eps\) be such that \(|\Pr(X_n \leq M_{\eps/2}) - \Pr(Y_n \leq M_{\eps/2})| \leq \eps/2\) for all \(n \geq \tilde N_\eps\). Then for all \(n \geq N_\eps \vee \tilde N_{\eps/2}\), 
    \begin{align*}
    \Pr(X_n > M_{\eps/2}) 
    &\leq \Pr(Y_n > M_{\eps/2}) + |\Pr(X_n > M_{\eps/2}) - \Pr(Y_n > M_{\eps/2})| \\ 
    &\leq \eps/2 + |\Pr(Y_n \leq M_{\eps/2}) - \Pr(X_n \leq M_{\eps/2})| \\ 
    &\leq \eps/2 + \eps/2 = \eps
    \end{align*}
\end{proof}
\begin{lemma}[]
    \label{lemma:Op1-gc}
    Suppose that \(X_n\) and \(Y_n\) are sequences of (real-valued) random variables such that \(Y_n = O_p(1)\) and for any \(\Delta \in \SR\)
    \[
        \sup_{x \leq \Delta}|\Pr(X_n \leq x) - \Pr(Y_n \leq x)| \to 0
    \]
    Then \(\sup_{x\in\SR}|\Pr(X_n \leq x) - \Pr(Y_n \leq x)| \to 0\).
\end{lemma}
\begin{proof}
    Pick an \(\eps > 0\). By \Cref{lemma:Op1-comp}, \(X_n = O_p(1)\). Pick a constant \(M_{\eps/3}\) such that \(\Pr(X_n > M_{\eps/3}) \leq \eps/3\) and \(\Pr(Y_n > M_{\eps/3}) \leq \eps/3\). Then for any \(x \in \SR\) we can bound \(|\Pr(X_n \leq x) - \Pr(Y_n \leq x)|\) by considering two cases:

    \textbf{Case 1.} If \(x \leq M_{\eps /3}\), then,
    \begin{equation}
        \label{eq:op1-gc-1}
        |\Pr(X_n \leq x) - \Pr(Y_n \leq x)| \leq \sup_{x \leq M_{\eps / 3}}|\Pr(X_n \leq x) - \Pr(Y_n \leq x)|
    \end{equation}
    by hypothesis, there is an \(N_{\eps}\) such that for \(n \geq N_{\eps}\) the RHS of \eqref{eq:op1-gc-1} is less than \(\eps\).

    \textbf{Case 2.} If \(x > M_{\eps/3}\) we can bound
    \begin{align*}
        |\Pr(X_n \leq x) - \Pr(Y_n \leq x)| 
        &\leq |\Pr(X_n \leq M_{\eps/3}) - \Pr(Y_n \leq M_{\eps/3})| \\ 
        &\;\;\;+ |\Pr(M_{\eps/3} < X_n \leq x) - \Pr(M_{\eps/3} < Y_n \leq x)| \\ 
        &\leq  |\Pr(X_n \leq M_{\eps/3}) - \Pr(Y_n \leq M_{\eps/3})| + \eps/3 + \eps/3 \numberthis\label{eq:op1-gc-2}
    \end{align*}
    By hypothesis, there is an \(N_{\eps/3}\) such that \(|\Pr(X_n \leq M_{\eps/3}) - \Pr(Y_n \leq N_{\eps/3})| \leq \eps/3\).

    WLOG \(N_{\eps/3} \geq N_\eps\). Combining the bounds in \eqref{eq:op1-gc-1} and \eqref{eq:op1-gc-2}, for any \(n \geq N_{\eps/3}\) and any \(x \in \SR\), 
    \[
        |\Pr(X_n \leq x) - \Pr(Y_n \leq x)| \leq \eps 
    \]
    Since this holds for all \(x\), this gives the result.
\end{proof}
\begin{lemma}[Approximate Distribution]
    \label{lemma:approximate-distribution}
    Under \Cref{assm:balanced-design,assm:local-identification} and the conditions of \Cref{thm:feasible-local-power}
    \[
        \sup_{a \in \SR}|\Pr(\JK_I(\beta_0) \leq a) - \Pr(\JK_G(\beta_0) \leq a)|\to 0
    \]
\end{lemma}
\begin{proof}[Proof of \Cref{lemma:approximate-distribution}]
    First, fix a \(\Delta \geq 0\) and consider any \(a \leq \Delta\). As in \Cref{lemma:denom-anti-concentration}, let \(\tilde \varphi(\cdot): \SR \to \SR\) be three times continuously differentiable with bounded derivatives up to the third order such that \(\tilde\varphi(x)\) is 1 if \(x \leq 0\), \(\tilde\varphi(x)\) is decreasing if \(x \in (0,1)\), and \(\tilde\varphi(x)\) is zero if \(x \geq 1\). Consider a sequence \(\gamma_n \searrow 0\) slowly enough such that \((\gamma_n^{-2} + \gamma_n^{-3})/\sqrt{n} \to 0\) and define \(\varphi_n(x) = \tilde\varphi(\frac{x}{\gamma_n})\).

    By \Cref{lemma:lindeberg-interpolation} we can write for some constant \(M\) that depends only on \(\Delta\):
    \begin{align*}
        \Pr(\JK_I(\beta_0) \leq a) 
        = \Pr(\JK^a \leq 0) 
        &\leq \E[\varphi_n(\JK^a)] \\ 
        &\leq \E[\varphi_n(\tilde \JK^a)] + \frac{M}{\sqrt{n}}(\gamma_n^2 + \gamma_n^{-3}) \\ 
        &\leq \Pr(\tilde\JK^a \leq 0) + \Pr(0 \leq \tilde N^2 - a \tilde D \leq \gamma_n) + \frac{M}{\sqrt{n}}(\gamma_n^2 + \gamma_n^{-3})
        \intertext{Applying \Cref{lemma:decomp-to-ratio} and \(\{\tilde\JK^a \leq 0\} = \{\JK_G(\beta_0) \leq a\}\) gives:}
        &\leq \Pr(\JK_G(\beta_0) \leq a) + \underbrace{\Pr(a \leq \tilde N^2/\tilde D \leq a + \gamma_n^{1/2})}_{\vA} \\
        &\;\;\;+ \underbrace{\Pr(\tilde D \leq \gamma_n^{1/2})}_{\vB}+ \frac{M}{\sqrt{n}}(\gamma_n^{-2} + \gamma_n^{-3}) 
    \end{align*}
    By \Cref{lemma:approximation-error-bound}, we can bound \(\vA \leq M\gamma_n^{1/2}\)  while by \Cref{lemma:denom-anti-concentration} and \Cref{rem:explicit-anticoncentration}, \(\vB \leq M \gamma_n^{1/4}\). Since \(\gamma_n\) is chosen such that \(\frac{M}{\sqrt{n}}(\gamma_n^{-2} + \gamma_n^{-3}) \to 0\) we can conclude that \(\Pr(\JK_I(\beta_0) \leq a) \leq \Pr(\JK_G(\beta_0) \leq a) + o(1)\). A symmetric argument with \(\varphi_n(x) = \tilde\varphi(1 - \frac{x}{\gamma_n})\) gives a lower bound so that, in total 
    \[
        \Pr(\JK_G(\beta_0) \leq a) - \ve \leq \Pr(\JK_I(\beta_0) \leq a) \leq \Pr(\JK_G(\beta_0) \leq a) + \ve
    \]
    where 
    \[
        \ve = M\big(\frac{\gamma_n^{-2} + \gamma_n^{-3}}{\sqrt{n}} + \gamma_n^{1/2} + \gamma_n^{1/4}\big) = o(1)
    \]
    Since the constant M depends only on \(\Delta\), this gives us that for any fixed \(\Delta > 0\)
    \begin{equation}
        \label{eq:almost-uniform-bound}
        \sup_{a \leq \Delta}\big|\Pr(\JK_I(\beta_0) \leq a) - \Pr(\JK_G(\beta_0) \leq a)\big| \leq C\big(\frac{\gamma_n^{-2} + \gamma_n^{-3}}{\sqrt{n}} + \gamma_n^{1/2} + \gamma_n^{1/4}\big) = o(1)
    \end{equation}
    where \(C\) is a constant that depends only on \(\Delta\).
    Noting that the numerator \(\JK_G(\beta_0)\) is \(O_p(1)\) under \Cref{assm:local-identification} while the inverse of the denominator of \(\JK_G(\beta_0)\) is \(O_p(1)\) by \Cref{lemma:denom-anti-concentration}, we can apply \Cref{lemma:Op1-gc}. This step shows that the result in \eqref{eq:almost-uniform-bound} implies that the approximation error tends to zero uniformly over the real line, which is the desired result. Optimizing over \(\gamma_n\) in the expression of \eqref{eq:almost-uniform-bound} yields the rate of decay in \Cref{rem:infeasible-error}.
\end{proof}

\subsection{Proof of \Cref{lemma:estimation-error-high-level}}
\label{subsec:estimation-error}

\begin{proof}[Proof of \Cref{lemma:estimation-error-high-level}]
     For \(N\) and \(D\) defined at the top of \Cref{subsec:local-power-proof} define \(\widehat N = N + \Delta_N\) and \(\widehat D = D + \Delta_D\). We can then write \(\JK(\beta_0) = \widehat N^2 / \widehat D\) and rewrite
\[
    \JK(\beta_0) - \JK_I(\beta_0) = \frac{2ND\Delta_N + D\Delta_N - N^2\Delta_D}{D^2 + D\Delta_D} 
\]
Apply \Cref{lemma:numerator-bound} to see that \(N^2 = O_p(1)\) while under \Cref{assm:balanced-design}, \(D = O_p(1)\). Thus, \(2ND\Delta_n + D\Delta_n - N^2\Delta_D = o_p(1)\). Meanwhile, by \Cref{lemma:denom-anti-concentration-regular}, \(\Pr(D^2 \leq  \delta_n) \to 0\) for any sequence \(\delta_n \to 0\). Apply \Cref{lemma:estimation-error-fraction} to obtain that \(|\JK(\beta_0) - \JK_I(\beta_0)|\to_p 0\).

Finally, apply \Cref{lemma:probability-to-distribution} with \(X_n = \JK(\beta_0)\), \(Y_n = \JK_I(\beta_0)\) and \(Z_n = \JK_G(\beta_0)\). The density of \(Z_n\) is uniformly bounded by \Cref{lemma:approximation-error-bound} to show that the distribution of \(\JK(\beta_0)\) may be uniformly approximated by the distribution of \(\JK_G(\beta_0)\).
\end{proof}

\begin{lemma}[]
    \label{lemma:estimation-error-fraction}
    Let \(A_n, B_n\) and \(Y_n\) be sequences of random variables such that \(A_n = o_p(1)\) and \(B_n = o_p(1)\). If \(Y_n\) is such that for any sequence \(\delta_n \to 0\), \(\Pr(|Y_n| \leq \delta_n)\to 0\), then,
    \[
        \bigg|\frac{A_n}{Y_n + B_n}\bigg| = o_p(1)
    \]
\end{lemma}
\begin{proof}
    Fix any \(\eps > 0\). We show that
    \[
        \bigg|\frac{A_n}{Y_n + B_n}\bigg| \leq \eps
    \]
    on an intersection of events whose probability tends to one. By \Cref{lemma:op1-sequence} there is a sequence \(\eps_n \searrow 0\) such that
    \[
        \Pr(|A_n| \leq \eps_n) \to 1 \andbox \Pr(\eps|B_n| \leq \eps_n)\to 1
    \]
    Consider the intersection of events \(\Omega_1 \cap\Omega_2\cap\Omega_3\) where
    \[
        \Omega_1 := \{\eps|Y_n| \geq 2\eps_n\},\;\;\Omega_2 := \{\eps|B_n| \leq \eps_n\},\;\;\Omega_3 := \{|A_n| \leq \eps_n\}
    \]
    By assumption, \(\Pr(\Omega_1\cap\Omega_2\cap\Omega_3) \to 1\). On this event \(|Y_n + B_n| \geq \eps_n/\eps > 0\) and \(|A_n| \leq \eps_n\) so that \(|A_n/(Y_n + B_n)| \leq |\eps_n /(\eps_n/\eps)| \leq \eps\).
\end{proof}

\begin{lemma}[Denominator Interpolation]
    \label{lemma:denominator-interpolation}
    Suppose that the moment bounds of \Cref{thm:feasible-local-power} and \Cref{assm:balanced-design} hold. Let \(\varphi(\cdot):\SR \to \SR\) be such that \(\varphi(\cdot) \in C_b^3(\SR)\) with \(L_2(\varphi) = \sup_x|\varphi''(x)|\) and \(L_3(\varphi) = \sup_x|\varphi'''(x)|\). Then there is a constant \(M\) that depends only on the constant \(c\) such that: 
    \[
        |\E[\varphi(D) - \varphi(\tilde D)]| \leq \frac{M}{\sqrt{n}}(L_2(\varphi) + L_3(\varphi))
    \]
\end{lemma}
\begin{proof}[Proof of \Cref{lemma:denominator-interpolation}]
    We inherit the definitions of \(D_{-i}\), \(\Delta_{2i}^a\), \(\Delta_{2i}^b\), \(\tilde\Delta_{2i}^a\), and \(\tilde\Delta_{2i}^b\) from the proof of \Cref{lemma:lindeberg-interpolation} with \(a = 1\). Then, as before we can write 
    \begin{align*}
        \E[\varphi(D) - \varphi(\tilde D)] 
        &= \sum_{i=1}^n \E[\varphi(D_{-i} + n^{-1}\Delta_{2i}^a + n^{-1}\Delta_{2i}^b)] \\ 
        &\hphantom{=\sum_{i=1}^n }-\E[\varphi(D_{-i} + n^{-1}\tilde\Delta_{2i}^a + n^{-1}\tilde\Delta_{2i}^b)]
    \end{align*}
    We examine each term via a second-order Taylor expansion around \(D_{-i}\)
    \begin{align*}
        \E[\text{Term}_i] 
        &= \frac{1}{n}\E[\varphi'(D_{-i})\{(\Delta_{2i}^a - \tilde\Delta_{2i}^a) + (\Delta_{2i}^b - \tilde\Delta_{2i}^b)\}] \\ 
        &+ \frac{1}{2n^2}\E[\varphi''(D_{-i})\{((\Delta_{2i}^a)^2 - (\tilde\Delta_{2i}^a)^2) + 2(\Delta_{2i}^a \Delta_{2i}^b - \tilde\Delta_{2i}^a \tilde\Delta_{2i}^b) + ((\Delta_{2i}^b)^2 - (\Delta_{2i}^b)^2)\}] \\ 
        &+ R_i + \tilde R_i
    \end{align*}
    where \(R_i\) and \(\tilde R_i\) are remainder terms to be analyzed later. Using the restrictions in \eqref{eq:interpolation-matched-moments} we can simplify the above display:
    \begin{align*}
        \E[\text{Term}_i]
        &= \underbrace{0.5n^{-2}\E[\varphi''(D_{-i})((\Delta_{2i}^a)^2 - (\tilde\Delta_{2i}^a)^2)]}_{\dot\vA_i} + \underbrace{n^{-2}\E[\varphi''(K_{-i})(\Delta_{2i}^a \Delta_{2i}^b - \tilde\Delta_{2i}^a \tilde\Delta_{2i}^b)}_{\dot\vB_i} \\ 
        &+ R_i + \tilde R_i
    \end{align*}
    Using \Cref{lemma:delta-moment-bounds} we can bound
    \begin{align*}
        |\vA_i| &\leq \frac{M}{n^2}L_2(\varphi)  & |\vB_i| &\leq \frac{M}{n^{3/2}}L_2(\varphi)
    \end{align*}
    For some \(\bar D_{1i}\) and \(\bar D_{2i}\) we can express 
    \begin{align*}
        R_i &= \E[\varphi'''(\bar D_{1i})\{n^{-1}\Delta_{2i}^a + \Delta_{2i}^b\}^3] \leq \frac{M}{n^{3/2}}L_3(\varphi) + \frac{M}{n^3}L_3(\varphi)\\
        R_i &= \E[\varphi'''(\bar D_{2i})\{n^{-1}\tilde\Delta_{2i}^a + \tilde\Delta_{2i}^b\}^3] \leq \frac{M}{n^{3/2}}L_3(\varphi) + \frac{M}{n^3}L_3(\varphi)
    \end{align*}
    where the inequalities again come from applications of \Cref{lemma:delta-moment-bounds}. Combining these bounds and summing over the \(n\) terms gives the result.
\end{proof}

\begin{lemma}[Denominator anti-concentration]
    \label{lemma:denom-anti-concentration-regular}
    Suppose that the moment bounds of \Cref{thm:feasible-local-power} and \Cref{assm:balanced-design} hold. Then, for any sequence \(\delta_n \searrow 0\),
    \[
        \Pr(D \leq \delta_n) \to 0
    \]
\end{lemma}
\begin{proof}[Proof of \Cref{lemma:denom-anti-concentration-regular}]
  Let \(\tilde \varphi(\cdot): \SR \to \SR\) be three times continuously differentiable with bounded derivatives up to the third order such that \(\tilde\varphi(x)\) is 1 if \(x \leq 0\), \(\tilde\varphi(x)\) is decreasing if \(x \in (0,1)\), and \(\tilde\varphi(x)\) is zero if \(x \geq 1\). 
    Consider a second sequence \(\gamma_n \searrow 0\) slowly enough such that \((\gamma_n^{-2} + \gamma_n^{-3})/\sqrt{n} \to 0\). Take \(\varphi_n(x) = \tilde\varphi(\frac{x- \delta_n}{\gamma_n})\).
    By \Cref{lemma:denominator-interpolation} and since \(\tilde\varphi(\cdot)\) has bounded derivatives up to the third order, there is a fixed constant \(M_1 > 0\) that depends only on \(c\) such that 
    \begin{align*}
        \Pr(D \leq \delta_n) 
        &\leq \Pr(\tilde D \leq \delta_n + \gamma_n) + \frac{M_1}{\sqrt{n}}(\gamma_n^{-2} + \gamma_n^{-3})
    \end{align*}
    Let \(\gamma_n\) be a sequence tending to zero such that \((\gamma_n^{-2} + \gamma_n^{-3})/\sqrt{n} \to 0\) and conclude by applying \Cref{lemma:denom-anti-concentration}.
\end{proof}

\begin{lemma}[]
    \label{lemma:probability-to-distribution}
    Let \(X_n\), \(Y_n\), and \(Z_n\) be sequences of random variables such that \(|X_n - Y_n| \to_p 0\),  the distribution of \(Z_n\) is absolutely continuous with respect to Lebesgue measure and the density functions of \(Z_n\) are uniformly bounded and \(\sup_{a\in\SR}|\Pr(Y_n \leq a) - \Pr(Z_n \leq a)|\to 0.\) Then \(\sup_{a \in \SR}|\Pr(X_n \leq a) - \Pr(Z_n \leq a)|\to 0
    .\)
\end{lemma}
\begin{proof}
   For any \(a \in \SR\) and \(\eps > 0\) we have that \(\{X_n \leq a\} \subseteq \{Y_n \leq a + \eps\} \cup \{|X_n - Y_n| > \eps\}\); thus, by applying union bound and rearranging we obtain:
   \begin{align*}
       \Pr(X_n \leq a) &\leq \Pr(Y_n \leq a + \eps) + \Pr(|Y_n - X_n| > \eps) \\ 
                       &\leq \Pr(Z_n \leq a + \eps) + |\Pr(Y_n \leq a + \eps) - \Pr(Z_n \leq a + \eps)| \\ 
                       &\;\;\;\;\;\;\;\;\;+ \Pr(|Y_n - X_n| > \eps) 
        \intertext{so that}
        \Pr(X_n \leq a) - \Pr(Z_n \leq a) 
                       &\leq \Pr(a < Z_n \leq a+\eps) + |\Pr(Y_n \leq a + \eps) - \Pr(Z_n \leq a + \eps)|\\ 
                       &\;\;\;\;\;\;\;\;\;+ \Pr(|Y_n - X_n| >\eps)
   \end{align*}
   Let \(\eps_n \to 0\) be a sequence tending to zero such that \(\Pr(|X_n - Y_n| > \eps_n) \to 0\) (\Cref{lemma:op1-sequence}). Applying a supremum to the above display yields
   \begin{align*}
       \sup_{a\in \SR} \Pr(X_n \leq a) - \Pr(Z_n \leq a) 
       &\leq \sup_{a\in \SR} \Pr(a < Z_n \leq a + \eps_n) \\ 
       &\;+ \sup_{a \in \SR} |\Pr(Y_n \leq a + \eps_n) - \Pr(Z_n \leq a + \eps_n)|\\ 
       &\;\;+ \Pr(|Y_n - X_n| >\eps_n)
   \end{align*}
   The first term goes to zero as \(\eps_n \to 0\) since \(Z_n\) has a uniformly bounded density; the second term goes to zero by \(\sup_{a \in \SR} |\Pr(Y_n \leq a) - \Pr(Z_n \leq a)| \to 0\) and the third term goes to zero by definition of \(\eps_n\) and \(|Y_n - X_n| \to_p 0\).

   We can apply a symmetric argument to show that \(\sup_{a\in\SR}\Pr(Z_n \leq a) - \Pr(X_n \leq a) \leq o(1)\) which completes the claim of the lemma.
\end{proof}

\section{Proof of \Cref{thm:consistency}}
\begin{proof}[Proof of \Cref{thm:consistency}]
    As at the top of \Cref{subsec:local-power-proof}, recall that \(\tilde h_{ii} = 0\), and define
    \begin{align*}
        N &= \frac{1}{\sqrt{n}}\sum_{i=1}^n \eps_i(\beta_0)\sum_{j=1}^n \tilde h_{ij} r_j &
        D &= \frac{1}{n}\sum_{i=1}^n \eps_i^2(\beta_0)(\sum_{j=1}^n \tilde h_{ij} r_j)^2
    \end{align*}
    where \(\tilde h_{ij} = s_n h_{ij}\). A primary goal is to show that tests based on the infeasible statistic, \(\JK_I(\beta_0)\), are consistent. That is, \(\Pr(\JK_I(\beta_0) \leq a) \to 0\) for any fixed \(a \in \SR_+\). The event \(\{\JK_I(\beta_0) \leq a\}\) is equivalently expressed \(\{N^2 - a D \leq 0\}\) so that \(\Pr(\JK(\beta_0) \leq a) = \Pr(N^2 - aD \leq 0)\).
    Under the moment bounds of \Cref{thm:feasible-local-power} and \Cref{assm:balanced-design}, \(aD = O_p(1)\) so by \Cref{lemma:difference-op1} it suffices to show that \(\Pr(|N| \leq M) \to 0\) for any fixed \(M \geq 0\). By assumption \(P = \E[N^2] \to \infty\) so we move to show that \(\Var(N) = O(1)\) and then apply \Cref{lemma:second-diverges-variance-bounded} to conclude. To this end, recall the definition of \(\eta_i = \eps_i(\beta_0) - \E[\eps_i(\beta_0)]\), define \(\mu_i= \E[\eps_i(\beta_0)] = \Pi_i(\beta - \beta_0)\), and let
    \begin{align*}
        N_1 &\coloneqq \frac{1}{\sqrt{n}}\sum_{i=1}^n \eta_i \sum_{j=1}^n \tilde h_{ij} r_j & 
        N_2 &\coloneqq \frac{1}{\sqrt{n}}\sum_{i=1}^n \mu_i \sum_{j=1}^n \tilde h_{ij} r_j
    \end{align*}
    Notice that \(N = N_1 + N_2\).  To show that \(\Var(N_1) = O(1)\), define \(\va_i = \eta_i \sum_{j=1}^n  \tilde h_{ij} r_j\). Since \(\E[\eta_i r_i] = 0\), we have that \(\Cov(\va_i,\va_j) = 0\) for \(i \neq j\). Thus,
    \[
        \Var(N_1) = \Var(\sum_{i=1}^n \va_i/\sqrt{n}) = n^{-1}\sum_{i=1}^n \Var(\va_i) = n^{-1}\sum_{i=1}^n  \Var(\eta_i)\E[(\sum_{j=1}^n \tilde h_{ij} r_j)^2] \leq c^2
    \]
    where the final inequality follows from the upper bound on \(\Var(\eta_i)\) and by definition of \(\tilde h_{ij} = s_n h_{ij}\) from \Cref{assm:balanced-design}. 

    To show that \(\Var(N_2) = O(1)\) let \(\vb_i = \sum_{j=1}^n \tilde h_{ji}\tilde\Pi_j(\beta - \beta_0)\) and rewrite \(N_2 = \frac{1}{\sqrt{n}}\sum_{i=1}^n r_i \vb_i\). Under \Cref{assm:local-identification}(ii), \(|\vb_i| = |\E[\sum_{j=1}^n \tilde h_{ji}\eps_j(\beta_0)]| \leq c^{1/2}\), so we can bound
    \[
        \Var(N_2) = \Var(\sum_{i=1}^n  r_i \vb_i /\sqrt{n}) = n^{-1}\sum_{i=1}^n \vb_i^2  \Var(r_i) \leq c^2
    \] 
    Since \(\Var(N) \leq 2\Var(N_1) + 2\Var(N_2)\), we can conclude that tests based on \(\JK_I(\beta_0)\) are consistent.

    Finally, we want to show that this fact, along with \((\Delta_N,\Delta_D)\to_p0\) implies that tests based on \(\JK(\beta_0)\) are consistent. To do this, notice that we can write 
    \[
        \JK(\beta_0) = \frac{(N + \Delta_N)^2}{D + \Delta_D}
    \]
    and thus that \(\JK(\beta_0) \leq a\) if and only if
    \[
        \widehat\JK_a \coloneqq (N + \Delta_N)^2 - a(D + \Delta_D)  = N^2 - aD + 2N\Delta_N + \Delta_N^2 - a\Delta_D \leq 0.
    \]
    Define \(\JK_a = N^2 - aD\). Using that \(\{\widehat\JK_a \leq 0\} \subseteq \Big\{\frac{\widehat\JK_a}{\JK_a}\JK_a \leq 0\Big\} \cup \big\{\JK_a \leq 0\}\) we can write 
    \begin{align*}
        \Pr(\widehat\JK_a \leq 0) 
        &\leq \Pr\Big(\frac{\widehat\JK_a}{\JK_a} \JK_a\leq 0\Bigg) + \Pr(\JK_a \leq 0)  \\
        &\leq 2\Pr(\JK_a \leq 0) + \Pr\Big(\frac{\widehat\JK_a}{JK_a} \leq \frac{1}{2}\Big)
    \end{align*}
    By consistency of the test based on the infeasible \(\JK_I(\beta_0)\) statistic, we have that \(\Pr(\JK_a \leq 0)\to 0\). Thus, it only remains to show that \(\Pr(\widehat\JK_a /\JK_a \leq 1/2) \to 0\). This, in turn, follows if
    \[
        \frac{\widehat\JK_a - \JK_a}{\JK_a} = \frac{2N\Delta_N + \Delta_N^2 - a\Delta_D}{N^2 - aD}\to_p 0.
    \]
    The above results can be used to show that \(\Pr(|N^2 - aD| \leq \delta_n) \to 0\) for any sequence \(\delta_n \searrow 0\) so that \(\frac{1}{\JK_a} = O_p(1)\). Combined with \((\Delta_N,\Delta_D)\to_p0\) this implies that \(\{\Delta_N^2 - a\Delta_D\}/\JK_a \to_p0\). What remains is to show that \(2N\Delta_N/(N^2 - aD) \to_p 0\). Write
    \[
        \frac{2N\Delta_N }{N^2 - aD} = \frac{\frac{2N\Delta_N}{N^2}}{1 - a\frac{D}{N^2}}.
    \]
    Since \(D = O_p(1)\) while \(\Pr(N^2 \leq M) \to 0\) for any \(M\) we have that \(D/N^2 \to_p 0\). Moreover, \(\Pr(|N| \leq M) \to 0\) for any fixed M implies \(N/N^2 = O_p(1)\) so that \(2N\Delta_N/N^2 \to_p0\). We can apply continuous mapping theorem to conclude.
\end{proof}
\begin{lemma}[]
    \label{lemma:second-diverges-variance-bounded}
    Suppose that \(X_n\) is a sequence of random variables such that \(\E[X_n^2] \to \infty\) while \(\Var(X_n) = O(1)\). Then, for any \(M \geq 0\), \(\Pr(|X_n| \leq M) \to 0\).
\end{lemma}
\begin{proof}
    First, note that \(\Var(|X_n|) \leq \Var(X_n)\) so \(\Var(|X_n|) = O(1)\). Moreover \(\Var(|X_n|) = \E[X_n^2] - (\E[|X_n|])^2\), so \(\E[X_n^2] \to \infty\) and \(\Var(|X_n|) = O(1)\) implies that \(\E[|X_n|] \to \infty\). Then, 
    \begin{align*}
        \Pr(|X_n| \leq M) 
        &= \Pr(|X_n| - \E[|X_n|] \leq M - \E[|X_n|]) \\ 
        &= \Pr(\E[|X_n|] - |X_n| \geq \E[|X_n| - M) \\ 
        &\leq \Pr(|\E[|X_n|] - |X_n|| \geq \E[|X_n|] - M) \\ 
        &\leq \frac{\Var(|X_n|)}{\E[|X_n|] - M}
    \end{align*}
    Since \(\Var(|X_n|) = O(1)\) but \(\E[|X_n|] \to \infty\), this tends to zero.
\end{proof}

\begin{lemma}[]
    \label{lemma:difference-op1}
    Suppose that \(X_n\) and \(Y_n\) are random variables such that \(Y_n = O_p(1)\) and, for any \(M \geq 0\), \(\Pr(|X_n|\leq M) \to 0\). Then, for any \(M_1 \geq 0\), \(\Pr(X_n^2 - Y_n \leq M_1) \to 0\).
\end{lemma}
\begin{proof}
    Pick any \(\eps > 0\). We want to show that, eventually, \(\Pr(X_n^2 - Y_n > M_1) \geq 1- \eps\). Since \(Y_n = O_p(1)\), there is a fixed constant \(M_Y\) such that \(\Pr(|Y_n| \leq M_Y) \geq 1-\eps/2\). Since \(\Pr(|X_n| \leq M)\to 0\) for any \(M \geq 0\), there exists an \(N_X\) such that, for \(n \geq N_X\), \(\Pr(X_n^2 \leq M_1 + M_Y) \leq \eps/2\). A union bound completes the argument (on the eventuality \(n \geq N_X\)):
    \begin{align*}
        \Pr(X_n^2 - Y_n > M) 
        &\geq \Pr(X_n^2 > M_1 + M_Y, |Y_n| \leq M_Y)\\ 
        &= 1 - \Pr(\{X_n^2 < M_1 + M_Y\} \cup \{|Y_n| > M_Y\}) \\
        &\geq 1 - \eps/2 - \eps/2 = 1-\eps
    \end{align*}
\end{proof}

\section{Proof of \Cref{thm:estimation-error}}
\label{subsec:estimation-error-ell1-proof}

I provide the proof for the case where \(d_x = 1\) with the general case following from symmetric logic. For any \(j = 1,\dots, d_b\) define the matrix \(B_j = \diag(b_j(z_1),\dots,b_j(z_n))\) and collect observations \(\eps(\beta_0) = (\eps_1(\beta_0),\dots,\eps_n(\beta_0))'\in \SR^n\), \(r = (r_1,\dots,r_n)'\in\SR^n\), \(\hat r = (\hat r_1,\dots, \hat r_n)'\in\SR^n\), and \(\xi = (\xi_1,\dots,\xi_n)' \in \SR^n\). In addition, collect \(b_\eps = (b_{\eps1},\dots,b_{\eps n})\in \SR^{d_b\times n}\) where \(b_{\eps i} = \eps_i(\beta_0)b(z_i) \in \SR^{d_b}\). Finally, let \(\mH = \frac{s_n}{\sqrt{n}}H\), \(\tilde H = s_n H\) and \(\tilde h_{ij} = s_n h_{ij}\).

\textit{Step 1: \(\Delta_N\to_p 0\).} To show that \(\Delta_N \to_p 0\) write
\begin{align*}
    \Delta_N 
    &= |\eps(\beta_0)'\mH(\hat r - r)| \\ 
    &= |\eps(\beta_0)'\mH(b_\eps'\hat\gamma - b_\eps'\gamma) - \eps(\beta_0)'\mH\xi|\\ 
    &\leq \underbrace{\max_{1 \leq j \leq d_b}|\eps(\beta_0)'\mH B_j\eps(\beta_0)|\|\hat\gamma-\gamma\|_1}_{\vA} + \underbrace{\|\eps(\beta_0)'\mH\|_2\|\xi_2\|_2}_{\vB}
\end{align*}
To bound \(\vA\) we move to apply \Cref{thm:hanson-wright} to the quadratic form \(\eps(\beta_0)'(\mH B_j)\eps(\beta_0)\). First notice that, under \Cref{assm:estimation-error}(v), we have 
\[
    \|\E[\mH b_j\eps(\beta_0)]\|_2 = \frac{1}{n}\sum_{i=1}^n (\E[s_n\sum_{j\neq i} h_{ij} b(z_j)\eps_j(\beta_0)])^2 \leq c^2 
\] 
In the notation of \Cref{thm:hanson-wright} this give us an upper bound on \(\|\E f^{(1)}(X)\|_{\text{HS}}\). Next, \Cref{assm:balanced-design} gives us that the Frobenius norm of \(\mH = \frac{s_n}{\sqrt{n}}\mH\) is bounded, since the rows of \(s_n H\) are square summable, \(\sum_{j\neq i} (s_n h_{ij})^2 \leq c\) for all \(i = 1,\dots, n\). In the notation of \Cref{thm:hanson-wright} this gives us an upper bound on \(\|\E f^{(2)}(X)\|_{\text{HS}}\). Applying \Cref{thm:hanson-wright} and a union bound then gives us that
\begin{equation}
    \label{eq:hanson-wright-application1}
    \max_{1 \leq j \leq d_b} |\eps(\beta_0)'\mH B_j \eps(\beta_0) - \E[\eps(\beta_0)'\mH B_j \eps(\beta_0)]| = O_p(\log^{2/a}(d_b))
\end{equation}
Since \(\max_{1 \leq j \leq d_b}|\E[\eps(\beta_0)'\mH B_j\eps(\beta_0)]| \leq c\) under \Cref{assm:estimation-error}(v), \eqref{eq:hanson-wright-application1} gives that
\[
    \max_{1 \leq j \leq d_b} |\eps(\beta_0)'\mH B_j \eps(\beta_0)| = O_p(\log^{2/a}(d_b))
\]
Since \(\log^{2/a}(d_b)\|\widehat\gamma - \gamma\|_1 \to_p 0\) by assumption, this yields that \(\vA \to_p 0\).

To bound \(\vB\) see that \(\|\eps(\beta_0)'\mH\|_2 = \frac{s_n^2}{n}\sum_{i=1}^n (\sum_{j\neq i} h_{ij}\eps_i(\beta_0))^2 = O_p(1)\) under \Cref{assm:local-identification}(ii) while under \Cref{assm:estimation-error} \(\|\xi\|_2 = o(1)\).

\textit{Step 2: \(\Delta_D\to_p 0\).} Notice that \(a^2 - b^2 = 2b(a - b) + (a - b)^2\) and bound:
\begin{align*}
    |\Delta_D| 
    &\leq \underbrace{\frac{1}{n}\sum_{i=1}^n \eps_i^2(\beta_0)\big|\sum_{j\neq i}  \tilde h_{ij} r_j\big|}_{\vE}\times \max_i |\sum_{j\neq i} \tilde h_{ij}(\hat r_j - r_j)| \\ 
    &+ \underbrace{\frac{1}{n}\sum_{i=1}^n \eps_i^2(\beta_0)}_{\vF}\times \max_i |\sum_{j\neq i} \tilde h_{ij} (\hat r_j - r_j)|^2
\end{align*}
Since both \(\vE = O_p(1)\) and \(\vF = O_p(1)\) under the moment bounds of \Cref{thm:feasible-local-power} and \Cref{assm:balanced-design}, it suffices to show that 
\[
    \max_i |\sum_{j\neq i} \tilde h_{ij} (\widehat r_j - r_j)| \to_p 0
\]
To do so write
\begin{align*}
    \max_i \big|\sum_{j\neq i} \tilde h_{ij}\{\hat r_j - r_j\}\big| 
    \leq 
    \underbrace{\max_{\substack{1 \leq i \leq n \\ 1 \leq j \leq d_b}}\big|\sum_{j\neq i} \tilde h_{ij} b(z_j)\eps_j(\beta_0)\big|\|\hat\gamma - \gamma\|_1}_{\vA} + \underbrace{\max_{\substack{1 \leq i \leq n \\ 1 \leq j \leq d_b}}\big|\sum_{j\neq i} \tilde h_{ij} b(z_j)\xi_j\big|}_{\vB}
\end{align*}
To bound \(\vA\), note that by \Cref{assm:estimation-error}(v) \(\max_{i,j}|\E[\sum_{j\neq i} \tilde h_{ij} b(z_j) \eps_j(\beta_0)| \leq c\). Under \Cref{assm:balanced-design,assm:estimation-error}(ii), \(\max_{i,j} \sum_{j\neq i} \tilde h_{ij}^2 b^2(z_j) \leq c^2\) so we can apply \Cref{thm:hanson-wright} and a union bound to obtain that
\[
    \max_{\substack{1 \leq i \leq n \\ 1 \leq j \leq d_b}} \big|\sum_{j\neq i} \tilde h_{ij} b(z_j)\eps_j(\beta_0)\big| = O_p(\log^{1/a}(d_bn))
\]
Along with the implied rate on \(\|\hat\gamma - \gamma\|_1\) from \Cref{assm:estimation-error}(iv) this shows that \(\vA \to_p 0\).

To show that \(\vB \to 0\), use Cauchy-Schwarz, \(\sum_{j\neq i}  \tilde h_{ij}^2b^2(z_j) \leq c\) for  any \(i,j\) by \Cref{assm:balanced-design,assm:estimation-error}(ii), and \(\sum_{i=1}^n \xi_i^2 = o(1)\) by  \Cref{assm:estimation-error}(iii).

\section{Proof of \Cref{thm:multiple-feasible-local-power}}
\label{sec:multiple-proofs}
Throughout this section, define the scaled elements of the infeasible and gaussian numerators and denominators
\small
\begin{align*}
    N_\ell &= \frac{s_{n,\ell}}{\sqrt{n}}\sum_{i=1}^n \eps_i(\beta_0)\sum_{j=1}^n h_{ij}r_j &
    \tilde N_\ell &= \frac{s_{n,\ell}}{\sqrt{n}}\sum_{i=1}^n \tilde\eps_i(\beta_0)\sum_{j=1}^n h_{ij}\tilde r_j \\ 
    D_{\ell k} &= \frac{s_{\ell,n}s_{m,k}}{n}\sum_{i=1}^n \eps_i^2(\beta_0)(\sum_{j=1}^n h_{ij}r_{\ell j})(\sum_{j=1}^n h_{ij} r_{kj}) & 
    \tilde D_{\ell k} &= \frac{s_{\ell,n}s_{m,k}}{n}\sum_{i=1}^n \eps_i^2(\beta_0)(\sum_{j=1}^n h_{ij}\tilde r_{\ell j})(\sum_{j=1}^n h_{ij} \tilde r_{kj})
\end{align*}
\normalsize
Collect these in \(N = (N_1,\dots N_{d_x})' \in \SR^{d_x}\), \(\tilde N = (\tilde N_1,\dots, \tilde N_{d_x})' \in \SR^{d_x}\), \(D = [D_{\ell k}]_{\ell, k \in [d_x]} \in \SR^{d_x \times d_x}\), and  \(\tilde D = [\tilde D_{\ell k}]_{\ell, k \in [d_x]} \in \SR^{d_x \times d_x}\). After multiplying by scaling matrix \(\diag(s_{1,n},\dots,s_{d_x,n})\) and the inverse of the scaling matrix we rewrite the infeasible and gaussian test statistics
\begin{align*}
    \JK_I(\beta_0) &= N'D^{-1}N \bm{1}_{\{\lambda_{\min}(D) > 0\}} &  \JK_G(\beta_0) &= \tilde N'\tilde D^{-1}\tilde N
\end{align*}

As with \Cref{thm:feasible-local-power}, the result \Cref{thm:multiple-feasible-local-power} follows directly from combining the following lemmas. The first is the main technical lemmma, and shows that the distribution of the infeasible statistic \(\JK_I(\beta_0)\) can be uniformly approximated by that of \(\JK_G(\beta_0)\). The proof of this technical lemma is involved and deferred to \Cref{sec:joint-proofs}. The second lemma establishes that estimation error can be treated as negligible. As with \Cref{lemma:estimation-error-high-level}, the main difficulty hear is in dealing with the fact that neither the numerator vector nor denominator matrix of the \(\JK(\beta_0)\) statistic may have stable limiting distributions.

\begin{lemma}[]
    \label{thm:multiple-infeasible-local-power}
    Suppose that \Cref{assm:balanced-design,assm:local-identification} hold as well as the moment conditions of \Cref{thm:multiple-feasible-local-power}. Then, 
    \[
        \sup_{a\in\SR}\big|\Pr(\JK_I(\beta_0) \leq a) - \Pr(\JK_G(\beta_0) \leq a)\big| \to 0.
    \]
\end{lemma}
\begin{proof}[Proof of \Cref{thm:multiple-infeasible-local-power}]
    \Cref{thm:multiple-infeasible-local-power} follows as a consequence of the joint gaussian approximation with the combination statistic established in \Cref{sec:joint-proofs}.
\end{proof}

\begin{lemma}[]
    \label{lemma:multiple-estimation-error-high-level}
    Suppose that \Cref{assm:balanced-design,assm:local-identification} hold along with the moment conditions of \Cref{thm:multiple-feasible-local-power}. Then, if \((\Delta_N, \Delta_D) \to_p0\), \(\big|\JK(\beta_0) -\JK_I(\beta_0)\big|\to_p0\). 
\end{lemma}
\begin{proof}[Proof of \Cref{lemma:multiple-estimation-error-high-level}]
    Define the matrix \(\Delta_{D} = [(\Delta_{D})_{\ell k}]_{\ell,k\in [d_x]}\) and the vector \(\Delta_N = [(\Delta_N)_{\ell}]_{\ell \in [d_x]}\) where 
\begin{align*}
    (\Delta_D)_{\ell k} 
    &\coloneqq \frac{s_{\ell,n}s_{k,n}}{n}\sum_{i=1}^n \eps_i^2(\beta_0)\big(\widehat\Pi_{\ell, i} \widehat\Pi_{k,i} - \widehat\Pi_{\ell,i}^I\widehat\Pi_{k,i}^I\big) \\ 
    (\Delta_N)_\ell &\coloneqq \frac{s_{\ell, n}}{\sqrt{n}}\sum_{i=1}^n \eps_i(\beta_0)(\widehat\Pi_{\ell,i} - \widehat\Pi^I_{\ell, i})
\end{align*}
By assumption we have that \(\|\Delta_D\| \to_p 0\) and \(\|\Delta_N\|\to_p 0\).
Using this notation, we can write the infeasible version of the test statistic as \(\JK^I(\beta_0) = N'D^{-1}N\) while the feasible version is written \(\JK(\beta_0) = (N + \Delta_N)'(D + \Delta_D)^{-1}(N + \Delta_N)\). Add and subtract \(D^{-1}\) to get
\begin{align*}
    \JK(\beta_0) 
    &= \big(N + \Delta_N\big)'\big( (D + \Delta_D)^{-1} \pm D^{-1}\big)\big(N + \Delta_N\big) \\ 
    &= \JK^I(\beta_0) + N'\big( (D + \Delta_D)^{-1} - D^{-1}\big)N + \Delta_N\big((D + \Delta_D)^{-1} - D^{-1})N\\ 
    &\;\;+ \Delta_N'\big((D + \Delta_D)^{-1} - D^{-1}\big)\Delta_N + N'D^{-1}\Delta_N + \Delta_ND^{-1}N  + \Delta_ND^{-1}\Delta_N
\end{align*}
Via \Cref{lemma:denominator-anticoncentration} we have that \(\|D^{-1}\| = (\lambda_{\min}(D))^{-1} =  O_p(1)\) and by assumption we have that \(\Delta_N \to_p 0\). It therefore suffices to show that
\begin{equation}
    \label{eq:inverse-difference}
    \|(D + \Delta_D)^{-1} - D^{-1}\| \to_p 0
\end{equation}
To do so, we can use the following equality from \citet{horn2012matrix}, p. 381.
\[
    \|(D + \Delta_D)^{-1} - D^{-1}\| \leq  \frac{\|D^{-1}\|^2\|\Delta_D\|}{1 - \|D^{-1}\Delta_D\|}
\]
Since \(\|D^{-1}\| = O_p(1)\) and \(\Delta_D \to_p 0\), this gives \eqref{eq:inverse-difference}.
\end{proof}

\section{Proofs of Results in \Cref{sec:power-properties}}
\label{sec:combo-proofs}
The statement of \Cref{thm:joint-combination} relies on showing
\begin{align*}
        &\sup_{(a_1,a_2)\in \SR^2} \big|\Pr(\JK(\beta_0) \leq a_1, C \leq a_2) - \Pr(\JK_G(\beta_0) \leq a_1, C_G \leq a_2)\big| \to 0\\
    \andbox
        &\sup_{(a_1,a_2)\in \SR^2} \big|\Pr(\S(\beta_0) \leq a_1, C \leq a_2) - \Pr(\S_G(\beta_0) \leq a_1, C_G \leq a_2)\big| \to 0
\end{align*}
In particular, since \((\JK_G(\beta_0) \perp C_G)\)  and \((\S_G(\beta_0)\perp C_G)\) under \(H_0\), showing the above will imply the test based on \(T(\beta_0;\tau)\) has asymptotic size \(\alpha\) for any choice of cutoff \(\tau\).
The second line in the above display follows imediately from \Cref{prop:hyperrectangle} after verifying \Cref{assm:rectangle-gaussian-approximation-estimation}, below.

The first line in the top display relies on a joint interpolation of the infeasible \(\JK_I(\beta_0)\) test statistic and the infeasible conditioning statistic \(C_I\), which could be constructed if \(\rho(z_i)\) was known to the researcher.
\begin{equation}
    \label{eq:infeasible-conditioning}
    C_I \coloneqq \max_{1 \leq i \leq n} \big|\frac{1}{\sqrt{n}} \sum_{i=1}^n h_{ij}r_j \big/ (n^{-1}\sum_{i=1}^n h_{ij}^2)^{1/2}\big|
\end{equation}
This joint interpolation argument is rather involved however, and deferred to \Cref{sec:joint-proofs}. The interpolation argument for the conditioning statistic very closely follows the results in \citet{cck2013}. The results of \Cref{sec:power-properties} rely on showing that the difference between \(C\) and \(C_I\) can be treated as negligible. This in turn reduces to verifying \Cref{assm:rectangle-gaussian-approximation-estimation}, which is done in \Cref{lemma:combination-error-negligible}, below.

\begin{lemma}[]
    \label{lemma:combination-error-negligible}
    Suppose that \Cref{assm:estimation-error} holds. Then there are sequences \(\delta_n \searrow 0\), \(\beta_n \searrow 0\) such that 
    \[
        \Pr\big(\max_{i \in [n]} n^{-1}\sum_{j=1}^n \dot h_{ij}^2(\widehat r_j - r_j)^2 > \delta_n^2 /\log^2(n)\big) \leq \beta_n
    \]
    where \(\dot h_{ij} = h_{ij}/(n^{-1}\sum_{j=1}^n h_{ij}^2)^{1/2}\). 
\end{lemma}
\begin{proof}
In view of \Cref{lemma:op1-sequence} it suffices to show
\begin{equation}
    \label{eq:combination-error-condition}
    \max_{1 \leq i \leq n} \frac{1}{n}\sum_{j=1}^n\dot h_{ij}^2 (\hat r_i - r_i)^2  = o_p(1/\log^2(n))
\end{equation}
Notice that we can bound
\begin{align*}
    \max_{1 \leq i \leq n} \frac{1}{n}\sum_{j=1}^n (\hat r_i - r_i)^2
    &= \max_{1 \leq i \leq n}\big|(\widehat\gamma - \gamma)'n^{-1}\sum_{j=1}^n\eps_j^2(\beta_0)b(z_i)b(z_j)'(\widehat\gamma - \gamma)\big| \\ 
    &\;\;\; + \max_{1 \leq i \leq n} | n^{-1}\sum_{j=1}^n \dot h_{ij}^2\xi_j^2 |\\
    &\leq \max_{\substack{1 \leq i \leq n \\ 1 \leq j,k \leq d_b}}\underbrace{\big|n^{-1}\sum_{j=1}^n \eps_j^2(\beta_0)b_j(z_j)b_{k}(z_j)\big|}_{\vA_{ijk}}\|\widehat\gamma - \gamma\|_1^2 \\ 
    &\;\; + n^{-1/2}\max_{1 \leq i \leq n}(n^{-1}\sum_{j=1}^n \dot h_{ij}^4)^{1/2}(\sum_{j=1}^n \xi_j^4)^{1/2}
\end{align*}
Under \Cref{assm:estimation-error}(i,ii) each \(\vA_{ijk}\) is \(\upsilon\)-sub-exponential by \Cref{thm:hanson-wright} (that is \(\|\vA_{ijk}\|_{\psi_\upsilon}\) is bounded). An application of \Cref{lemma:sub-exponential-power-max-bound} then yields that \(\max_{i,j,k}|\vA_{ijk}| = O_p(\log^{1/\nu}(d_bn))\). Along with \Cref{assm:estimation-error}(iv) this gives that \(\max_{i,j,k}|\vA_{ijk}|\|\widehat\gamma - \gamma\|_1 = O_p(\log^{-3/(v\wedge 1)}(d_bn)) = o_p(\log^{-2}(n))\). Meanwhile by definition of \(\dot h_{ij}\), \(\max_i (n^{-1}\sum_{j=1}^n \dot h_{ij}^4)^{1/2} = O(1)\) while by \Cref{assm:estimation-error}(iii) \((\sum_{j=1}^n \xi_j^4)^{1/2} = o(1)\). Since \(\log^2(n)/\sqrt{n} \to 0\) this shows \eqref{eq:combination-error-condition}.
\end{proof}

\subsection{Proof of \Cref{thm:joint-combination}}

The first result in \Cref{thm:joint-combination} with \(\JK(\beta_0)\) and \(C\) replaced with their infeasible analogs \(\JK_I(\beta_0)\) and \(C_I\) follows from the argument in \Cref{sec:joint-proofs}. After verifying that \(|\JK(\beta_0) - \JK(\beta_0)| \to_p 0\) via \Cref{thm:estimation-error} and that \Cref{assm:rectangle-gaussian-approximation-estimation} is satisfied via \Cref{lemma:combination-error-negligible} follow the same steps as in the proof of \citet{belloni2018highdimensional}, {Theorem 2.1} to see that approximation result holds for the feasible \(\JK(\beta_0)\) and \(C\).

For the second statement, I show that the conditions of \Cref{thm:hyperrectangle-bootstrap} are satisfied. To see that \Cref{assm:rectangle-gaussian-approximation}(i,ii) is satisfied under the moment assumptions of \Cref{thm:feasible-local-power} use (i) the definition of \(\dot h_{ij} = h_{ij}\big/(n^{-1}\sum_{j=1}^n h_{ij}^2)^{1/2}\); (ii) that the variance of each \(r_{j}\) is bounded away from zero and (iii) that the fourth moments of \(r_j\) are bounded from above. \Cref{assm:rectangle-gaussian-approximation}(iii) is satisfied with \(B_n = \log^{1/\upsilon}(n)\) by \Cref{assm:comb-conditions}(i,iii) and \Cref{lemma:sub-exponential-power-max-bound}. Finally \Cref{assm:rectangle-gaussian-approximation-estimation} is satisfied by applying \Cref{lemma:combination-error-negligible}. Apply \Cref{thm:hyperrectangle-bootstrap} to conclude.

\section{Joint Gaussian Approximation of \(\JK(\beta_0)\) and \(C\)}
\label{sec:joint-proofs}
 \Cref{thm:multiple-feasible-local-power,thm:joint-combination} rely on a joint interpolation of the conditioning and testing statistics as well as a joint interpolation of the conditioning and testing statistics. The joint interpolation of \(\JK(\beta_0)\) and the conditioning statistic \(C\) is given in \Cref{subsec:joint-main} after introducing some notation in \Cref{sec:joint-notation}. The joint gaussian approximation of \(S(\beta_0)\) and \(C\) follows immediately from results in \cites{belloni2018highdimensional,CCK-2018-HDCLT}. The result is presented below for the general form of the \(\JK(\beta_0)\) statistic under \(H_0\) however the proof strategy is very similar when using the decomposed form of \(\JK(\beta_0)\) when \(d_x = 1\). This proof is available on request.

\subsection{Notation}
\label{sec:joint-notation}
\paragraph{Jackknife Statistic Definitions.} Define \(\tilde h_{\ell, ij} = s_{n,\ell} h_{ij}\) for each \(\ell = 1,\dots,d_x\) and the scaled leave-one-out quasi-numerator and denominators
\begin{align*}
    U_{-i} &= \bigg[ \frac{1}{\sqrt{n}}\sum_{j=1}^n \dot\eps_j(\beta_0)\sum_{k \neq i} \tilde h_{\ell, jk} \dot r_{\ell k}\bigg]_{1 \leq \ell \leq d_x} \in \SR^{d_x} \\
    D_{-i} &= \bigg[\frac{1}{n} \sum_{j=1}^n \ddot\eps_i^2(\beta_0)\big(\sum_{k \neq i} \tilde h_{\ell, ij} \dot r_{\ell j}\big)\big(\sum_{k \neq i} \tilde h_{\ell,ij}\dot r_{m j}\big)\bigg]_{\substack{1 \leq \ell \leq d \\ 1 \leq m \leq d_x}} \in \SR^{d_x\times d_x} \\ 
\end{align*}
where \(\dot\eps_j(\beta_0)\) is equal to \(\tilde\eps_j(\beta_0)\) if \(j < i\) and equal to \(\eps_j(\beta_0)\) if \(j > i\), \(\dot r_{\ell j}\) is equal to \(\tilde r_{\ell j}\) if \(j < i\) and equal to \(r_j\) if \(j > i\), and \(\ddot \eps_j(\beta_0)\) is equal to \(\E[\eps_j^2(\beta_0)]\) if \(j < i\) and equal to \(\eps_j(\beta_0)\) if \(j > i\). As in the proof of \Cref{thm:local-power} while the definitions of \(\dot\eps_j(\beta_0), \dot r_{\ell j}\), and \(\ddot\eps_j(\beta_0)\) depend on \(i\) this dependence is suppressed to conslidate notation and since we only consider one step deviations at a time.

Also define the one step deviations
\begin{align*}
    \Delta_{Ui} 
    &= \big[\eps_i(\beta_0)\sum_{j=1}^n \tilde h_{\ell,ij}\dot r_{\ell j} + r_{\ell i}\sum_{j=1}^n \tilde h_{\ell,ji} \dot\eps_j(\beta_0)\big]_{1 \leq \ell \leq d} \in \SR^d \\ 
    \tilde\Delta_{Ui} 
    &= \big[\tilde\eps_i(\beta_0)\sum_{j=1}^n \tilde h_{\ell,ij}\dot r_{\ell j} + \tilde r_{\ell i}\sum_{j=1}^n \tilde h_{\ell,ji} \dot\eps_j(\beta_0)\big]_{1 \leq \ell \leq d} \in \SR^d \\ 
    \Delta_{Di}
    &= \underbrace{\big[(\Delta_{Di}^{a})_{\ell m}\big]_{\substack{1 \leq \ell \leq d \\ 1 \leq m \leq d}}}_{\Delta_{Di}^a} + \underbrace{\big[(\Delta_{Di}^{b})_{\ell m}\big]_{\substack{1 \leq \ell \leq d \\ 1 \leq m \leq d}}}_{\Delta_{Di}^b}\\
    \tilde\Delta_{Di}
    &= \underbrace{\big[(\tilde\Delta_{Di}^{a})_{\ell m}\big]_{\substack{1 \leq \ell \leq d \\ 1 \leq m \leq d}}}_{\tilde\Delta_{Di}^a} + \underbrace{\big[(\tilde\Delta_{Di}^{b})_{\ell m}\big]_{\substack{1 \leq \ell \leq d \\ 1 \leq m \leq d}}}_{\tilde\Delta_{Di}^b}
\end{align*}
where 
\begin{align*}
    (\Delta_{Di}^a)_{\ell m} 
    &= \eps_i^2(\beta_0)\big(\sum_{j=1}^n \tilde h_{\ell,ij} r_{\ell j}\big)\big(\sum_{j=1}^n \tilde h_{\ell, ij} \dot r_{\ell j}\big)\big(\sum_{j=1}^n h_{m, ij} r_{m,  ij}\big)^2 + r_{\ell i} r_{ki}\sum_{j=1}^n  \tilde h_{\ell, ij} \tilde h_{m, ij}\ddot\eps_j^2(\beta_0)\\
    (\tilde\Delta_{Di}^a)_{\ell m} 
    &= \tilde \eps_i^2(\beta_0)\big(\sum_{j=1}^n \tilde h_{\ell,ij} r_{\ell j}\big)\big(\sum_{j=1}^n \tilde h_{\ell, ij} \dot r_{\ell j}\big)\big(\sum_{j=1}^n h_{m, ij} r_{m,  ij}\big)^2 + \tilde r_{\ell i} \tilde r_{ki}\sum_{j=1}^n  \tilde h_{\ell, ij} \tilde h_{m, ij}\ddot\eps_j^2(\beta_0)\\ 
(\Delta_{Di}^{b})_{\ell m} 
    &= r_{\ell i}\sum_{j=1}^n \ddot\eps_j^2(\beta_0)\sum_{k \neq i} \tilde h_{\ell, ji}\tilde h_{m, jk} \dot r_{m k} + r_{ki}\sum_{j=1}^n  \ddot\eps_j^2(\beta_0)\sum_{k \neq i}\tilde h_{\ell, ji}\tilde h_{m, jk}\dot r_{\ell k} \\
(\tilde\Delta_{Di}^{b})_{\ell m} 
    &= \tilde r_{\ell i}\sum_{j=1}^n \ddot\eps_j^2(\beta_0)\sum_{k \neq i} \tilde h_{\ell, ji}\tilde h_{m, jk} \dot r_{m k} + \tilde r_{ki}\sum_{j=1}^n  \ddot\eps_j^2(\beta_0)\sum_{k \neq i}\tilde h_{\ell, ji}\tilde h_{m, jk}\dot r_{\ell k} \\
\end{align*}
Notice that in this notation we can write the test statistic and gaussian test statistics, after scaling by \(\diag(s_{n,1},\dots,s_{n,d_x})\), as
\begin{align*}
    C(\beta_0) &= (U_{-1} + \Delta_{U1}/\sqrt{n})'(D_{-1} + \Delta_{D1}/n)^{-1}(U_{-1} + \Delta_{U1}/\sqrt{n})\bm{1}\{\lambda_{\text{min}}(D_{-1} + \Delta_{D1})^{-1}) > 0\}\\
    \tilde C(\beta_0) &= (U_{-n} + \tilde\Delta_{Un}/\sqrt{n})'(\tilde D_{-1} + \tilde\Delta_{D1}/n)^{-1}(U_{-n} + \tilde\Delta_{U1}/\sqrt{n})
\end{align*}
In this proof we will use these representations for the test statistics. Finally define 
\begin{align*}
    U &= U_{-1} + \Delta_{U1}/\sqrt{n} & \tilde U &= U_{-n} + \tilde\Delta_{Un}/\sqrt{n} \\ 
    D &= D_{-1} + \Delta_{D1}/n & \tilde D &= D_{-n} + \Delta_{Dn}/n
\end{align*}

\paragraph{Conditioning Statistic Definitions.} Let \(h_{\ell, ii} = 0\) for any \(\ell = 1,\dots,d_x\) and \(i= 1,\dots,n\). Define \(\tilde h_{\ell,ij} = h_{\ell,ij}/\omega_{\ell i}\) for \(\omega_{\ell i} = n^{-1}\sum_{j=1}^n |h_{\ell,ij}|\). Also define the one-step deviations:
\begin{align*}
    \Delta_{Ci} &\coloneqq (\tilde h_{1, ji}r_{1i}, -\tilde h_{1, ji}r_{1i},\dots,\tilde h_{d_x,ji}r_{d_xi},-\tilde h_{d_x,ji}r_{d_xi})_{1 \leq j \leq n}' \in \SR^{2nd_x}\\
    \Delta_{Ci} &\coloneqq (\tilde h_{1, ji}\tilde r_{1i}, -\tilde h_{1, ji}\tilde r_{1i},\dots,\tilde h_{d_x,ji}\tilde r_{d_xi},-\tilde h_{d_x,ji}\tilde r_{d_xi})_{1 \leq j \leq n}' \in \SR^{2nd_x}\\
    \intertext{And the leave-one-out vector}
    C_{-i} &\coloneqq \frac{1}{\sqrt{n}}\sum_{j < i} \tilde \Delta_{Cj} + \frac{1}{\sqrt{n}}\sum_{j > i} \Delta_{Cj} \in \SR^{2nd_x}
\end{align*}
Notice that \(C = \max_{1 \leq \iota \leq 2nd_x} (C_{-1} + \frac{1}{\sqrt{n}}\Delta_{C1})_\iota\) while \(\tilde C = \max_{1 \leq \iota \leq 2nd_x}(C_{-n} + \Delta_{Cn})_\iota \).

\paragraph{Function Definitions.}
As in \citet{cck2013} consider the ``smooth max'' function, \(F_\beta:\SR^p \to \SR\) defined
\[
    F_\beta(z) = \beta^{-1}\log\bigg(\sum_{i=1}^n \exp(\beta z_i)\bigg)
\]
which satisfies
\[
    0 \leq F_\beta(z) - \max_{1 \leq i \leq n} z_i \leq \beta^{-1}\log p.
\]
\Cref{subsec:smooth-max-properties} notes some useful properties of the smooth max function which we will use in the joint interpolation argument. In addition let \(\varphi(\cdot) \in C_b^3(\SR)\) be such that \(\varphi(x) = 1\) if \(x \leq 0\), \(\varphi'(x) < 0\) for \(x \in (0,1)\), and \(\varphi(x) = 0\) for \(x \geq 1\). For any \(\gamma > 0\) and \(a = (a_1,a_2)'\in \SR^2\) define the function \(\tilde \varphi(\cdot,\cdot,\cdot): \SR^{d_x} \times \vec(\SR^{d_x \times d_x}) \times \SR^{2nd_x} \to \SR\) via 
\begin{equation}
    \label{eq:joint-tilde-phi}
    \tilde\varphi_{\gamma,a}(u, \vec(d), c) \coloneqq \phi_{\gamma,a_1}(u,\vec(d))\tau_{\gamma,a_2}(c) 
\end{equation}
where
\begin{align*}
    \phi_{\gamma,a_1}(u,\vec(d)) &\coloneqq \varphi\left(\frac{u'd^{-1}u - a_1}{\gamma\text{det}^5(d)}\right)\\ 
    \tau_{\gamma,a}(c) &\coloneqq \varphi \left(\frac{F_{1/\gamma}(c) - a_2}{\gamma}\right)
\end{align*}
The function \(\tilde\varphi_{\gamma,a}(\cdot,\cdot,\cdot)\) is meant to approximate the indicator function \(\bm{1}\{K(\beta_0) \leq a_1\}\bm{1}\{C \leq a_2\}\) with \(\gamma\) governing the quality of approximation. Where it is obvious, we will supress the subscripts \(\gamma,a\) from our notation.

\subsection{Main Argument}
\label{subsec:joint-main}
\begin{lemma}[Joint Lindeberg Interpolation]
    \label{lemma:joint-interpolation}
    Suppose that \Cref{assm:balanced-design,assm:local-identification} hold as well as the moment conditions of \Cref{thm:multiple-feasible-local-power}. Then there are fixed constants \(M_1,M_2\) such that
    \begin{equation}
        \label{eq:joint-interpolation}
        \left|\E[\tilde\varphi_{\gamma,a}(U,\vec(D),C) - \tilde\varphi_{\gamma,a}(\tilde U, \vec(\tilde D), \tilde C)]\right| \leq \frac{M_1\log^{M_2}(n)}{\sqrt{n}}(\gamma^{-1} + \gamma^{-2} + \gamma^{-3})
    \end{equation}
\end{lemma}
\begin{proof}[Proof of \Cref{lemma:joint-interpolation}]
    We can bound the difference on the left hand side of \eqref{eq:joint-interpolation} using the telescoping sum
   \begin{equation}
       \label{eq:joint-telescoping-sum}
       \begin{split}
           \sum_{i=1}^n\big| \E[\tilde\varphi_{\gamma,a}(U_{-i}
           &+ \Delta_{Ui}/\sqrt{n},\vec(D_{-i} + \Delta_{Di}/n), C_{-i} + \Delta_{Ci}/\sqrt{n})]\\
           &-\E[\tilde\varphi_{\gamma,a}(U_{-i} + \Delta_{Ui}/\sqrt{n},\vec(D_{-i} + \Delta_{Di}/n), C_{-i} + \Delta_{Ci}/\sqrt{n})]\big|
       \end{split}
   \end{equation}
   By second degree Taylor expansion, we break each of the summands in \eqref{eq:joint-telescoping-sum} into first order, second order, and remainder terms; each of which are bounded below. We make use of the following moment conditions implied by (i) indpendence of observations across \(i = 1,\dots,n\) and (ii) the mean and covariance matrix of \((\eps_i(\beta_0),r_i)\) being equal to the mean and covariance matrix of \((\tilde\eps_i(\beta_0),r_i)\)
   \begin{equation}
       \label{eq:joint-matched-moments}
       \begin{split}
           0 &=\E[\Delta_{Ui} - \tilde\Delta_{Ui}|\calF_{-i}] 
           = \E[\Delta_{Ui}\Delta_{Ui}' - \tilde\Delta_{Ui}\tilde\Delta_{Ui}'|\calF_{-i}] = \E[\vec(\Delta_{Di}) - \vec(\tilde\Delta_{Di})|\calF_{-i}] \\
             &= \E[\Delta_{Ci} - \tilde\Delta_{Ci}|\calF_{-i}] = \E[\Delta_{Ui}\otimes\vec(\Delta_{Di}^b)' - \tilde\Delta_{Ui}\otimes\vec(\tilde\Delta_{Di}^b)'|\calF_{-i}]\\ 
             &= \E[\Delta_{Ci}\otimes\Delta_{Ui} - \tilde\Delta_{Ci}\otimes\tilde\Delta_{Ui}|\calF_{-i}] = \E[\Delta_{Ci}\otimes\vec(\tilde\Delta_{Di}^b) - \tilde\Delta_{Ci}\otimes\vec(\tilde\Delta_{Di}^b)|\calF_{-i}]\\ 
             &= \E[\vec(\Delta_{Di}^b)\vec(\Delta_{Di}^b)' - \vec(\tilde\Delta_{Di}^b)\vec(\tilde\Delta_{Di}^b)'|\calF_{-i}]
       \end{split}
   \end{equation}
   where \(\calF_{-i}\) denotes the sub-sigma algebra generated by all observations not equal to \(i\), \(\otimes\) denotes the Kronecker product, and I apologize for the abuse of the equal sign in the above display.

   \textbf{First Order Terms.}
   First order terms can be expressed 
   \begin{align*}
       \text{First Order}_i 
       &= \sum_{\ell = 1}^{d_x} \E\left[\frac{\partial }{\partial U_\ell}\tilde\varphi(U_{-i},\vec(D_{-i}),C_{-i})((\Delta_{Ui})_\ell - (\tilde\Delta_{Ui})_\ell)\right]/\sqrt{n}  \\ 
       &+ \sum_{\ell=1}^{d_x}\sum_{m = 1}^{d_x} \E\left[\frac{\partial }{\partial D_{\ell m}}\tilde\varphi(U_{-i},\vec(D_{-i}),C_{-i})((\Delta_{Di})_{\ell m} - (\tilde\Delta_{Di})_{\ell m})\right]/n \\ 
       &+ \sum_{\ell = 1}^{2nd_x} \E\left[\frac{\partial }{\partial C_{\ell}}\tilde\varphi(U_{-i}, \vec(D_{-i}), C_{-i})((\Delta_{Ci})_\ell - (\tilde\Delta_{Ci})_\ell)\right] /\sqrt{n}
   \end{align*}
   These terms are all equal to zero after applying the matched moments in \eqref{eq:joint-matched-moments}.

   \textbf{Second Order Terms.} After canceling out terms using the matched moments in \eqref{eq:joint-matched-moments} the second order terms that remain can be expressed
   \small
   \begin{align*}
       \text{2nd Order}_i 
       &= \frac{1}{n^{3/2}}\sum_{\ell = 1}^{d_x} \sum_{m = 1}^{d_x} \sum_{n = 1}^{d_x} \underbrace{\E\bigg[\frac{\partial^2 }{\partial U_{\ell} \partial D_{mn}}\tilde\varphi(U_{-i},\vec(D_{-i}),C_{-i})((\Delta_{Ui})_\ell (\Delta_{Di}^a)_{mn} - (\tilde\Delta_{Ui})_\ell(\tilde\Delta_{Di}^a)_{mn}) \bigg]}_{\vA_{\ell mn}} \\ 
       &= \frac{1}{n^{2}}\sum_{\ell = 1}^{d_x} \sum_{m = 1}^{d_x} \sum_{n = 1}^{d_x}\sum_{o=1}^{d_x} \underbrace{\E\bigg[\frac{\partial^2 }{\partial U_{\ell} \partial D_{mn}}\tilde\varphi(U_{-i},\vec(D_{-i}),C_{-i})((\Delta_{Di}^a)_{\ell m} (\Delta_{Di}^a)_{no} - (\tilde\Delta_{Di}^a)_{\ell m}(\tilde\Delta_{Di}^a)_{no}) \bigg]}_{\vB_{\ell m n o}} \\ 
       &= \frac{2}{n^{2}}\sum_{\ell = 1}^{d_x} \sum_{m = 1}^{d_x} \sum_{n = 1}^{d_x}\sum_{o=1}^{d_x} \underbrace{\E\bigg[\frac{\partial^2 }{\partial U_{\ell} \partial D_{mn}}\tilde\varphi(U_{-i},\vec(D_{-i}),C_{-i})((\Delta_{Di}^b)_{\ell m} (\Delta_{Di}^a)_{no} - (\tilde\Delta_{Di}^a)_{\ell m}(\tilde\Delta_{Di}^b)_{no}) \bigg]}_{\vC_{\ell m n o}} \\ 
       &= \frac{1}{n^{3/2}}\sum_{\ell = 1}^{2nd_x} \sum_{m = 1}^{d_x} \sum_{n = 1}^{d_x} \underbrace{\E\bigg[\frac{\partial^2 }{\partial C_{\ell} \partial D_{mn}}\tilde\varphi(U_{-i},\vec(D_{-i}),C_{-i})((\Delta_{Ci})_{\ell} (\Delta_{Di}^a)_{mn} - (\tilde\Delta_{Ci})_{\ell}(\tilde\Delta_{Di}^a)_{mn}) \bigg]}_{\vD_{\ell m n}} 
   \end{align*}
   \normalsize
   To bound each \(\vA_{\ell m n}\), \(\vB_{\ell mno}\), and \(\vC_{\ell m n o}\) we use the fact that the second order derivatives of \(\tilde\varphi\) are bounded up to a log power of \(n\) via repeated application of \Cref{lemma:multiple-ND-bounds,lemma:total-derivatives}. Under the moment conditions of \Cref{thm:multiple-feasible-local-power} the absolute value of terms \((\Delta_{Ui})_{\ell}\),\(|\Delta_{Di}^a|_{mn}\), and \((\Delta_{Di}^b/\sqrt{n})_{no}\) can also be shown to have bounded third moments via the exact same steps as in the proof of \Cref{lemma:delta-moment-bounds}. Putting these together with generalized Holder's inequality will yield a finite constants \(M_1\) and \(M_2\) such that \(|\vA_{lmn}| \leq M_1\log^{M_2}(n)(\gamma^{-1} + \gamma^{-2})\), 
   \(\vB_{\ell mno} \leq M_1\log^{M_2}(n)(\gamma^{-1} + \gamma^{-2})\), and \(|\vC_{\ell mno}| \leq M_1\log^{M_2}(n)n^{1/2}(\gamma^{-1} + \gamma^{-2})\). To bound \(\vD_{\ell m n}\) terms notice that 
   \small
   \begin{align*}
       \sum_{\ell = 1}^{2nd_x}\vD_{\ell mn} 
       &=\sum_{\ell = 1}^{2nd_x} \E\bigg[\frac{\partial }{\partial D_{mn}}\phi(U_{-i},\vec(D_{-i}))\frac{\partial }{\partial C_\ell}\tau(C_{-i})((\Delta_{C-i})_\ell(\Delta_{Di}^a)_{mn} - (\tilde\Delta_{Ci})_\ell(\tilde\Delta_{Di}^a)_{mn}))\bigg] 
       \intertext{Apply \Cref{lemma:delta-moment-bounds} to bound \(\Delta_{Di}^a\), and \Cref{lemma:multiple-ND-bounds,lemma:total-derivatives} to bound the derivative of \(\phi(\cdot)\) and Cauchy-Schwarz to split up the \(\Delta_{Ci}\) and \(\Delta_{Di}\) terms}
       &\leq \sqrt{M_1\log^{M_2}(n)\gamma^{-2}}\E\bigg[\sum_{\ell =1}^{2nd_x}(\partial_\ell \tau(C_{-i}))^2((\Delta_{Ci})_\ell + (\tilde\Delta_{Ci})_\ell)^2\bigg]^{1/2}\\ 
       &\leq \sqrt{M_1\log^{M_2}(n)\gamma^{-2}}\E\bigg[\max_{1 \leq \ell \leq n} ((\Delta_{Ci})_{2\ell} + (\tilde\Delta_{Ci})_{2\ell})^2 \sum_{\ell = 1}^{2nd_x} (\partial_\ell\tau(C_{-i}))^2\bigg]^{1/2}
       \intertext{By \Cref{lemma:smooth-max-derivatives} and chain rule we have that \(\sum_{\ell = 1}^{2nd_x}(\partial_\ell \tau(C_{-i}))^2 \leq \gamma^{-2}\). Moreover \((\Delta_{Ci})^{a/2}_\ell\) is sub-exponential so via \Cref{lemma:sub-exponential-power-max-bound} the second moment of the maximum is bounded by a power of \(\log(n)\). After updating the constant \(M_1\) and \(M_2\) this yields}
       &\leq M_1\log^{M_2}(n)\gamma^{-2}
   \end{align*}
   \normalsize
   Putting these all together and summing over the remaining indices gives
   \begin{equation}
       \label{eq:joint-second-order-bound}
       |\text{Second Order}_i| \leq \frac{M_1\log^{M_2}(n)}{n^{3/2}}(\gamma^{-1} + \gamma^{-2})
   \end{equation}
   \textbf{Remainder Terms.} The first remainder term can be expressed
   \small
   \begin{align*}
       \text{Remainder}_i 
       &= \frac{1}{n^{3/2}}\sum_{\ell =1}^{d_x}\sum_{m = 1}^{d_x}\sum_{n=1}^{d_x} \E\bigg[\frac{\partial^3 }{\partial U_\ell \partial U_m \partial U_n}\tilde\varphi(\bar U, \vec(\bar D), \bar C)(\Delta_{Ui})_\ell(\Delta_{Ui})_m(\Delta_{Ui})_n\bigg] \\ 
       &+ \frac{1}{n^3}\sum_{(\ell,m)}\sum_{(n,o)}\sum_{(q,p)} \E\bigg[\frac{\partial^3 }{\partial D_{\ell m}\partial D_{no}\partial D_{pq}}\tilde\varphi(\bar U, \vec(\bar D), \bar C)(\Delta_{Di})_{\ell m}(\Delta_{Di})_{no}(\Delta_{Di})_{qp}\bigg] \\  
       &+ \frac{1}{n^{3/2}}\sum_{\ell =1}^{2nd_x}\sum_{m = 1}^{2nd_x}\sum_{n=1}^{2nd_x} \E\bigg[\frac{\partial^3 }{\partial C_\ell \partial C_m \partial C_n}\tilde\varphi(\bar U, \vec(\bar D), \bar C)(\Delta_{Ci})_\ell(\Delta_{Ci})_m(\Delta_{Ci})_n\bigg] \\  
       &+ \frac{1}{n^{2}}\sum_{\ell =1}^{d_x}\sum_{m = 1}^{d_x}\sum_{(n,o)} \E\bigg[\frac{\partial^3 }{\partial U_\ell \partial U_m \partial D_{no}}\tilde\varphi(\bar U, \vec(\bar D), \bar C)(\Delta_{Ui})_\ell(\Delta_{Ui})_m(\Delta_{Di})_{no}\bigg] \\ 
       &+ \frac{1}{n^{5/2}}\sum_{\ell =1}^{d_x}\sum_{(m,n)}\sum_{(o,p)} \E\bigg[\frac{\partial^3 }{\partial U_\ell \partial D_{mn} \partial D_{op}}\tilde\varphi(\bar U, \vec(\bar D), \bar C)(\Delta_{Ui})_\ell(\Delta_{Di})_{mn}(\Delta_{Di})_{op}\bigg] \\ 
       &+ \frac{1}{n^{5/2}}\sum_{\ell =1}^{2nd_x}\sum_{(m,n)}\sum_{(o,p)} \E\bigg[\frac{\partial^3 }{\partial C_\ell \partial D_{mn} \partial D_{op}}\tilde\varphi(\bar U, \vec(\bar D), \bar C)(\Delta_{Ci})_\ell(\Delta_{Di})_{mn}(\Delta_{Di})_{op}\bigg] \\ 
       &+ \frac{1}{n^{2}}\sum_{\ell =1}^{2nd_x}\sum_{m = 1}^{2nd_x}\sum_{(n,o)} \E\bigg[\frac{\partial^3 }{\partial C_\ell \partial C_m \partial D_{no}}\tilde\varphi(\bar U, \vec(\bar D), \bar C)(\Delta_{Ci})_\ell(\Delta_{Ci})_m(\Delta_{Di})_{no}\bigg] \\ 
       &+ \frac{1}{n^{3/2}}\sum_{\ell =1}^{2nd_x}\sum_{m = 1}^{2nd_x}\sum_{n=1}^{d_x} \E\bigg[\frac{\partial^3 }{\partial C_\ell \partial C_m \partial U_{n}}\tilde\varphi(\bar U, \vec(\bar D), \bar C)(\Delta_{Ci})_\ell(\Delta_{Ci})_m(\Delta_{Ui})_{n}\bigg] \\ 
       &+ \frac{1}{n^{2}}\sum_{\ell =1}^{2nd_x}\sum_{m = 1}^{2nd_x}\sum_{n=1}^{d_x} \E\bigg[\frac{\partial^3 }{\partial C_\ell \partial C_m \partial U_{n}}\tilde\varphi(\bar U, \vec(\bar D), \bar C)(\Delta_{Ci})_\ell(\Delta_{Ci})_m(\Delta_{Ui})_{n}\bigg] \\ 
   \end{align*}
   \normalsize
   where \(\bar U, \vec(\bar D)\), and \(\bar C\) vary term by term but are always in the hyper-rectangles \([U_{-i}, U + \Delta_{Ui}]\), \([\vec(D_{-i}), \vec(D_{-i} + \Delta_{Di})]\), and \([C_{-i}, C_{-i} + \Delta_{Ci}]\), respectively. As such, any moment conditions that apply to \(U, D, C\) also apply to \((\bar U,\bar D, \bar C)\).
   Repeated application of generalized Hölder inequality, \Cref{lemma:delta-moment-bounds} to bound moments of \(\Delta_{Ui}\) and \((\Delta_{Di}/\sqrt{n})\), \Cref{lemma:total-derivatives} to bound moments of the second and third derivatives of \(\phi(\tilde U, \vec(\tilde D))\), \Cref{lemma:composition-derivative-bounds} to bound the sums of derivatives of \(\tau(\tilde C)\), and \Cref{lemma:sub-exponential-power-max-bound} to bound moments of \(\max_{1 \leq \ell \leq n} (\Delta_{Ci})_\ell\) will yield that
   \begin{equation}
        \label{eq:joint-remainder-bound}
        |\text{Remainder}_i| \leq \frac{M_1\log^{M_2}(n)}{n^{3/2}}(\gamma^{-1} + \gamma^{-2} + \gamma^{-3})
   \end{equation}
   Symmetric logic will bound the other remainder term. Summing \eqref{eq:joint-second-order-bound} and \eqref{eq:joint-remainder-bound} over indices gives the result.
\end{proof}

\begin{lemma}[Denominator Anticoncentration]
    \label{lemma:denominator-anticoncentration}
    Suppose that \Cref{assm:balanced-design,assm:local-identification} hold as well as the moment conditions of \Cref{thm:multiple-feasible-local-power}. Then for any sequence \(\delta_n \to 0\) we have that \(\Pr(\lambda_{\min}(\tilde D) \leq \tilde\delta_{n}) \to 0\).
\end{lemma}
\begin{proof}
    By \Cref{lemma:anti-concentration-suffecient} it suffices to show that for any  fixed \(a\in\calS^{d_x-1}\) and any \(\delta_n \to 0\), \(\Pr(a'Da \leq \delta_{n})\to 0\). For any such \(a\) write:
    \begin{align*}
        a'\tilde Da 
        &= \frac{1}{n}\sum_{i=1}^n \E[\eps_i^2(\beta_0)]\big(\sum_{\ell=1}^{d_x}\sum_{j=1}^n a_\ell \tilde h_{\ell,ij} r_{\ell,j}\big)^2 \\ 
        &\geq \frac{1}{cn}\sum_{i=1}^n\big(\sum_{\ell=1}^{d_x}\sum_{j=1}^n a_\ell \tilde h_{\ell,ij} r_{\ell,j}\big)^2  
        \intertext{Define \(\dot s_{n,j} = \max_{\{\ell: a_\ell \neq 0\}}s_{n,\ell}\) and \(\dot h_{ij} =s_n h_{ij}\)} 
        &= \frac{1}{cn}\sum_{i=1}^n \big(\sum_{j=1}^n \dot h_{ij} \sum_{\ell=1}^{d_x} \frac{a_\ell s_{n,\ell}}{s_n} r_{\ell, j}\big)^2 \\ 
    \end{align*}
    By the moment conditions required by \Cref{thm:multiple-feasible-local-power} we have that \(\lambda_{\min}(\E[D]) \geq \underline c\) so that \(\E[\frac{1}{n}\sum_{i=1}^n \big(\sum_{\ell=1}^{d_x} \sum_{j=1}^n a_\ell \tilde h_{\ell, ij}r_{\ell,j}\big)^2] \geq c^{-1}\). Moreover, by assumption, \(\Var(\sum_{\ell = 1}^{d_x} \frac{a_\ell s_{n,\ell}}{s_n})\) is bounded from above and below. Define the matrix \(\tilde H = [\dot h_{ij}]_{ij}\) and follow the same steps as \Cref{lemma:denominator-anticoncentration} to conclude.
\end{proof}

\begin{lemma}[Gaussian Approximation]
    \label{lemma:gaussian-approximation}
    Suppose that \Cref{assm:balanced-design,assm:local-identification} hold as well as the moment conditions of \Cref{thm:multiple-feasible-local-power}. Then,
    \[
        \sup_{a \in \SR} \big|\Pr(\JK_I(\beta_0) \leq a) - \Pr(\JK_G(\beta_0) \leq a)\big|\to 0
    \]
\end{lemma}
\begin{proof}
    Let \(a = (a_1,a_2)\) and \(\tilde\phi_{\gamma,a}\) be as in \eqref{eq:joint-tilde-phi}:
    \begin{align*}
        \Pr(N'D^{-1}N \leq a_1, C \leq a_2) 
        &\leq \E[\tilde\phi_{\gamma,a}(U, \vec(D), C)] \\ 
        &\leq \E[\tilde\phi_{\gamma,a}(\tilde U, \vec(\tilde D), \tilde C)] + \frac{M_1\log^M_2(n)}{\sqrt{n}}(\gamma^{-1} + \gamma^{-2}) \\
        &\leq \Pr(\tilde N'\tilde D^{-1}\tilde N \leq a_1, \tilde C \leq a_2) + 
        \Pr(a_1 \leq \tilde N'\tilde D^{-1}N \leq a_1 + \gamma\lambda_{\min}^5(D)) \\ 
        &\;\;\; + \Pr(a_2 \leq C \leq a_2 + \gamma) + \frac{M_1\log_2^{M_2}(n)}{\sqrt{n}}(\gamma^{-1} + \gamma^{-2} + \gamma^{-3}) \\ 
        &\leq \Pr(\tilde N'\tilde D^{-1}\tilde N \leq a_1, \tilde C \leq a_2) + 
        \Pr(a_1 \leq \tilde N'\tilde D^{-1}N \leq a_1 + \gamma\lambda_{\min}^5(D)) \\ 
        &\;\;\; + \Pr(a_2 \leq C \leq a_2 + \gamma) + \frac{M_1\log_2^{M_2}(n)}{\sqrt{n}}(\gamma^{-1} + \gamma^{-2} + \gamma^{-3}) \\ 
    \end{align*}
    Let \(\gamma \to 0\) at a rate such that \(\frac{\log^{M_2}(n)}{\sqrt{n}}\gamma^{-3}\to 0\) and apply \Cref{lemma:joint-interpolation,lemma:denominator-anticoncentration} to conclude as in the proof of \Cref{lemma:approximate-distribution}. A symmetric argument shows that the lower bound tends to zero.
    
\end{proof}

\begin{lemma}[]
    \label{lemma:anti-concentration-suffecient}
    Let \(\Sigma_n \in \SR^{d\times d}\) be a sequence of random positive-semidefinite matrices. Suppose that for any fixed \(a \in \calS^{d-1}\) and any \(\delta_n \to 0\) we have that \(\Pr(a'\Sigma_na \leq \delta_n) \to 0\) and \(\Pr(\lambda_{\max}^2(\Sigma_n) \geq \delta_n^{-1}) \to 0\). Then for any \(\delta_n \to 0\), \(\Pr(\lambda_{\min}^2(\Sigma_n) \leq \delta_n)\to 0\).
\end{lemma}
\begin{proof}
    Take any preliminary sequence \(\delta_n \to 0\). It suffices to show that there is another sequence \(\tilde\delta_n\) weakly larger than \(\delta_n/2\) such that \(\Pr(\lambda_{\min}^2(\Sigma_n) \leq \tilde\delta_n) \to 0\). For any \(m \in \SN\) let \(\calA_m\) be a set of points in \(\calS^{d-1}\) such that
    \[
        \max_{a\in\calS^{d-1}} \min_{\tilde a \in \calA_m} \|a - \tilde a\| \leq \delta_m^2
    \]
    From here let \(\tilde n_j\) be defined
    \[
        \tilde n_j = \inf\{n \geq j: \min_{\tilde a \in \calA_{n,j}}\Pr(\tilde a'\Sigma_na \leq 2\delta_{n_j}) < \delta_{n_j}\}
    \]
    Define a new sequence \(\tilde\delta_n \to 0\), weakly larger than \(\delta_n\), via
    \[
        \tilde\delta_n =
        \begin{cases}
        1 &\text{if }0 \leq 0 \leq n < \tilde n_1  \\
        \delta_i &\text{if } \tilde n_i \leq n < \tilde n_{i+1}
        \end{cases}
    \]
    and notice that, by definition \(\Pr(\min_{a \in \calA_{\tilde n_j}} a'\Sigma_na \leq 2\tilde\delta_n) < \delta_{\tilde n_j}\).
    We wish to show that \(\lambda_{\min}^2(\Sigma_n)  > \tilde\delta_n\) on an intersection of events whose probability tends to one. Since \(\Sigma_n\) is positive semi-definite, \(\|x\|^2_{\Sigma_n} = x'\Sigma_nx\) defines a seminorm. By triangle inequality 
    \[
        \lambda_{\min}^2(\Sigma_{n_j}) \geq \min_{\calA_{n_j}} a'\Sigma_{n_j} a - \lambda_{\max}^2(\Sigma_{n})\tilde\delta_{n_j}^2
    \]
    Define the events
    \[
        \Omega_1 = \{\min_{\calA_{\tilde n_j}} a'\Sigma_{n} a \geq 2\tilde\delta_n\} \andbox \Omega_2 = \{\lambda_{\max}(\Sigma_n) \leq \tilde\delta_n^{-1/2}\}
    \]
    On the intersection of these events, whose probabilities tend to one, we have \(\lambda_{\min}^2(\Sigma_{n}) \geq \tilde\delta_{n}\).
\end{proof}

% Supplementary Appendices
\newblankpage
\section{Incorporating Exogenous Controls}
\label{sec:exogenous-controls}

In this section, I analyze the model with exogeneous controls. To this end, define the vector \(z_2 = (z_{21}',\dots,z_{2n}')' \in \SR^{n \times d_c}\). Let \(P_2 = z_2(z_2'z_2)^{-1}z_2' \in \SR^{n\times n}\) denote the projection onto the column space of \(z_2\) and \(M_2 = I_n - P_2\) denote the projection onto to orthocomplement of the column space. Focus will be on the case where \(d_x = 1\) to simplify notation, but the basic concepts apply generally to \(d_x > 1\).

For \(y \coloneqq (y_1,\dots,y_n)' \in \SR^n\) and \(x \coloneqq (x_1',\dots,x_n')'\in \SR^{n\times}\) define \(y^\perp \coloneqq M_2y\) and \(x^\perp \coloneqq M_2 x\) as the ``partialled out'' versions of \(y\) and \(x\), respectively. Let \(y^\perp_i\) be the \(i\)\textsuperscript{th} element of \(y^\perp\) and \(x^\perp_i\) be the \(i\)\textsuperscript{th} element of \(x^\perp\). From here we can define \(\eps(\beta_0) \coloneqq y - x\beta_0 \), \(\eps^\perp(\beta_0) = M_2\eps(\beta_0)\) and \(r^\perp \coloneqq M_2r\) where as in the main text \(r = (r_1,\dots,r_n)'\) is constructed \(r_i = x_i - \rho(z_i)\eps_i(\beta_0)\). The definition of \(\rho(z_i)\) does not change after partialling out \(z_2\) since all expectations are understood to be conditional on the instruments \(z\). Notice that \(\eps^\perp(\beta_0)\) is mean zero. Finally I assume that the controls have been partialled out of hat matrix so that the effective hat matrix is \(M_2H\) and the vector \(\widehat\Pi \in
\SR^n\) is defined \(\widehat\Pi = (M_2 H)(M_2 r)\). This does not make a difference for the numerator of the \(\JK(\beta_0)\) statistic but does affect the denominator slightly. When this is not done, inference may be conservative.

Using matrix notation in the numerator to make things clear, we can write the version of the \(\JK(\beta_0)\) statistic with the partialled out vectors, \(\eps^\perp(\beta_0)\) and \(r^\perp\), in terms of the original vectors, \(\eps(\beta_0)\) and \(r\),
\begin{align*}
    \JK_I(\beta_0) 
    &= \frac{\big(\frac{1}{\sqrt{n}}\eps(\beta_0)'M_2\tilde H M_2 r\big)^2}{\frac{1}{n}\sum_{i=1}^n (\eps_i^\perp(\beta_0))^2\big(\sum_{j=1}^n \vh_{ij} r_j\big)^2}  \\
    &= \frac{\big(\frac{1}{\sqrt{n}}\sum_{i=1}^n \eps_i(\beta_0) \sum_{j=1}^n \vh_{ij}r_j\big)^2}{\frac{1}{n}\sum_{i=1}^n(\eps_i^\perp(\beta_0))^2\big(\sum_{j=1}^n \vh_{ij} r_j\big)^2}
\end{align*}
where \(\vh_{ij} = [M_2 \tilde HM_2]_{ij}\), \(\tilde H = s_n H\), and \(m_{ij} = [M_2]_{ij}\). I seek to characterize the limiting distribution of \(\JK(\beta_0)\) under \(H_0\). To do so, we show that quantiles \(\JK(\beta_0)\) can be approximated by quantiles of the gaussian analog statistic
\begin{align*}
    \JK_G(\beta_0) = \frac{\big(\frac{1}{\sqrt{n}}\tilde\eps(\beta_0)'M_2 \tilde H M_2 \tilde r\big)^2}{\frac{1}{n}\sum_{i=1}^n \Var(\eps_i)\big(\sum_{j = 1}^n \vh_{ij}\tilde r_j\big)^2}
\end{align*}
where \((\tilde\eps_i, \tilde \eps_i(\beta_0), \tilde r_i)\) are generated gaussian independent of the data and with the same mean and covariance as \((\eps_i, \eps_i(\beta_0),r_i)\). Since \(\Var(\tilde\eps(\beta_0)) = \Var(\eps_i)\) under \(H_0\), \(\E[\tilde\eps(\beta_0)'M_2] = 0\), and \(\tilde r \perp \tilde\eps(\beta_0)\), this gaussian analog statistic has a \(\chi^2_1\) distribution conditional on any realization of \(\tilde r\) and thus its unconditional distribution is also \(\chi^2_1\).

Showing that quantiles of \(\JK(\beta_0)\) can be approximated by quantiles of \(\tilde\JK(\beta_0)\) proceeds in two steps. In the first step, we show that  \(\JK(\beta_0)\) converges in probability to an intermediate statistic.
\begin{align*}
    \JK^{\text{int}}(\beta_0) = \frac{\big(\frac{1}{\sqrt{n}}\sum_{i=1}^n \eps_i(\beta_0)\sum_{j=1}^n \vh_{ij} r_j\big)^2}{\frac{1}{n}\sum_{i=1}^n \eps_i^2 (\sum_{j\neq i} \vh_{ij}r_j)^2} 
\end{align*}
We will then show that quantiles of this intermediate statistic can be approximated by quantiles of \(\tilde\JK(\beta_0)\). In view of \Cref{lemma:estimation-error-high-level}, it suffices to show for the first step that \(\Delta_D \to_p 0\), where
\begin{align*}
    \Delta_D = \frac{1}{n}\sum_{i=1}^n  ((\eps_i^\perp(\beta_0))^2 - \eps_i^2) \widehat\Pi_i^2
\end{align*}
To do this, notice that under \(H_0\) we can write \(\eps_i^\perp(\beta_0) = \eps_i + z_{2i}'(\widehat\Gamma - \Gamma)\) where \(\widehat\Gamma = (z_2'z_2)^{-1}z_2\eps(\beta_0)\) is a \(\sqrt{n}\)-consistent estimate of \(\Gamma\). Exploiting this fact we get
\begin{align*}
    \Delta_D = (\widehat\Gamma - \Gamma)'\frac{1}{n}\sum_{i=1}^n (\widehat\Pi_i)^2z_{2i} z_{2i}'(\widehat\Gamma- \Gamma) + 2(\widehat\Gamma- \Gamma)'\frac{1}{n}\sum_{i=1}^n \eps_iz_{2i}\widehat\Pi_i
\end{align*}
Both of these terms will tend to zero by the consistency \(\widehat\Gamma\) to \(\Gamma\), giving that \(\Delta_D \to_p 0\).

In our second step, we argue that quantiles of \(\JK^{\text{int}}(\beta_0)\) can be approximated by quantiles of \(\JK_G(\beta_0)\). To make this comparasion, we can follow almost exactly the same steps as in \Cref{sec:main-proofs}. The only difference between analysis in this case and analysis in the original case is that the partialling out of controls leads the test statistic to not strictly have a jackknife form; the effective hat matrix \(M_2HM_2\) no longer has a deleted diagonal. However, as I will argue below, this will not make a difference in the interpolation argument since the diagonal terms of \([P_2]_{ii}\) are small in the sense that they sum to \(d_c\).

The \eqref{eq:one-step-deviations} analog one step deviations for the numerator are given
\begin{align*}
    \Delta_{1i} &= \eps_i(\beta_0)\sum_{j\neq i} \vh_{ij}\dot r_j + r_i \sum_{j\neq i} \vh_{ji} \dot\eps_j(\beta_0) + \vh_{ii} \eps_i(\beta_0) r_i \\ 
    \tilde\Delta_{1i} &= \tilde\eps_i(\beta_0)\sum_{j\neq i} \vh_{ij}\dot r_j + \tilde r_j \sum_{j\neq i} \vh_{ji} \dot\eps_j(\beta_0) + \vh_{ii}\tilde\eps_i(\beta_0)\tilde r_i
\end{align*}
where as \Cref{sec:main-proofs}, a dotted variable is equal to the gaussian analog if \(j > i\) but equal to the standard version otherwise.
The first and second moments of the first two terms in \(\Delta_{1i}\) can be matched with their gaussian analog terms as in the proof of \Cref{lemma:lindeberg-interpolation}. While we cannot match seconds moments of the third term in the one step deviation, this sum of all these third terms can be treated as negligible after scaling by \(1/\sqrt{n}\) as \(\sum_{i=1}^n |\vh_{ii}| \lesssim d_c\). This is because \(M_2\tilde H M_2 = \tilde H - P_2 \tilde H - \tilde H P_2 - P_2 \tilde H P_2\). The matrix \(\tilde H\) has zeros on it's diagonal. Meanwhile
\begin{align*}
    |[P_2 \tilde H]_{ii}|^2 = \Big|\sum_{j=1}^n [P_2]_{ij}\tilde H_{ji}\Big|^2 \leq \Big(\sum_{j=1}^n [P_2]_{ij}^2  \Big)\Big(\sum_{j\neq i} H_{ji}^2\Big) \lesssim [P_2]_{ii}
\end{align*}
where the final inequality comes because the matrix \(P_2\) is symmetric and idempotent and since \(\Big(\sum_{j\neq i} H_{ji}^2\Big) \lesssim 1\) by \Cref{assm:balanced-design}(ii). A similar argument can be used to show that \([P_2 \tilde H P_2]_{ii}^2 \lesssim [P_2]_{ii} \). Since \(P_2\) is a projection matrix we must have that  \(\|P_2He_j\| \leq \|He_j\|\) for any basis vector \(e_j \in \SR^n\). Thus \(\sum_{j=1}^n [P_2H]_{ji}^2 \leq \sum_{j=1}^n [H]_{ji}^2\). Finally, we can use the fact that the trace of \(P_2\) is equal to its rank to show that \(\sum_{i=1}^n|\vh_{ii}| \lesssim d_c\)

The one step deviations in the denominator can be bounded using the same logic. These one step deviations are given
\begin{align*}
    \Delta_{2i} 
    &= \eps_i^2 (\sum_{j\neq i} \vh_{ij}\dot r_j)^2 + r_i^2 \sum_{j\neq i} \vh_{ji}^2 \ddot \eps_j^2 + r_i \sum_{j\neq i} \ddot \eps_j \big(\sum_{k \neq j,i} \vh_{ji}\vh_{jk} r_k\big)\\ 
    &\;\;\;+ \eps_i^2\big(\vh_{ii}^2 r_i^2 + 2\vh_{ii} r_j \sum_{j\neq i} \vh_{ij}r_j)^2 \\
    \tilde\Delta_{2i} &= \tilde\eps_i^2 (\sum_{j\neq i} \vh_{ij}\dot r_j)^2 + \tilde r_i^2 \sum_{j\neq i} \vh_{ji}^2 \ddot \eps_j^2 + \tilde r_i \sum_{j\neq i} \ddot \eps_j \big(\sum_{k \neq j,i} \vh_{ji}\vh_{jk} r_k\big)\\ 
    &\;\;\;+ \eps_i^2\big(\vh_{ii}^2 r_i^2 + 2\vh_{ii} r_j \sum_{j\neq i} \vh_{ij}r_j)^2 
\end{align*}
where \(\ddot\eps_j\) is equal to \(\Var(\eps_j)\) if \(j < i\) and equal to \(\eps_j\) if \(j > i\). The first three terms in this expansion are can be dealt with exactly as in the proof of \Cref{lemma:lindeberg-interpolation}. The fourth term is new, however summing over the fourth terms and scaling by \(1/n\) will be negligible as \(\sum_{i=1}^n |\vh_{ii}| \lesssim d_c\).  After showing the lindeberg interpolation step, the rest of the proof follows exactly as in \Cref{sec:main-proofs}.

\section{Alternative Construction of Test Statistic via Cross Fitting}
\label{sec:cross-fit}

To accomodate a general class of estimators for \(\hat\rho(z_i)\) I propose a cross-fit form of the jackknife K-statistic. In this section I detail the cross-fitting procedure and present high level conditions needed for estimation error to be treated as negligible. These conditions can be satisfied by a large class of machine learning estimators under alternate conditions. 

\subsection{Cross Fit Test Statistic}

To construct the cross fit test statistic, evenly (and randomly) split the sample into two subsets, \(\calI_1\) and \(\calI_2\) such that \(\calI_1 \cap \calI_2 = \emptyset\) and \(\calI_1 \cup \calI_2 = [n]\). For each \(k = 1,2\) construct an estimator \(\hat\rho^{(k)}(\cdot)\) using only the observations in \(\calI_k\). For observations in \(i \in \calI_1\) form the first-stage estimates 
\[
    \widehat\Pi_i = \sum_{j \in \calI_1\setminus\{i\}} h_{ij}^{(1)}(x_j - \eps_j(\beta_0)\hat\rho^{(2)}(z_j))
\]
where the weights \(h_{ij}^{(1)}\) come from a hat matrix \(H^{(1)}\) that only depends on the observations in \(\calI_1\), for example the ridge regression hat matrix of \Cref{sec:setup} using only the instruments \((z_i)_{i\in\calI_1}\). First stage estimates for observations in \(\calI_2\) symmetrically, using the estimator \(\hat\rho^{(1)}(\cdot)\). The test statistic is then constructed as before
\[
    \JK(\beta_0) = \frac{\big(\sum_{i=1}^n \eps_i(\beta_0)\widehat\Pi_i\big)^2}{\sum_{i=1}^n \eps_i^2(\beta_0)(\widehat\Pi_i)^2}
\]
Notice that while only half the sample is being used to construct each first stage estimate, the full sample is still being used to test the null hypothesis.

\subsection{Controlling Estimation Error}
After writing the overall hat matrix as 
\[
    H = \begin{pmatrix} H^{(1)} & \bm{0} \\ \bm{0} & H^{(2)} \end{pmatrix},
\]
analysis of the infeasible statistic proceeds as before. Thus, to show that the crossfit test statistic has a limiting \(\chi^2_1\) distribution by \Cref{lemma:estimation-error-high-level} it suffices to state high-level conditions under which \((\Delta_N, \Delta_D)' \to_p 0\).

\begin{assumption}[Cross-fit Conditions]
    \label{assm:cross-fit-conditions}
    Suppose (i) that there is a constant \(\nu \in (0,1] \cup \{2\}\) such that for all \(i \in[n]\), \(\|\eps_i\|_{\Psi_\nu} \leq c\) and that (ii) for \(k = 1,2\)
    \[
        \max_{i \in I_k} (\hat\rho^{(-k))}(z_j) - \rho(z_j))^2 = o_p(\log^{1/\nu}(n)).
    \]
    where \(\hat\rho^{(-k)}(\cdot)\) indicates the estimator of \(\rho(\cdot)\) computed using observations in \([n]\setminus \calI_k\).
\end{assumption}

\begin{lemma}[Negligible Cross-fit error]
    \label{lemma:cross-fit-convergence}
    Suppose that Assumptions~\ref{assm:balanced-design}, and \ref{assm:cross-fit-conditions} hold as well as the moment conditions of \Cref{thm:feasible-local-power}. Then, under \(H_0\), \((\Delta_N, \Delta_D)'\to_p 0\).
\end{lemma}
\begin{proof}
    We consider the statements \(\Delta_N\to_p 0\) and \(\Delta_D \to_p 0\) seperately.

    \paragraph*{\(\Delta_N\to_p 0\):} It suffices to show that \(\Delta_{N,1} \to_p 0\) for
    \[
        \Delta_{N,1} = \frac{1}{\sqrt{n}}\sum_{i\in \calI_1} \eps_i(\beta_0)\sum_{j \in \calI_1\setminus\{i\}} \eps_j(\beta_0)\tilde h_{ij} (\hat\rho^{(2)}(z_j) - \rho(z_j)).
    \]
    The corresponding statement for \(\Delta_{N,2}\) follows from the same logic and \(\Delta_N = \Delta_{N,1} + \Delta_{N,2}\). Define the event 
    \[
        \Omega(\eps) \coloneqq \{ \max_{i\in\calI_k} (\hat\rho^{(2)}(z_j) - \rho(z_j))^2 \leq \eps\}
    \]
    and consider a sequence \(\eps_n \to 0\) such that \(\Pr(\Omega(\eps_n)) \geq 1 - \eps_n\). Noting that \(\Omega(\eps_n) \perp (\eps_i(\beta_0))_{i \in \calI_k}\) we write 
    \[
        \Delta_{N,1} = \eps(\beta_0)\mH\eps(\beta_0)
    \]
    where \(\mH \in \SR^{|\calI_1|\times |\calI_1|} = \frac{1}{\sqrt{n}}\big(\tilde h_{ij}(\hat\rho^{(1)}(z_j) - \rho(z_j))\big)_{i,j \in \calI_1}\). Since \(\E[\Delta_{N,1}|\Omega(\eps_n)] = 0\), an application of the generalized Hanson-Wright inequality, \Cref{thm:hanson-wright}, gives us that there is a sequence \(\delta_n \to 0\) such that
    \[
        \Pr(\Delta_{N,1} \geq \eps_n|\Omega(\eps_n)) \leq \delta_n 
    \]
    which allows us to conclude.

    \paragraph*{\(\Delta_D \to_p 0\):} As before, it suffices to show that \(\Delta_{D,1} \to_p 0\) where 
    \[
        \Delta_{N,1} = \frac{1}{n}\sum_{i \in \calI_1} \eps_i^2(\beta_0)\big\{\big(\sum_{j \in \calI_1\setminus\{i\}}\tilde h_{ij}\hat r_j\big)^2 - \big(\sum_{j \in \calI_1\setminus\{i\}}\tilde h_{ij} r_j\big)^2\}
    \]
    From the proof of \Cref{thm:estimation-error}, it suffices to show that 
    \[
        \max_{i \in \calI_1} \big|\sum_{j \in \calI_1\setminus\{i\}} \tilde h_{ij}\eps_j(\beta_0)(\hat\rho^{(2)}(z_j) -\rho(z_j) \big|\to_p 0
    \]
    Consider a sequence \(\eps_n \to 0\) such that \(\log^{1/\nu}(n)\eps_n \to 0\) and \(\Pr(\Omega(\eps_n)) \geq 1 - \eps_n\) and apply \Cref{thm:hanson-wright} to conclude that there is a sequence \(\delta_n \to 0\) such that
    \[
        \Pr\big(\max_{i \in \calI_1} \big|\sum_{j \in \calI_1\setminus\{i\}} \tilde h_{ij}\eps_j(\beta_0)(\hat\rho^{(2)}(z_j) -\rho(z_j) \big| \geq \eps_n\mid \Omega(\eps_n)\big) \leq \delta_n
    \]
    Since \(\Pr(\Omega(\eps_n))\to1\) this gives the result.
\end{proof}

\section{Relevant Moment Bounds}
\label{sec:relevant-bounds}

\subsection{Moment Bounds for \Cref{sec:single-endogeneous}}
\label{subsec:single-endogenous-moment-bounds}
Here I provide some lemmas that are useful in the proof of \Crefrange{lemma:lindeberg-interpolation}{lemma:approximate-distribution}

\begin{lemma}[]
    \label{lemma:delta-moment-bounds}
    Let \(\Delta_{1i}, \tilde\Delta_{1i}, \Delta_{2i}^a, \tilde\Delta_{2i}^a, \Delta_{2i}^b\),\(\tilde\Delta_{2i}^b\) be as in \eqref{eq:one-step-deviations}. Then under \Cref{assm:balanced-design} and the moment conditions of \Cref{thm:feasible-local-power} there is a constant \(M > 0\) such that for any \(k=1,\dots,6\):
    \begin{align*}
        \E[|\Delta_{1i}|^k] 
        &\leq M & \E[|\tilde\Delta_{1i}|^k] &\leq M 
    \intertext{and for any \(k=1,\dots,3\):}
        \E[|\Delta_{2i}^a|^k] 
        &\leq M\alpha^k & \E[|\tilde\Delta_{2i}^k|] &\leq M\alpha^k   \\
        \E[|\Delta_{2i}^b/\sqrt{n}|^k] &\leq M\alpha^k & \E[|\tilde\Delta_{2i}^b/\sqrt{n}|^k] &\leq M\alpha^k
    \end{align*}
\end{lemma}
\begin{proof}
    First, since
    \[
        \sum_{j=1}^n h_{ij}^2\E[(r_j - \E[r_j])^2 \leq\E[(\sum_{i=1}^n \tilde h_{ij}r_j)^2] \leq 1
    \]
    the constants are bounded, \(\sum_{i=1}^n \tilde h_{ij}^2 \leq c\). Applying \Cref{lemma:higher-moments-sum} with \(X_i = h_{ij}r_j\) and \(X_i = h_{ij}\eps_j(\beta_0)\) we see that there is a constant \(A\) such that for any \(k = 1,\dots,6\)
    \begin{equation}
        \label{eq:higher-moments-bounded}
        \E\big[\big|\sum_{i=1}^n \tilde h_{ij}r_j\big|^k\big] \leq A\andbox \E\big[\big|\sum_{i=1}^n \tilde h_{ij}\eps_j(\beta_0)\big|^k\big] \leq A
    \end{equation}
    The bounds on \(\E[|\Delta_{1i}^k|]\) and \(\E[|\tilde\Delta_{1i}^k|]\) immediately follow from this result and the bounds on moments of \(r_i\) and \(\eps_i(\beta_0)\).
    The bounds on \(\E[|\Delta_{2i}^a|^k]\) and \(\E[|\tilde\Delta_{2i}^a|^k]\) also follow from \eqref{eq:higher-moments-bounded} after noting that there is a finite constant \(B\) such that: 
    \[
        \E[(\sum_{i=1}^n \tilde h_{ij}^2\eps_i^2(\beta_0))^k] \leq B
    \]
    Finally to bound \(\E[|\Delta_{2i}^b/\sqrt{n}|^k]\) and \(\E[|\tilde\Delta_{2i}^b/\sqrt{n}|^k]\) apply \Cref{lemma:rtn-bound} with \(v_j = \eps_j^2(\beta_0)\sum_{k\neq i,j}\tilde h_{jk}r_k\), noting that \(\E[|v_j|^3]\) is bounded by \eqref{eq:higher-moments-bounded}.
\end{proof}
\begin{lemma}[]
    \label{lemma:numerator-bound}
    Let \(N\) and \(N_{-i}\) be defined as in \Cref{subsec:local-power-proof}. Under \Cref{assm:balanced-design,assm:local-identification} and the moment conditions in \Cref{thm:feasible-local-power}, there is a fixed constant \(M\) such that for all \(i=1,\dots,n\) and any \(k = 1,\dots,6\),
    \[
        \E[|N|^k] + \E[|N_{-i}|^k] \leq M
    \]
\end{lemma}
\begin{proof}
    We show the bound for \(\E[|N|^k]\) and note that the bound for \(N_{-i}\) follows from symmetric logic. Write \(\eps_i(\beta_0) = \eta_i + \gamma_i\) where \(\gamma_i = \Pi_i(\beta - \beta_0)\) and \(\eta_i\) is mean zero. Decompose \(N = N_1 + N_2 + N_3\):
    \[
        N_1 = \frac{1}{\sqrt{n}}\sum_{i=1}^n \eta_i \sum_{j=1}^n \tilde h_{ij} \dot r_j,\; N_2 = \frac{1}{\sqrt{n}}\sum_{i=1}^n r_i \sum_{j=1}^n \tilde h_{ji}\gamma_j, \andbox N_3 = \frac{1}{\sqrt{n}}\sum_{i=1}^n \eta_i \sum_{j=1}^n \tilde h_{ij}\E[r_j]
    \]
    where \(\dot r_j = r_j - \E[r_j]\).  

   Since via \Cref{assm:balanced-design}, \(\sum_{i=1}^n h_{ji}^2 \leq c\), and \(|\gamma_j| \leq c\), we can bound,
    \[
        (\sum_{j=1}^n h_{ji}\gamma_j /\sqrt{n})^4 \leq (\frac{c}{\sqrt{n}}\sum_{i=1}^n |h_{ji}|)^4 \leq c^8 \implies (\sum_{j=1}^n h_{ji}\gamma_j /\sqrt{n})^6 \leq c^8 (\sum_{j=1}^n h_{ji}\gamma_j /\sqrt{n})^2
    \]
    Under \Cref{assm:local-identification}, \(\E[N_2^2] \leq c\) while \Cref{assm:balanced-design} implies that \((\sum_{i=1}^n h_{ij}\E[r_j])^2 \leq c\) so that \(\E[N_3^2] \leq c^2\).

    An absolute bound on the higher moments of \(N_2\) then follows from an application of  \Cref{lemma:higher-moments-sum} with \(X_i = r_i \sum_{j=1}^n h_{ji}\gamma_j/\sqrt{n}\). An absolute bound on the higher moments of \(N_3\) follows from symmetric logic.

    To bound higher moments of \(N_1\) define \(v_i = \sum_{j < i} \{\eta_i h_{ij}r_j + \dot r_i h_{ji}\eta_j\}\) and write \(N_1 = \frac{1}{\sqrt{n}}\sum_{i=2}^n v_i\). The sequence \(v_2,\dots,v_n\) is a martingale difference array. Via the same procedure as the bounds on \(\E[|\Delta_{1i}|^k]\) as in \Cref{lemma:delta-moment-bounds} one can verify that there is a fixed constant \(M\) such that \(\E[|v_i|^k] \leq M\) for all \(k= 1,\dots,6\). The bound on the higher moments of \(N\) then follows from \Cref{lemma:martingale-bound}.

    The bounds for moments of \(N_{-i}\) follow symmetric logic.
\end{proof}

\begin{lemma}[]
    \label{lemma:approximation-error-bound}
    Let \(\tilde N\) and \(\tilde D\) be defined as in \Cref{subsec:local-power-proof}. Let \(f(\cdot, \tilde r)\) be the density function of \(\frac{\tilde N}{\tilde D^{1/2}}|\tilde r\). Under \Cref{assm:local-identification} and the moment bounds of \Cref{thm:feasible-local-power}, there is a constant \(M > 0\) such that \(\sup_x |f(x, \tilde r)| \leq M\) for almost all \(\tilde r\).
\end{lemma}
\begin{proof}
    Recall that
    \[
        \tilde N = \frac{1}{\sqrt{n}}\sum_{i=1}^n \tilde\eps_i(\beta_0)\sum_{j=1}^n \tilde h_{ij} \tilde r_j \andbox \tilde D^{1/2} = \sqrt{\frac{1}{n}\sum_{i=1}^n  \kappa_i^2(\beta_0) (\sum_{j=1}^n \tilde h_{ij}\tilde r_j)^2}
    \]
    The distribution of \(\tilde\eps_i(\beta_0)|\tilde r_i\) is
    \[
        \tilde\eps_i(\beta_0)|\tilde r \sim N\left(\mu_i(r_i), (1 - \rho_i^2)\Var(\eps_i(\beta_0)) \right)
    \]
    where \(\mu_i(r_i) = \Pi_i(\beta - \beta_0) + \frac{\Cov(\eps_i(\beta_0),r_i)}{\Var(r_i)}(r_i - \E[r_i])\) and \(\rho_i = \corr(\eps_i(\beta_0), r_i)\). Define \(\bar\Pi_i := \sum_{j=1}^n \tilde h_{ij} \tilde r_j\). Then, conditional on \(\tilde r\),
    \begin{equation}
        \label{eq:conditional-normal-density}
        \frac{\tilde N}{\tilde D^{1/2}} \sim  N\bigg(\frac{\frac{1}{\sqrt{n}}\sum_{i=1}^n \mu_i(r_i)\bar\Pi_i}{\sqrt{\frac{1}{n}\sum_{i=1}^n\kappa_i^2(\beta_0)\bar\Pi_i^2}}, \frac{\frac{1}{n}\sum_{i=1}^n (1 - \rho_i^2)\Var(\eps_i(\beta_0))\bar\Pi_i^2}{\frac{1}{n}\sum_{i=1}^n \kappa_i^2(\beta_0)\bar\Pi_i^2}\bigg)
    \end{equation}
    The maximum of the normal density is proportional to the inverse of the standard deviation so it suffices to show that the variance in \eqref{eq:conditional-normal-density} is bounded away from zero. To this end, notice that under the moment bounds of \Cref{thm:feasible-local-power} and \Cref{assm:local-identification}
    \[
        (1-\delta^2)c^{-2}\leq (1-\rho_i^2)\frac{\Var(\eps_i(\beta_0))}{\kappa_i^2(\beta_0)} \leq c^2
    \]
    By \Cref{lemma:sum-bound} to this gives that the conditional variance is also larger than \((1 -\delta^2)c^{-2} > 0\).

\end{proof}

\begin{lemma}[]
    \label{lemma:higher-moments-sum}
    Let \(X_1,\dots,X_n\) be random variables such that \(\E[X_i] = \mu_i\) and \(\E[(\sum_{i=1}^n X_i)^2] \leq C\). Suppose that for any \(i = 1,\dots,n\) there is a constant \(U\) such that 
    \[
        \E[(X_i - \mu_i)^3] \leq U\E[(X_i - \mu_i)^2] \andbox \E[(X_i - \mu_i)^6]^{1/3}\leq U \E[(X_i -\mu_i)^2]
    \]
    Then \(\E[(\sum_{i=1}^n X_i)^6] \leq  64U^3C^3 + 32C^3\).
\end{lemma}
\begin{proof}
    First write
    \begin{align*}
        \E[(\sum_{i=1}^n X_i)^2] 
        = \sum_{i=1}^n \E(X_i - \mu_i)^2 + (\sum_{i=1}^n \mu_i)^2 \leq C
    \end{align*}
    To bound \(\E[(\sum_{i=1}^n X_i)^6]\) expand out 
    \begin{align*}
        \E[(\sum_{i=1}^n X_i)^6]
        &= \E[(\sum_{i=1}^n (X_i - \mu_i) + \sum_{i=1}^n \mu_i)^6] \\ 
        &\lesssim\E[(\sum_{i=1}^n (X_i - \mu_i))^6] + (\sum_{i=1}^n \mu_i)^6 \\ 
        &= \sum_{i=1}^n \E[(X_i - \mu_i)^6] + \sum_{i=1}^n \sum_{j=1}^n \E[(X_i - \mu_i)^3(X_j - \mu_j)^3] \\  
        &\;\;\;\;\;\;+  \sum_{i=1}^n \sum_{j=1}^n \E[(X_i - \mu_i)^4(X_j - \mu_j)^2] \\ 
        &+\sum_{i=1}^n \sum_{j=1}^n \sum_{k\neq i,j} \E[(X_i - \mu_i)^2(X_j-\mu_i)^2(X_k - \mu_k)^2] + (\sum_{i=1}^n \mu_i)^6\\ 
        &\leq \sum_{i=1}^n \E[(X_i - \mu_i)^6] + \sum_{i=1}^n \sum_{j=1}^n \E[(X_i - \mu_i)^3]\E[(X_j - \mu_j)^3] \\ 
        &+  \sum_{i=1}^n \sum_{j=1}^n \E[(X_i - \mu_i)^6]^{4/6}\E[(X_j - \mu_j)^6]^{2/6} \\
        &+ \sum_{i=1}^n \sum_{j=1}^n \sum_{k\neq i,j} \E[(X_i - \mu_i)^6]^{1/3}\E[(X_j-\mu_i)^6]^{1/3}\E[(X_k - \mu_k)^6]^{1/3} \\ 
        &+ C^3  \\ 
        &= \bigg(\sum_{i=1}^n (\E[(X_i - \mu_i)^6])^{1/3}\bigg)^3  + \sum_{i=1}^n \sum_{j=1}^n \E[(X_i - \mu_i)^3]\E[(X_j - \mu_j)^3]+ C^3\\ 
        &\leq \bigg(\sum_{i=1}^n (\E[(X_i - \mu_i)^6])^{1/3}\bigg)^3  + \bigg(\sum_{i=1}^n \E[(X_i - \mu_i)^3]\bigg)^2 +  C^3\\ 
        &\leq 2U^3\bigg(\sum_{i=1}^n \E[(X_i - \mu_i)^2]\bigg)^3 + C^3\\ 
        &\leq 2U^3C^3 + C^3 \\ 
    \end{align*}
    where the implied constant in the second line is 32 by an application of \Cref{lemma:sum-bound}, the third line comes from expanding out the power, the first inequality by application of Hölder's inequality, and the penultimate inequality comes from applying bounds on the third and sixth central moments in terms of the second moments.
\end{proof}

\begin{lemma}[]
    \label{lemma:square-summable-bound}
    Let \(h = (h_1,\dots,h_n)\in \SR^n\) be such that \(\sum_{i=1}^n h_i^2 \leq b\). Suppose that \(X_1,\dots,X_n\) are such that \(\E[|X_i|^k] \leq M\) for all \(k = 1,2,3\). Then 
    \[
        \E\big[\big|\sum_{i=1}^n h_i^2 X_i\big|^3\big] \leq b^3M^3
    \]
\end{lemma}
\begin{proof}
    We can expand out 
    \begin{align*}
        \E\big[\big|\sum_{i=1}^n h_i^2X_i\big|^3] 
        &\leq \sum_{i=1}^n \sum_{j=1}^n\sum_{k=1}^n h_i^2h_j^2h_k^2 \E[|X_i||X_j||X_k|] \\ 
        &\leq  M^3 \sum_{i=1}^n h_i^2 \sum_{j=1}^n h_j^2 \sum_{k=1}^n h_k^2\\ 
        &\leq M^3\big(\sum_{i=1}^n h_i^2)^3 \leq c^3M^3
    \end{align*}
\end{proof}

\begin{lemma}[]
    \label{lemma:rtn-bound}
    Let \(v_1,\dots,v_n\) be random variables such that \(\E[|v_i|^3] \leq M\) for all \(i=1,\dots,n\). Let \(h = (h_1,\dots,h_n)\in \SR^n\) be a vector of weights such that \(\|h\|_2 \leq c\). Then
    \[
        \E\big[\big|\frac{1}{\sqrt{n}}\sum_{i=1}^n h_iv_i\big|^3\big] \leq c^3M
    \]
\end{lemma}
\begin{proof}
    We can expand out
    \begin{align*}
        \E\big[\big|\frac{1}{\sqrt{n}}\sum_{i=1}^n h_iv_i\big|^3\big] 
        &\leq \frac{1}{n^{3/2}}\sum_{i=1}^n \sum_{j=1}^n\sum_{k=1}^n |h_i||h_j||h_k|\E[|v_i||v_j||v_k|] \\ 
        &\leq \frac{M}{n^{3/2}}\sum_{i=1}^n |h_i| \sum_{j=1}^n |h_j| \sum_{k=1}^n |h_k| \leq \frac{M}{n^{3/2}}\|h\|_1^3 \leq Mc^3
    \end{align*}
    where the second inequality follows from generalized Hölder's inequality,
    \[
        |\E[fgh]| \leq (\E[|f|^3]\E[|g|^3]\E[|h|^3])^{1/3}
    \]
    and the fourth inequality from \(\|h\|_1 \leq \sqrt{n}\|h\|_2\).
\end{proof}

\begin{lemma}[]
    \label{lemma:martingale-bound}
    Let \(v_1,\dots,v_n\) be a martingale difference array such that \(\E[|v_i|^l] \leq M\) for all \(l=1,\dots,k\). Then there is a fixed constant \(C_k\) that only depends on \(k\) such that
    \[
        \E[(\frac{1}{\sqrt{n}}\sum_{i=1}^n v_i)^k] \leq C_k M
    \]
\end{lemma}
\begin{proof}
    We move to apply \Cref{thm:bdg} with \(X_t = \sum_{i=1}^t v_i/\sqrt{n}\). 
    \begin{align*}
        \E[(\frac{1}{\sqrt{n}}\sum_{i=1}^n v_i)^k] 
        &\leq \E[(\max_{s \leq n} \sum_{t=1}^s X_s)^k] \\ 
        &\leq C_k \E\big[\big(\sum_{i=1}^n v_i^2/n\big)^{k/2}\big]\leq C_k\E\big[\frac{1}{n}\sum_{i=1}^n v_i^k\big] \leq C_kM
    \end{align*}
    where the second inequality comes from \Cref{thm:bdg} and the third comes from an application of Jensen's inequality to the sample mean.
\end{proof}

\subsection{Useful Properties of Smooth Max}
\label{subsec:smooth-max-properties}
\begin{lemma}[\cite{cck2013}, Lemma A.2]
    \label{lemma:smooth-max-derivatives}
    For every \(1 \leq j,k,l \leq p\),
    \begin{align*}
        \partial_j F_\beta(z) &= \pi_j(z), & 
        \partial_j\partial_k F_\beta(z) &= \beta w_{jk}(z), &
            \partial_j\partial_k\partial_l F_\beta(z) &= \beta^2 q_{jkl}(z)
    \end{align*}
    where for \(\delta_{jk} := \bm{1}\{j=k\}\),
    \begin{align*}
        \pi_j(z) &:= e^{\beta z_j}\bigg/\sum_{i=1}^n e^{\beta z_i},\hbox{ }\;\;\;\; w_{jk} := (\pi_j\delta_{jk} - \pi_j\pi_k)(z) \\ 
        q_{jkl}(z) &:=  (\pi_j \delta_{jl}\delta_{jk} - \pi_j\pi_l \delta_{jk} - \pi_j\pi_k(\delta_{jl} + \delta_{kl}) + 2\pi_j\pi_k\pi_l)(z)
    \end{align*}
    Moreover,
    \begin{align*}
        \pi_j(z) &\geq 0, &
        \sum_{j=1}^p \pi_i(z) &= 1, & 
        \sum_{j,k=1}^p |w_{jk}(z)| &\leq 2, &
        \sum_{j,k,l=1}^p |q_{jkl}| &\leq 6
    \end{align*}
\end{lemma}
\begin{lemma}[\cite{cck2013}, {Lemma A.3}]
    \label{lemma:smooth-max-lipschitz}
    For every \(x,z \in \SR^p\),
    \[
        |F_\beta(x)  -F_\beta(z)| \leq \max_{1 \leq j \leq p}|x_j - z_j|.
    \]
\end{lemma}
\begin{lemma}[\cite{cck2013}, {Lemma A.4}]
    \label{lemma:composition-derivatives}
    Let \(\varphi(\cdot):\SR\to\SR\) be such that \(\varphi \in C_b^3(\SR)\) and define  \(m:\SR^p \to \SR\), \(z \mapsto \varphi(F_\beta(z))\). The derivatives (up to the third order) of \(m\) are given
    \begin{align*}
        \partial_j m(z) &= (\partial g(F(\beta))\pi_j)(z) \\ 
        \partial_j\partial_k m(z) 
                        &= (\partial^2 g(F_\beta)\pi_j\pi_k + \partial g(F_\beta)\beta w_{jk})(z)  \\ 
        \partial_j\partial_k\partial_l m(z)
                        &= (\partial^3 g(F_\beta)\pi_j\pi_k\pi_l 
                        +\partial^2 g(F_\beta)\beta(w_{jk}\pi_l + w_{jl}\pi_k + w_{kl}\pi_j) + \partial g(F_\beta)\beta^2 q_{jkl})(z)
    \end{align*}
    where \(\pi_j, w_{jk}, q_{jkl}\) are as described in \Cref{lemma:smooth-max-derivatives}.
\end{lemma}

\begin{lemma}[\cite{cck2013}, {Lemma A.5}]
    \label{lemma:composition-derivative-bounds}
    Define \(L_1(\varphi) = \sup_x |\varphi'(x)|, L_2(\varphi) = \sup_x |\varphi''(x)|\), and \(L_3(\varphi) = \sup_x |\varphi'''(x)|\).
    For every \(1 \leq j,k,l \leq p\), 
    \begin{align*}
        |\partial_j \partial_k m(z)| \leq U_{jk}(z) \andbox |\partial_j\partial_k\partial_l m(z)| \leq U_{jkl}(z)
    \end{align*}
    where for \(W_{jk}(z) := (\pi_j\delta_{jk} + \pi_j\pi_k)(z)\),
    \begin{align*}
        U_{jk}(z) &:= (L_2 \pi_j\pi_k + L_1\beta W_{jk}(z) \\ 
        U_{jkl}(z) &:= (L_3 \pi_j\pi_k\pi_l + L_2\beta(W_{jk}\pi_l + W_{jl}\pi_k + W_{kl}\pi_j) + L_1\beta^2 Q_{jkl})(z) \\ 
        Q_{jkl}(z) &:= (\pi_j\delta_{jl}\delta_{jk} + \pi_j\pi_k\delta_{jk} + \pi_j\pi_k(\delta_{jl} + \delta_{kl}) + 2\pi_j\pi_k\pi_l)(z).
    \end{align*}
    Moreover, 
    \begin{align*}
        \sum_{j,k=1}^p U_{jk}(z) \leq (L_2 + 2L_1\beta) \andbox
        \sum_{j,k,l=1}^p U_{jkl}(z) \leq (L_3 + 6L_2\beta + 6L_1 \beta^2).
    \end{align*}
\end{lemma}
% Lemma A.6 is only needed for the truncation part of the argument, which we don't employ here.
% \begin{lemma}[\citet{cck2013}, Lemma A.6]
%     \label{lemma:stability-properties}
%     For every \(z \in \SR^p\), \(w \in \SR^p\) with \(\max_{j \leq p}|w_j|\beta \leq 1\), \(\tau \in [0,1]\), and every \(1 \leq j,k,l \leq p\), we have
%     \[
%         U_{jk}(z) \lesssim U_{jk}(z + \tau w) \lesssim U_{jk}(z) \andbox U_{jkl}(z) \lesssim U_{jkl}(z+\tau w)\lesssim U_{jkl}(z)
%     \]
%     
% \end{lemma}
%
\subsection{Moment Bounds for \Cref{thm:multiple-feasible-local-power,thm:joint-combination}}
\label{subsec:multiple-moment-bounds}

\begin{lemma}[]
    \label{lemma:multiple-ND-bounds}
    Suppose that the moment conditions of \Cref{thm:multiple-feasible-local-power} hold and let \(N\) and \(D\) be as defined at the top of \Cref{subsec:joint-main} Then under \(H_0\), for any \(k\) there is a fixed constant \(C_k\) such that for any \(\ell = 1,\dots,d_x\)
    \[
        \E[|N_\ell|^k] \leq C_k \andbox \E[|D_{\ell\ell}|^k] \leq C_k\log^{2k/a}(n)
    \]
\end{lemma}
\begin{proof}
    Let \(\eta_{\ell i} = r_i - \E[r_i]\) and write
    \[
        N_\ell= \underbrace{\frac{1}{\sqrt{n}}\sum_{i=1}^n \eps_i(\beta_0)\sum_{j=1}^n \tilde h_{ij}\eta_{\ell j}}_{N_\ell^1} + \frac{1}{\sqrt{n}}\underbrace{\sum_{i=1}^n \eps_i(\beta_0)\sum_{j=1}^n \tilde h_{ij}\E[r_{\ell j}]}_{N_\ell^2} 
    \]
    To bound moments of \(N_\ell^1\) use the fact that \(N_\ell^1\) is a quadratic form in mean-zero \(a\)-sub-exponential variables. By \Cref{thm:hanson-wright}, \(N_\ell^1\) is therefore also \(a\)-sub-exponential with parameter \(a/2\); thus \((N_\ell^1)^{a/2}\) is sub-exponential and \Cref{lemma:sub-exponential-power-max-bound} provides the moment bound for arbitrary moments. To bound moments of \(N_\ell^2\) we use the fact that \(\max_i \big|\sum_{j=1}^n  \tilde h_{ij} \E[r_{\ell j}]\big|\) is bounded by assumption and apply Burkholder-Davis-Gundy (\Cref{thm:bdg}) after adding and subtracting \(\E[\eps_i(\beta_0)]\).

    To bound moments of \(D_{\ell\ell}\) we decompose
    \[
        |D| \leq \frac{1}{n}\sum_{i=1}^n \eps_i^2(\beta_0) \max_{1 \leq i \leq n}\big|\sum_{j=1}^n h_{ij}r_j\big|^2
    \]
    Apply \Cref{thm:hanson-wright} to see that \(\sum_{j=1}^n h_{ij}r_j\) is \(\alpha\)-sub-exponential and \Cref{lemma:sub-exponential-power-max-bound} to bound the RHS by a log-power of \(n\).
\end{proof}

\subsection{Matrix Derivative Lemmas}
\label{subsec:matrix-derivatives}

The purpose of this section is largely to establish some matrix derivative expressions that will be useful for the Lindeberg interpolation in 

\begin{lemma}[]
    \label{lemma:quadratic-form-derivatives}
    Let \(D \in \SR^{d\times d}\) be a symmetric, real matrix such that \(\det(D)\neq 0\). Let \(N \in \SR^d\) be a vector. The derivatives up to the derivatives of quadratic form \(N'D^{-1}N\) are given.

    First Order:
    \small
    \begin{align*}
        \frac{\partial }{\partial N_l} 
        &= 2\sum_{j=1}^d (D^{-1})_{jl} N_j,\hbox{ }\hbox{ }\hbox{ }\hbox{ }
        \frac{\partial }{\partial D_{l m}} 
        = -2\sum_{j=1}^d \sum_{k=1}^d (D^{-1})_{jl}(D^{-1})_{k m} N_jN_k,
        \intertext{Second Order:}
        \frac{\partial^2}{\partial N_l N_m} 
        &= 2(D^{-1})_{l m}, 
        \hbox{ }\hbox{ }\hbox{ }\hbox{ }\frac{\partial^2}{\partial N_l \partial D_{pq}} 
        = -2\sum_{j=1}^d (D^{-1})_{jp}(D^{-1})_{ql}N_j, \\ 
        \frac{\partial^2 }{\partial D_{l m}\partial D_{qj}} 
        &= \sum_{j=1}^d\sum_{k=1}^d\left\{ (D^{-1})_{l p}(D^{-1})_{qj})(D^{-1})_{km} + (D^{-1})_{kp}(D^{-1})_{mq}(D^{-1})_{l j}\right\}N_jN_k
        \intertext{Third Order:}
        \frac{\partial^3 }{\partial N_l \partial N_m \partial N_p} 
        &= 0,\hbox{ }\hbox{ }\hbox{ }\hbox{ } 
        \frac{\partial^3 }{\partial N_l \partial N_m \partial D_{pq}} =  
        -2(D^{-1})_{l p}(D^{-1})_{qm} \\ 
        \frac{\partial^3 }{\partial D_{l m}\partial D_{pq}\partial N_r} 
        &= 2 \sum_{j=1}^d \bigg\{(D^{-1})_{l p}(D^{-1})_{qj}(D^{-1})_{rm} + (D^{-1})_{rp}(D^{-1})_{mq}(D^{-1}){l j}\bigg\}N_j \\ 
        \frac{\partial^3 }{\partial D_{l m}D_{pq}D_{rs}} 
        &= 2\sum_{j=1}^d \sum_{j=1}^d \bigg\{(D^{-1})_{l r}(D^{-1})_{ps}(D^{-1})_{qj}(D^{-1})_{km} + (D^{-1})_{l p}(D^{-1})_{qr}(D^{-1})_{js}(D^{-1})_{km} \\ 
        &\hphantom{2\sum\sum\sum}
        +(D^{-1})_{l p}(D^{-1})_{qj}(D^{-1})_{kr}(D^{-1})_{ms} + (D^{-1})_{kr}(D^{-1})_{ps}(D^{-1})_{mq}(D^{-1})_{l j} \\ 
        &\hphantom{2\sum\sum\sum }
    +(D^{-1})_{kp}(D^{-1})_{mr}(D^{-1})_{qs}(D^{-1})_{l j} + (D^{-1})_{rp}(D^{-1})_{mq}(D^{-1})_{l r}(D^{-1}){js}\bigg\}N_jN_k
    \end{align*}
    \normalsize
\end{lemma}
\begin{proof}
    The derivative of an element of the the inverse of a matrix \(\mX\) can be expressed \citep{matrix-cookbook}
    \begin{equation}
        \label{eq:inverse-derivative}
        \frac{\partial(\mX^{-1})_{kl}}{\partial \mX_{ij}} = - (\mX^{-1})_{ki}(\mX^{-1})_{jl}
    \end{equation}
    repeated application of this identity as well as the expression of the quadratic form
    \[
        N'D^{-1}N = \sum_{j=1}^d\sum_{k=1}^d (D^{-1})_{jk}N_jN_k
    \]
    leads to the result, bearing in mind that the inverse of a symmetric matrix is symmetric.
\end{proof}
\begin{lemma}[]
    \label{lemma:determinant-derivatives}
    Let D be a symmetric positive definite matrix. Then, for any \(p > 3\), the derivatives of \((\det(D))^p\) are given up to the third order by 
    \begin{align*}
        \frac{\partial\,(\det(D))^p}{\partial D_{lm}} 
        &= p(\det(D))^{p-1}(D^{-1})_{lm} \\ 
        \frac{\partial^2\,(\det(D))^p }{\partial D_{lm}\partial D_{pq}} 
        &= \frac{p!}{(p-2)!}(\det(D))^{p-2}(D^{-1})_{pq}(D^{-1})_{lm} \\ 
        &\;\;+ p(\det(D))^{p-1}(D^{-1})_{lp}(D^{-1})_{mq} \\
        \frac{\partial^3\,(\det(D))^p }{\partial D_{lm}\partial D_{pq}\partial D_{rs}} 
        &= \frac{p!}{(p-3)!}(\det(D))^{p-3}(D^{-1})_{rs}(D^{-1})_{pq}(D^{-1})_{lm} \\
        &\;\; + \frac{p!}{(p-2)!}(\det(D))^{p-2}\bigg\{(D^{-1})_{pq}(D^{-1})_{lr}(D^{-1})_{ps} + (D^{-1})_{pr}(D^{-1})_{qs}(D^{-1})_{l m} \\ 
        &\hphantom{\;\;+\frac{p!}{(p-2)!}(\det(D))^{p-2}\bigg\{(D^{-1})_{pq}(D^{-1})_{lr}(D^{-1})_{ps}}\;\,+(D^{-1})_{rs}(D^{-1})_{lp}(D^{-1})_{mq}\bigg\}\\
        &\;\; + p(\det(D))^{p-1}\bigg\{(D^{-1})_{lr}(D^{-1})_{qs}(D^{-1})_{mq} + (D^{-1})_{lp}(D^{-1})_{mr}(D^{-1})_{q s}\bigg\}
    \end{align*}
\end{lemma}
\begin{proof}
    We can express the derivative of the detrminant \citep{matrix-cookbook}, 
    \begin{equation}
        \label{eq:determinant-derivative}
        \begin{split}
            \frac{\partial,\det(\mX)}{\partial\, \mX_{ij}}= \det(\mX)(\mX^{-1})_{ij}
        \end{split}
    \end{equation}
    Repeated application of this and \eqref{eq:inverse-derivative} yields the result.
\end{proof}
\begin{lemma}[]
    \label{lemma:total-derivatives}
    For any \(p > 4\) define the function \(\gamma(N,\vec(D)):\SR^d \times \SR^{d^2}\) by
    \[
        \gamma(N,\vec(D)) := \begin{cases}
         (\det(D))^p(N'D^{-1}N - c) &\text{if } \det(D) \neq 0 \\ 
     0  &\text{if }\det(D) = 0    \end{cases}
    \]
    This function is thrice continously differentiable. Futher the \(k^\text{th}\) moments of all partial derivatives of this function up to the third order are bounded
    \[
        \E[(\partial^\alpha \gamma(N,\vec(D))^k] \leq C_k (\max_{\iota \leq d}\E[|D_{\iota\iota}|^{2pdk}] \vee \max_{\iota \leq d}\E[|N_{\iota\iota}|^{6k}) 
    \]
    where \(C_k\) is a positive constant that only depends on \(k\) and \(d\).
\end{lemma}
\begin{proof}
    The first statement is clear by examination of the derivatives in \Cref{lemma:quadratic-form-derivatives,lemma:determinant-derivatives} as well as the inequality~\eqref{eq:inverse-element-bound} below.  For the moment bounds, we may extensive use of following bounds on elements of \(D^{-1}\) for a positive-definite \(D^{-1}\):
    \begin{equation}
        \label{eq:inverse-element-bound}
        \begin{split}
            |\det(D)(D^{-1})_{jk}| \leq \det(D)\trace(D^{-1}) 
            &\leq d\lambda_{\max}(D^{-1})\big(\prod_{m=1}^d \lambda_m(D)\big) \\ 
            &= d\prod_{m=2}^d \lambda_m(D) \\ 
            &\leq d\big(\sum_{m=2}^d \lambda_m(D)\big)^{d-1} \\ 
            &\leq d(\trace(D))^{d-1}
        \end{split}
    \end{equation}
    where the first inequality uses the fact that the largest element of a positive semidefinite matrix is on the diagonal and the fact that the diagonal elements of a positive semidefinite matrix are weakly positive, the second inequality uses the fact that the trace is the sum of the eigenvalues and the determinant is the product of the eigenvalues, the equality comes from \(\frac{1}{\lambda_{\min}(D)} = \lambda_{\max}(D^{-1})\), the third inequality uses the AM-GM inequality and the fourth again uses that the trace is the sum of the (weakly positive) eigenvalues.

     The moment bounds follow from \eqref{eq:inverse-element-bound} and the expressions in \Cref{lemma:quadratic-form-derivatives,lemma:determinant-derivatives}. We give an example of how this is done for the first order derivatives, higher order derivatives follow from similar logic. For the following let \(A\) be an arbitrary random variable.
    \textit{First Order.}
    \begin{align*}
        \E\bigg|A\frac{\partial \gamma}{\partial N_l}\bigg|^k
        &\lesssim \sum_{j=1}^d\E|(\trace(D))^{kdp}N_j^kA^k| \\ 
        &\lesssim \sum_{j=1}^d\sum_{\iota =1}^d\E[D_{\iota\iota}^{kdp}N_j^kA^k] \\ 
        &\leq \sum_{j=1}^d\sum_{\iota =1}^d \gamma^{2kdp} \E[N_j^{2k}A^{2k}]\\ 
        \E\bigg|A \frac{\partial \gamma}{\partial D_{lm}}\bigg|^k 
        &= p\E\bigg|A \det(D)^{p-1}\sum_{j=1}^d\sum_{j' = 1}^d (D^{-1})_{lm}(D^{-1})_{jj'}N_jN_{j'} \bigg|^k\\ 
        &\lesssim p\sum_{j=1}^d\sum_{j'=1}^d \E[|(\trace(D))^{2k(d-1) + (p-3)kd}A^kN_j^kN_{j'}^k| \\ 
        &\leq \sum_{j=1}^d \sum_{j'=1}^d \gamma^{2kd(p-1)}\E[A^{2k}N_j^{2k}N_{j'}^{2k}]
    \end{align*}
    
\end{proof}

\section{Technical Lemmas}
\label{sec:technical-lemmas}    

\subsection{Probability Lemmas}
\label{subsec:probability-lemmas}
% This first lemma I have seen before in the literature, but could not remember where. I have attached a proof of the result below. 
\begin{lemma}[]
    \label{lemma:op1-sequence}
    Let \(X_n\) be a sequence of random variables such that \(X_n = o_p(1)\), that is for any \(\delta > 0\), \(\Pr(|X_n| \geq \delta) \to 0\). Then, there is a sequence \(\delta_n \to 0\) such that \(\Pr(|X_n| \geq \delta_n) \to 0\).
\end{lemma}
\begin{proof}
    Take a preliminary sequence \(\tilde\delta_n \to 0\) and define
    \[
        \tilde n_j = \inf\{n : \Pr(|X_n| > \tilde\delta_j) < \tilde\delta_j\}
    \]
    Because \(\Pr(|X_n| >\delta) \to 0\) for any fixed \(\delta\), we know that \(n_j\) is finite. Define a new sequence \(\delta_n \to 0\) as below: 
    \begin{equation}
        \label{eq:to-zero}
        \begin{split}
            \delta_n = 
            \begin{cases}
            1 &\text{if }0 \leq n < \tilde n_1  \\
            \tilde\delta_i &\text{if }\tilde n_i \leq n < \tilde n_{i+1} 
            \end{cases}
        \end{split}
    \end{equation}
    By construction, this sequence satisfies \(\Pr(X_n \geq \delta_n) \leq \delta_n\) whenever \(n \geq n_1\).
\end{proof}
\begin{lemma}[]
    \label{lemma:sub-exponential-power-max-bound}
    Suppose that \(X_1,\dots,X_n\) are \(\alpha\)-subexponential such that \(\Pr(|X_i| \geq t) \leq 2\exp(-t^\alpha/K)\) for all \(t \geq 0\) and fixed constants \(K\). For any \(p \geq 1\) there is a constant \(C\) that depends only on \(p,K\) such that: 
    \[
        \E\bigg[\max_{i \leq n} \frac{|X_j|^p}{(1 + \log i)^{p/\alpha}}\bigg] \leq C
    \]
    As a consequence
    \[
        \E\big[\max_{i \leq n} |X_i|^p\big] \leq C(\log n)^{p/\alpha}
    \]
\end{lemma}
\begin{proof}
    Argument below is provided for \(\alpha = 1\). This can be extended to \(\alpha \neq 1\) by noting that if \(\Pr(|X_i| \geq t) \leq 2\exp(-t^\alpha /K)\) for some \(\alpha > 0\) then \(\Pr(|X_i|^{\alpha} \geq t) \leq 2\exp(-t/K)\).
    \begin{align*}
       \E\max_{i \leq n} \frac{|X_i|^p}{(1 + \log i)^p}
       &= \int_{0}^{\infty} \Pr\bigg(\max_i \frac{|X_i|^p}{(1 + \log i)^p} > t\bigg)\,dt \\  
       &= \int_{0}^{2^{p/\alpha}} \Pr\bigg(\max_i \frac{|X_i|^p}{(1 + \log i)^p} > t\bigg)\,dt + \int_{2^{p/\alpha}}^{\infty} \Pr\bigg(\max_i \frac{|X_i|^p}{(1 + \log i)^p} > t\bigg)\,dt \\ 
       &\leq 2^p + \int_{2^{p/\alpha}}^{\infty} \sum_{i=1}^n \Pr\bigg(\frac{|X_i|}{1 + \log i}> t^{1/p}\bigg)\,dt \\ 
       &\leq 2^p + \int_{2^p}^{\infty} \sum_{i=1}^n 2\exp\bigg(-\frac{t^{1/p}(1 + \log i)}{K}\bigg)\,dt \\ 
       &= 2^p + 2\sum_{i=1}^n \int_{2^p}^{\infty}  \exp\bigg(-\frac{t^{1/p}}{K}\bigg)i^{-t^{1/p}}\,dt \\ 
       &\leq 2^p+ 2\sum_{i=1}^n \int_{2^p}^{\infty}\exp(-t^{-1/p}/K)i^{-2}\,dt \\
       &\leq 2^p+ 2\bigg(\sum_{i=1}^n i^{-2}\bigg)\bigg( \int_{2^p}^{\infty}\exp(-t^{-1/p}/K)\,dt \bigg)
   \end{align*}
   Both the integral and the summation are bounded, which gives the result.
\end{proof}

\subsection{Matrix Lemmas}
\label{subsec:matrix-lemmas}

\begin{lemma}[]
    \label{lemma:eigenvalues}
    Given a matrix \(M\) and a matrix \(P\) of full rank, the matrix \(M\) and the matrix \(P^{-1}MP\) have the same eigenvalues.
\end{lemma}
% I think this is a known result, but if not should cite 
% https://www.cse.iitk.ac.in/users/rmittal/prev_course/s18/reports/7psdmatrices.pdf 
\begin{proof}
    Suppose \(\lambda\) is a eigenvalue of \(P^{-1}MP\) with eigenvector \(p\). Then
    \[
        P^{-1}MPv = \lambda v \implies M(Pv) = \lambda Pv
    \]
    Hence \(Pv\) is an eigenvector of \(M\) with eigenvalue \(\lambda\). Similarly, given an eigenvector \(v\) of \(M\), it can be shown that \(P^{-1}v\) is an eigenvector of \(P^{-1}MP\);
    \[
        P^{-1}MP(P^{-1}v) = P^{-1}Mv = \lambda P^{-1}v
    \]
\end{proof}

% This result is from the following math.stackexchange post
% https://math.stackexchange.com/questions/174737/are-there-any-bounds-on-the-eigenvalues-of-products-of-positive-semidefinite-mat
\begin{lemma}[]
    \label{lemma:product-eigenvalues}
    Let \(A \in \SR^{n\times n}\) and \(B\in \SR^{n\times n}\) be real symmetric positive semidefinite matrices. For an arbitary square matrix \(M\) let \(\lambda_k(M)\) denote the \(k^\text{\tiny th}\) largest eigenvalue of \(M\). Then  for any \(k = 1,\dots,n\):
    \[
        \lambda_k(A)\lambda_n(B) \leq \lambda_k(AB) \leq \lambda_k(A)\lambda_1(B)
    \]
\end{lemma}

\begin{lemma}[]
    \label{lemma:quadratic-form-eigenvalues}
    Let \(D \in \SR^{n \times n}\) be a diagonal real matrix such that \(d_{ii} \in [u, U]\) for all \(i =1 ,\dots, n\). Let \(A \in \SR^{n\times n}\) be a symmetric real matrix. For an arbitrary square matrix \(M\), let \(\lambda_k(M)\) denote the \(k^\text{\tiny th}\) largest eigenvalue of \(M\). Then for any \(k = 1,\dots,n\):
    \[
        u\lambda_k(A^2) \leq \lambda_k(ADA) \leq U \lambda_k(A^2)
    \]
\end{lemma}
\begin{proof}
    Consider any vector \(a \in \SR^n\) and define \(\va = a'H\). Then
    \begin{align*}
        \alpha'HDH\alpha 
        = \va'D\va =  \sum_{i=1}^n d_{ii}(\va_i)^2 
        &\in \bigg[u \sum_{i=1}^n (\va_i)^2,\, U \sum_{i=1}^n (\va_i)^2 \bigg] \\ 
        &= \bigg[u \times a'H^2a,\, U \times a'H^2a\bigg]
    \end{align*}
    The result then follows from an application of Courant-Fischer-Weyl min-max principle.
\end{proof}

% This result is from the following math.stackexchange post
% https://mathoverflow.net/questions/95096/anti-concentration-of-gaussian-quadratic-form
\begin{lemma}[]
    \label{lemma:anticoncentration-sum}
    Let \(X_1,\dots,X_n\) denote i.i.d standard normal random variables and \(a_1,\dots,a_n\) denote weakly positive constants. Then
    \[
        \Pr \left(\sum_{i=1}^n a_i X_i^2 \leq \eps \sum_{i=1}^n a_i\right) \leq \sqrt{e\eps}
    \]
\end{lemma}

\subsection{Miscellaneous Lemmas}
\label{subsec:misc-lemmas}

\begin{lemma}[]
    \label{lemma:ratio-bound}
    Let \(a_1,\dots,a_n\) and \(b_1,\dots,b_n\) be two sequences of real numbers. If \(a_i \leq Ub_i\) for some \(U > 0\), then \(\sum_i a_i/\sum_i b_i \leq U\). Conversely if \(a_i \geq L b_i\) for some \(L > 0\) then \(\sum_i a_i/\sum_i b_i \geq L\).
\end{lemma}
\begin{proof}
    Replace \(a_i \leq Ub_i\) for the upper bound and \(a_i \geq Lb_i\) for the lower bound.
\end{proof}

The following is a standard bound, but it is used a lot so it is restated here.
\begin{lemma}[]
    \label{lemma:sum-bound}
    Let \(a_1,\dots,a_m\) be constants and \(p > 1\). Then 
    \[
        |a_1 + \dots a_m|^p \leq m^{p-1}\sum_{i=1}^m |a_i|^p
    \]
\end{lemma}
\begin{proof}
    Apply Hölder's inequality with \(\frac{1}{p} + \frac{p-1}{p} = 1\) to the vectors \((a_1,\dots,a_m)\in \SR^m\) and \((1,\dots,1)\in \SR^m\)
\end{proof} 

\section{Assorted Results from Literature}
\label{sec:literature-results}

\subsection{Concentration Inequalities and Tail Bounds}
\begin{theorem}[\cite{Gotze-subexponential-polynomial}*{Theorem 1.2}]
    \label{thm:hanson-wright}
    Let \(X_1,\dots,X_n\) be independent random variables satisfying \(\|X_i\|_{\Psi_a}\leq M\) for some \(a \in (0,1] \cup \{2\}\) and let \(f:\SR^n\to\SR\) be a polynomial of total degree \(D \in \SN\). Then for all \(t > 0\);
    \[
        \Pr(|f(X) - \E[f(X)]| \geq t) \leq 2\exp\bigg(- \frac{1}{C_{D,a}}\min_{1 \leq d \leq D}\bigg(\frac{t}{M^d\|\E f^{(d)}(X)\|_{\text{HS}}}\bigg)^{a/d}\bigg)
    \]
    In particular, if \(\|\E f^{(d)}(X)\|_{\text{HS}} \leq 1\) for \(d = 1,\dots D\), then 
    \[
        \E\exp\bigg(\frac{C_{D,a}}{M^a}|f(X)|^{\frac{a}{D}}\bigg) \leq 2,
    \]
    or equivalently
    \[
        \|f(X)\|_{\Psi_{\frac{a}{D}}} \leq C_{d,a}M^D
    \]
\end{theorem}

\begin{theorem}[Hoeffding's Inequality]
    \label{thm:hoeffding}
    Let \(X_1,\dots,X_n\) be independent, mean-zero sub-gaussian random variables, and let \(a = (a_1,\dots,a_n)\in \SR^n\). Then, for every \(t \geq 0\), we have
    \[
        \Pr\bigg\{\bigg|\sum_{i=1}^n a_iX_i\bigg|\geq t\bigg\} \leq 2\exp\bigg(-\frac{ct^2}{K^2\|a\|_2^2}\bigg)
    \]
    where \(K = \max_i \|X_i\|_{\psi_2}\).
\end{theorem}

\begin{theorem}[Burkholder-Davis-Gurdy for Discrete Time Martingales]
    \label{thm:bdg}
    For any \(1 \leq k < \infty\)  there exist positive constants \(c_k\) and \(C_k\) such that for all local martingales with \(X_0 = 0\) and stopping times \(\tau\)
    \[
        c_k \E\big[\big(\sum_{t=1}^\tau (X_t - X_{t-1})^2\big)^{k/2}\big] \leq \E\big[(\sup_{t \leq \tau} X_t)^k\big] \leq C_k\E\big[\big(\sum_{t=1}^\tau (X_t - X_{t-1})^2\big)^{k/2}\big]
    \]
\end{theorem}
%
% For the below define
%
% \begin{theorem}[Hanson-Wright for \(\alpha\)-sub-exponential r.vs]
%     \label{thm:hanson-wright-sub-exponential}
%     Let \(X_1,\dots,X_n\) be independent random variables satisfying \(\|X_i\|_{\Psi_a} \leq M\) for some \(a \in (0,1]\) and let \(f:\SR^n \to \SR\) be a polynomial of total degree \(D\). Then, for all \(t > 0\)
%     \[
%         \Pr\big(|f(X) - \E[f(X)]| \geq t\big) \leq 2\exp\bigg(-\frac{1}{C_{D,\alpha}}\min_{1 \leq d \leq D}\bigg(\frac{t}{M^d\|\E[f^{(d)}(X)]\|_{\text{HS}}}\bigg)^{\alpha/d}\bigg)
%     \]
% \end{theorem}
\subsection{Anticoncentration Bounds}
Let \(\xi \in \SR^n\) follow a normal distribution on \(\SR^n\)  with mean zero and covariance matrix \(\Sigma_\xi\). Order the eigenvalues of \(\Sigma_\xi\) in non-increasing order \(\lambda_{1\xi} \geq \lambda_{2\xi} \geq ... \geq \lambda_{n\xi}\). Define the quantities
\[
    \Lambda_{k\xi}^2 = \sum_{j=k}^\infty \lambda_{j\xi}^2,\;\; k =1,2
\]
\begin{theorem}[\cite{anticoncentration-bernoulli-gotze-et-al}, {Theorem 2.6}]
    \label{thm:quadratic-density-bound}
    Let \(\xi\) be a gaussian element with zero mean and covariance \(\Sigma_\xi\). Then it holds for any \(\va \in \SR^n\) that
    \[
        \sup_{x \geq 0} p_\xi(x,\va) \lesssim (\Lambda_{1\xi}\Lambda_{2\xi})^{-1/2}
    \]
    where \(p_\xi(x,a)\) denotes the p.df of \(\|\xi - \va\|^2\).     
\end{theorem}

We use the following anticoncentration lemma from \citet{Nazarov2003} noted in \citet{CCK-2018-HDCLT}.
\begin{lemma}
    \label{lemma:hyper-rectangle-anticoncentration}
    Let \(Y = (Y_1,\dots,Y_p)'\) be a centered Gaussian random vector in \(\SR^p\) such that \(\E[Y_j^2] \geq b\) for all \(j = 1,\dots,p\) and some constant \(b > 0\). Then for every \(y \in \SR^p\) and \(a > 0\),
    \[
        \Pr(Y \leq y + a) - \Pr(Y \leq y) \leq Ca\sqrt{\log(p)}
    \]
    where \(C\) is a constant only depending on \(b\).
\end{lemma}

\subsection{Gaussian Comparasions and Approximations}

We also use the following gaussian approximation results from \cites{belloni2018highdimensional,CCK-2018-HDCLT}. Let \(X_1,\dots,X_n \in \SR^p\) be independent, mean zero, random vectors and let \(Y_1,\dots,Y_n \in \SR^p\) be independent random vectors such that \(Y_i \sim N(0, \E[X_iX_i'])\). Suppose that the researcher does not directly observe \(X_1,\dots,X_n\) but instead observes noisy estimates \(\widehat X_1,\dots,\widehat X_n \in \SR^p\).

Define the sums
\begin{align*}
    S_n^X &= \frac{1}{\sqrt{n}}\sum_{i=1}^n \widehat X_i & 
    S_n^Y &= \frac{1}{\sqrt{n}}\sum_{i=1}^n Y_i
\end{align*}
Let \(\calA^{\text{re}}\) be the class of all hyperrectangles in \(\SR^p\); that is, \(\calA^{\text{re}}\) consists of all sets \(A\) of the form 
\begin{equation*}
    \begin{split}
        A = \{w\in\SR^p: a_j \leq w_j \leq b_j\text{ for all }j = 1,\dots,p\}
    \end{split}
\end{equation*}
for some \(-\infty \leq a_j \leq b_j \leq \infty\), \(j = 1,\dots, p\). Define
\[
    \rho_n(\calA^{\text{re}}) \coloneqq \sup_{A \in \calA^{\text{re}}}\big|\Pr(S_n^X \in A) - \Pr(S_n^Y \in A)\big|
\]
Bounding \(\rho_n(\calA^{\text{re}})\) relies on the following moment conditions:
\begin{assumption}[]
    \label{assm:rectangle-gaussian-approximation}
    Suppose there are constants \(B_n \geq 1\), \(b > 0,q > 0\) such that
    \begin{enumerate}[(i)]
        \item \(n^{-1}\sum_{i=1}^n \E[X_{ij}^2] \geq b\) for all \(j = 1,\dots,p\)
        \item \(n^{-1}\sum_{i=1}^n \E[|X_{ij}|^{2+k}] \leq B_n^k\) for all \(j = 1,\dots,p\) and \(k = 1,2\).
        \item \(\E[(\max_{1 \leq j \leq p}|X_{ij}|/B_n)^4] \leq 1\) for all \(i = 1,\dots, n\) and \(\left(\frac{B_n^4\ln^7(pn)}{n}\right)^{1/6} \leq \delta_n\).
    \end{enumerate}
\end{assumption}
as well as the following bounds on the estimation error
\begin{assumption}[]
    \label{assm:rectangle-gaussian-approximation-estimation}
    The estimates \(\hat X_1,\dots,\hat X_n\) satisfy
    \[
        \Pr\left(\max_{1 \leq j \leq p} \E_n[(\widehat X_{ij} - X_{ij})^2] > \delta_n^2/\log^2(pn)\right)\leq \beta_n
    \]
\end{assumption}

\begin{theorem}[\citet{belloni2018highdimensional}, {Theorem 2.1}]
    \label{prop:hyperrectangle}
    Suppose that \Cref{assm:rectangle-gaussian-approximation,assm:rectangle-gaussian-approximation-estimation} hold. Then there is a constant \(C\) which depends only on \(b\) such that
    \[
        \rho_n(\calA^{\text{re}}) \leq C\{\delta_n + \beta_n\}
    \]
\end{theorem}

Let \(e_1,\dots,e_n \overset{\text{iid}}{\sim}N(0,1)\) be generated independently of the data. A gaussian bootstrap draw is defined
\[
    S_n^{X,\star} \coloneqq  \frac{1}{\sqrt{n}}\sum_{i=1}^n e_i \widehat X_i
\]
\begin{theorem}[\citet{belloni2018highdimensional}, {Theorem 2.2}]
    \label{thm:hyperrectangle-bootstrap}
    Suppose that \Cref{assm:rectangle-gaussian-approximation,assm:rectangle-gaussian-approximation-estimation} hold. Then there is a constant \(C\) which depends only on \(b\) such that
    \[
        \sup_{A \in \calA^{\text{re}}} \big|\Pr_e(S_n^{X,\star} \in A) - \Pr(S_n^Y \in A) \big| \leq C\delta_n
    \]
    with probability at least \(1 - \beta_n - (\log n)^{-2}\) where \(\Pr_e(\cdot)\) denotes the probability measure only taken with respect to the variables \(e_1,\dots,e_n\) conditional on the data used to estimate \(\widehat X\).
\end{theorem}

% \newpage
% \section{Additional Tables from Simulation Study}
% \label{sec:simulation-additional}
% \input{misc/simulation_tables_total.tex}
%

\end{document}